\documentclass[11pt]{article}%
\usepackage{amsmath}
\usepackage{hyperref}
\usepackage{amsfonts}
\usepackage{amssymb}
\usepackage{graphicx}
\usepackage{fullpage}
\usepackage{pgf}
\usepackage{tikz}
\usepackage{xy}
\usepackage{ifpdf}
\usepackage{siunitx}%
\providecommand{\U}[1]{\protect\rule{.1in}{.1in}}
\usetikzlibrary{calc,positioning,fit}
\usetikzlibrary{fit,svg.path,positioning}
\newtheorem{theorem}{Theorem}

\newtheorem{corollary}[theorem]{Corollary}

\newtheorem{lemma}[theorem]{Lemma}

\newtheorem{proposition}[theorem]{Proposition}

\newenvironment{proof}[1][Proof]{\noindent\textbf{#1.} }{\ \rule{0.5em}{0.5em}}
\newcommand{\T}[1]{T_{#1}}
\newcommand{\F}[1]{F_{#1}}
\def\Fredkin{\mathsf{Fredkin}}
\def\NOT{\mathsf{NOT}}
\def\NOTNOT{\mathsf{NOTNOT}}
\def\CNOT{\mathsf{CNOT}}
\def\CNOTNOT{\mathsf{CNOTNOT}}
\newcommand{\gateMOD}[1]{\mathsf{MOD#1}}
\xyoption{matrix}
\xyoption{frame}
\xyoption{arrow}
\xyoption{arc}
\ifpdf
\else
\PackageWarningNoLine{Qcircuit}{Qcircuit is loading in Postscript mode.  The Xy-pic options ps and dvips will be loaded.  If you wish to use other Postscript drivers for Xy-pic, you must modify the code in Qcircuit.tex}
\xyoption{ps}
\xyoption{dvips}
\fi
\entrymodifiers={!C\entrybox}

\newcommand{\qw}[1][-1]{\ar @{-} [0,#1]}
\newcommand{\qwx}[1][-1]{\ar @{-} [#1,0]}

\newcommand{\gate}[1]{*+<.6em>{#1} \POS ="i","i"+UR;"i"+UL **\dir{-};"i"+DL **\dir{-};"i"+DR **\dir{-};"i"+UR **\dir{-},"i" \qw}

\newcommand{\control}{*!<0em,.025em>-=-<.2em>{\bullet}}

\newcommand{\ctrl}[1]{\control \qwx[#1] \qw}

\newcommand{\targ}{*+<.02em,.02em>{\xy ="i","i"-<.39em,0em>;"i"+<.39em,0em> **\dir{-}, "i"-<0em,.39em>;"i"+<0em,.39em> **\dir{-},"i"*\xycircle<.4em>{} \endxy} \qw}
\newcommand{\qswap}{*=<0em>{\times} \qw}

\newcommand{\Qcircuit}{\xymatrix @*=<0em>}

\begin{document}

\title{The Classification of Reversible Bit Operations}
\author{Scott Aaronson\thanks{MIT. \ Email: aaronson@csail.mit.edu. \ Supported by an
Alan T.\ Waterman Award from the National Science Foundation, under grant
no.\ 1249349.}
\and Daniel Grier\thanks{MIT. \ Email: grierd@mit.edu. \ Supported by an NSF
Graduate Research Fellowship under Grant No. 1122374.}
\and Luke Schaeffer\thanks{MIT. \ Email: lrs@mit.edu.}}
\date{}
\maketitle

\begin{abstract}
We present a complete classification of all possible sets of classical
reversible gates acting on bits, in terms of which reversible transformations
they generate, assuming swaps and ancilla bits are available for free. \ Our
classification can be seen as the reversible-computing analogue of
\textit{Post's lattice}, a central result in mathematical logic from the
1940s. \ It is a step toward the ambitious goal of classifying all possible
quantum gate sets acting on qubits.

Our theorem implies a linear-time algorithm (which we have implemented), that
takes as input the truth tables of reversible gates $G$ and $H$, and that
decides whether $G$ generates $H$. \ Previously, this problem was not even
known to be decidable (though with effort, one can derive from abstract
considerations an algorithm that takes triply-exponential time). \ The theorem
also implies that any $n$-bit reversible circuit can be \textquotedblleft
compressed\textquotedblright\ to an equivalent circuit, over the same gates,
that uses at most\ $2^{n}\operatorname*{poly}\left(  n\right)  $\ gates and
$O(1)$\ ancilla bits; these are the first upper bounds on these quantities
known, and are close to optimal. \ Finally, the theorem implies that every
non-degenerate reversible gate can implement either every reversible
transformation, or every affine transformation, when restricted to an
\textquotedblleft encoded subspace.\textquotedblright

Briefly, the theorem says that every set of reversible gates generates either
all reversible transformations on $n$-bit strings (as the Toffoli gate does);
no transformations; all transformations that preserve Hamming weight (as the
Fredkin gate does); all transformations that preserve Hamming weight mod $k$
for some $k$; all affine transformations (as the Controlled-NOT gate does);
all affine transformations that preserve Hamming weight mod $2$ or mod $4$,
inner products mod $2$, or a combination thereof; or a previous class
augmented by a NOT or NOTNOT gate. \ Prior to this work, it was not even known
that every class was finitely generated. \ Ruling out the possibility of
additional classes, not in the list, requires some arguments about
polynomials, lattices, and Diophantine equations.

\end{abstract}
\tableofcontents

\section{Introduction\label{INTRO}}

The \textit{pervasiveness of universality}---that is, the likelihood that a
small number of simple operations already generate all operations in some
relevant class---is one of the central phenomena in computer science. \ It
appears, among other places, in the ability of simple logic gates to generate
all Boolean functions (and of simple quantum gates to generate all unitary
transformations); and in the simplicity of the rule sets that lead to
Turing-universality, or to formal systems to which G\"{o}del's theorems apply.
\ Yet precisely because universality is so pervasive, it is often more
interesting to understand the ways in which systems can \textit{fail} to be universal.

In 1941, the great logician Emil Post \cite{post}\ published a complete
classification of all the ways in which sets of Boolean logic gates can fail
to be universal: for example, by being monotone (like the $\operatorname{AND}$
and $\operatorname*{OR}$\ gates) or by being affine over $\mathbb{F}_{2}%
$\ (like $\operatorname{NOT}$\ and $\operatorname{XOR}$). \ In universal
algebra, closed classes of functions are known, somewhat opaquely, as
\textit{clones}, while the inclusion diagram of all Boolean clones is called
\textit{Post's lattice}. \ Post's lattice is surprisingly complicated, in part
because Post did not assume that the constant functions $0$ and $1$ were
available for free.\footnote{In Appendix \ref{POST}, we prove for completeness
that if one \textit{does} assume constants are free, then Post's lattice
dramatically simplifies, with all non-universal gate sets either monotone or
affine.}

This paper had its origin in our ambition to find the analogue of Post's
lattice for all possible sets of \textit{quantum} gates acting on qubits. \ We
view this as a large, important, and underappreciated goal: something that
could be to quantum computing theory almost what the Classification of Finite
Simple Groups was to group theory. \ To provide some context, there are many
finite sets of $1$-, $2$- and $3$-qubit quantum gates that are known to be
universal---either in the strong sense that they can be used to approximate
any $n$-qubit unitary transformation to any desired precision, or in the
weaker sense that they suffice to perform universal quantum computation
(possibly in an encoded subspace). \ To take two examples, Barenco et
al.\ \cite{barenco}\ showed universality for the $\operatorname*{CNOT}$\ gate
plus the set of all $1$-qubit gates, while Shi \cite{shi:gate}\ showed
universality for the $\operatorname*{Toffoli}$\ and $\operatorname*{Hadamard}$\ gates.

There are also sets of quantum gates that are known \textit{not} to be
universal: for example, the basis-preserving gates, the $1$-qubit gates, and
most interestingly, the so-called \textit{stabilizer gates}%
\ \cite{gottesman:hamm,ag}\ (that is, the $\operatorname*{CNOT}$,
$\operatorname*{Hadamard}$, and $\pi/4$-$\operatorname*{Phase}$ gates), as
well as the stabilizer gates conjugated by $1$-qubit unitary transformations.
\ What is \textit{not} known is whether the preceding list basically exhausts
the ways in which quantum gates on qubits can fail to be universal. \ Are
there other elegant discrete structures, analogous to the stabilizer gates,
waiting to be discovered? \ Are there any gate sets, other than conjugated
stabilizer gates, that might give rise to intermediate complexity classes,
neither contained in $\mathsf{P}$\ nor equal to $\mathsf{BQP}$?\footnote{To
clarify, there are many restricted models of quantum computing known that are
plausibly \textquotedblleft intermediate\textquotedblright\ in that sense,
including BosonSampling \cite{aark}, the one-clean-qubit model
\cite{knilaflamme}, and log-depth quantum circuits \cite{cw}. \ However, with
the exception of conjugated stabilizer gates, none of those models arises from
simply considering which unitary transformations can be generated by some set
of $k$-qubit gates. \ They all involve non-standard initial states, building
blocks other than qubits, or restrictions on how the gates can be composed.}
\ How can we claim to understand quantum circuits---the bread-and-butter of
quantum computing textbooks and introductory quantum computing courses---if we
do not know the answers to such questions?

Unfortunately, working out the full \textquotedblleft quantum Post's
lattice\textquotedblright\ appears out of reach at present. \ This might
surprise readers, given how much is known about particular quantum gate sets
(e.g., those containing $\operatorname*{CNOT}$\ gates), but keep in mind that
what is asked for is an accounting of \textit{all} possibilities, no matter
how exotic. \ Indeed, even classifying $1$- and $2$-qubit quantum gate sets
remains wide open (!), and seems, without a new idea, to require studying the
irreducible representations of thousands of groups. \ Recently, Aaronson and
Bouland \cite{aarbouland}\ completed a much simpler task, the classification
of $2$-mode beamsplitters; that was already a complicated undertaking.

\subsection{Classical Reversible Gates\label{REVINTRO}}

So one might wonder:\ can we at least understand all the possible sets of
\textit{classical reversible gates} acting on bits, in terms of which
reversible transformations they generate? \ This an obvious prerequisite to
the quantum case, since every classical reversible gate is also a unitary
quantum gate. \ But beyond that, the classical problem is extremely
interesting in its own right, with (as it turns out) a rich algebraic and
number-theoretic structure, and with many implications for reversible
computing as a whole.

The notion of reversible computing
\cite{fredkintoffoli,toffoli,landauer,bennett,lloyd:gate,saeedi}\ arose from
early work on the physics of computation, by such figures as Feynman, Bennett,
Benioff, Landauer, Fredkin, Toffoli, and Lloyd. \ This community was
interested in questions like: does universal computation inherently require
the generation of entropy (say, in the form of waste heat)? \ Surprisingly,
the theory of reversible computing showed that, in principle, the answer to
this question is \textquotedblleft no.\textquotedblright\ \ \textit{Deleting}
information unavoidably generates entropy, according to \textit{Landauer's
principle} \cite{landauer}, but deleting information is not necessary for
universal computation.

Formally, a reversible gate\ is just a permutation $G:\left\{  0,1\right\}
^{k}\rightarrow\left\{  0,1\right\}  ^{k}$\ of the set of $k$-bit strings, for
some positive integer $k$. \ The most famous examples are:

\begin{itemize}
\item the $2$-bit $\operatorname*{CNOT}$\ (Controlled-NOT) gate, which flips
the second bit if and only if the first bit is $1$;

\item the $3$-bit $\operatorname*{Toffoli}$\ gate, which flips the third bit
if and only if the first two bits are both $1$;

\item the $3$-bit $\operatorname*{Fredkin}$\ gate, which swaps the second and
third bits if and only if the first bit is $1$.
\end{itemize}

These three gates already illustrate some of the concepts that play important
roles in this paper. \ The $\operatorname*{CNOT}$\ gate can be used to copy
information in a reversible way, since it maps $x0$\ to $xx$; and also to
compute arbitrary affine functions over the finite field $\mathbb{F}_{2}$.
\ However, because $\operatorname*{CNOT}$\ is \textit{limited} to affine
transformations, it is not computationally universal. \ Indeed, in contrast to
the situation with irreversible logic gates, one can show that \textit{no}
$2$-bit classical reversible gate is computationally universal. \ The
$\operatorname*{Toffoli}$ gate is computationally universal, because (for
example) it maps $x,y,1$\ to $x,y,\overline{xy}$, thereby computing the
$\operatorname*{NAND}$\ function. \ Moreover, Toffoli showed \cite{toffoli}%
---and we prove for completeness in Section \ref{NONAFFCIRC}---that the
$\operatorname*{Toffoli}$\ gate is universal in a stronger sense: it generates
all possible reversible transformations $F:\left\{  0,1\right\}
^{n}\rightarrow\left\{  0,1\right\}  ^{n}$ if one allows the use of ancilla
bits, which must be returned to their initial states by the end.

But perhaps the most interesting case is that of the $\operatorname*{Fredkin}%
$\ gate. \ Like the $\operatorname*{Toffoli}$\ gate, the
$\operatorname*{Fredkin}$\ gate is computationally universal: for example, it
maps $x,y,0$\ to $x,\overline{x}y,xy$, thereby computing the
$\operatorname*{AND}$\ function. \ But the $\operatorname*{Fredkin}$\ gate is
\textit{not} universal in the stronger sense.\ \ The reason is that it is
\textit{conservative}: that is, it never changes the total Hamming weight of
the input. \ Far from being just a technical issue, conservativity was
regarded by Fredkin and the other reversible computing pioneers as a sort of
discrete analogue of the conservation of energy---and indeed, it plays a
central role in certain physical realizations of reversible computing (for
example, billiard-ball models, in which the total number of billiard balls
must be conserved).

However, all we have seen so far are three specific examples of reversible
gates, each leading to a different behavior. \ To anyone with a mathematical
mindset, the question remains: what are all the \textit{possible} behaviors?
\ For example: is Hamming weight the only possible \textquotedblleft conserved
quantity\textquotedblright\ in reversible computation? \ Are there other ways,
besides being affine, to fail to be computationally universal? \ Can one
\textit{derive}, from first principles, why the classes of reversible
transformations generated by $\operatorname*{CNOT}$, $\operatorname*{Fredkin}%
$, etc.\ are somehow special, rather than just pointing to the sociological
fact that these are classes that people in the early 1980s happened to study?

\subsection{Ground Rules\label{RULES}}

In this work, we achieve a complete classification of all possible sets of
reversible gates acting on bits, in terms of which reversible transformations
$F:\left\{  0,1\right\}  ^{n}\rightarrow\left\{  0,1\right\}  ^{n}$\ they
generate. \ Before describing our result, let us carefully explain the ground rules.

First, we assume that swapping bits is free. \ This simply means that we do
not care how the input bits are labeled---or, if we imagine the bits carried
by wires, then we can permute the wires in any way we like. \ The second rule
is that an unlimited number of ancilla bits may be used, \textit{provided }the
ancilla bits are returned to their initial states by the end of the
computation. \ This second rule might look unfamiliar, but in the context of
reversible computing, it is the right choice.

We need to allow ancilla bits because if we do not, then countless
transformations are disallowed for trivial reasons. \ (Restricting a
reversible circuit to use \textit{no} ancillas is like restricting a Turing
machine to use no memory, besides the $n$ bits that are used to write down the
input.) \ We are forced to say that, although our gates might generate some
reversible transformation $F\left(  x,0\right)  =\left(  G\left(  x\right)
,0\right)  $, they do not generate the smaller transformation $G$. \ The exact
value of $n$ then also takes on undeserved importance, as we need to worry
about \textquotedblleft small-$n$\ effects\textquotedblright: e.g., that a
$3$-bit gate cannot be applied to a $2$-bit input.

As for the number of ancilla bits: it will \textit{turn out}, because of our
classification theorem, that every reversible gate needs only $O(1)$\ ancilla
bits\footnote{Since it is easy to show that a constant number of ancilla bits
are sometimes needed (see Proposition \ref{needancilla}), this is the optimal
answer, up to the value of the constant (which might depend on the gate set).}
to generate every $n$-bit reversible transformation that it can generate at
all. \ However, we do not wish to prejudge this question; if there had been
reversible gates that could generate certain transformations, but only by
using (say) $2^{2^{n}}$\ ancilla bits, then that would have been fascinating
to know. \ For the same reason, we do not wish prematurely to restrict the
number of ancilla bits that can be $0$, or the number that can be $1$.

On the other hand, the ancilla bits must be returned to their original states
because if they are not, then the computation was not really reversible. \ One
can then learn something about the computation by examining the ancilla
bits---if nothing else, then the fact that the computation was done at all.
\ The symmetry between input and output is broken; one cannot then run the
computation backwards without setting the ancilla bits differently. \ This is
not just a philosophical problem: if the ancilla bits carry away information
about the input $x$, then \textit{entropy}, or waste heat, has been leaked
into the computer's environment. \ Worse yet, if the reversible computation is
a subroutine of a quantum computation, then the leaked entropy will cause
\textit{decoherence}, preventing the branches of the quantum superposition
with different $x$\ values from interfering with each other, as is needed to
obtain a quantum speedup. \ In reversible computing, the technical term for
ancilla bits that still depend on $x$ after a computation is complete is
\textit{garbage}.\footnote{In Section \ref{ALTGEN}\ and Appendix \ref{LOOSE},
we will discuss a modified rule, which allows a reversible circuit to change
the ancilla bits, as long as they change in a way that is independent of the
input $x$. \ We will show that this \textquotedblleft loose ancilla
rule\textquotedblright\ causes only a small change to our classification
theorem.}

\subsection{Our Results\label{RESULTS}}

Even after we assume that bit swaps and ancilla bits are free, it remains a
significant undertaking to work out the complete list of reversible gate
classes, and (especially!) to prove that the list is complete. \ Doing so is
this paper's main technical contribution.

We give a formal statement of the classification theorem in Section
\ref{THEOREM}, and we show the lattice of reversible gate classes in Figure
\ref{lattice}. \ (In Appendix \ref{NUMBER},\ we also calculate the exact
number of $3$-bit gates that generate each class.) \ For now, let us simply
state the main conclusions informally.

\begin{itemize}
\item[(1)] \textbf{Conserved Quantities.} \ The following is the complete list
of the \textquotedblleft global quantities\textquotedblright\ that reversible
gate sets can conserve (if we restrict attention to non-degenerate gate sets,
and ignore certain complications caused by linearity and affineness): Hamming
weight, Hamming weight mod $k$ for any $k\geq2$, and inner product mod $2$
between pairs of inputs.

\item[(2)] \textbf{Anti-Conservation.} \ There are gates, such as the
$\operatorname*{NOT}$\ gate, that \textquotedblleft
anti-conserve\textquotedblright\ the Hamming weight mod $2$ (i.e., always
change it by a fixed nonzero amount). \ However, there are no analogues of
these for any of the other conserved quantities.

\item[(3)] \textbf{Encoded Universality.} \ In terms of their
\textquotedblleft computational power,\textquotedblright\ there are only three
kinds of reversible gate sets: degenerate (e.g., $\operatorname{NOT}$s,
bit-swaps),\ non-degenerate but affine (e.g., $\operatorname*{CNOT}$), and
non-affine (e.g., $\operatorname*{Toffoli}$, $\operatorname*{Fredkin}$).
\ More interestingly, every non-affine gate set can implement every reversible
transformation, and every non-degenerate affine gate set can implement every
affine transformation, \textit{if} the input and output bits are encoded by
longer strings in a suitable way. \ For details about \textquotedblleft
encoded universality,\textquotedblright\ see Section \ref{ENCODED}.

\item[(4)] \textbf{Sporadic Gate Sets.} \ The conserved quantities interact
with linearity and affineness in complicated ways, producing \textquotedblleft
sporadic\textquotedblright\ affine gate sets that we have classified. \ For
example, non-degenerate affine gates can preserve Hamming weight mod $k$, but
only if $k=2$ or $k=4$. \ All gates that preserve inner product mod $2$ are
linear, and all linear gates that preserve Hamming weight mod $4$ also
preserve inner product mod $2$. \ As a further complication, affine gates can
be orthogonal or mod-$2$-preserving or mod-$4$-preserving in their linear
part, but not in their affine part.

\item[(5)] \textbf{Finite Generation.} \ For each closed class of reversible
transformations, there is a single gate that generates the entire class.
\ (\textit{A priori}, it is not even obvious that every class\ is finitely
generated, or that there is \textquotedblleft only\textquotedblright\ a
countable infinity of classes!) \ For more, see Section \ref{NATURE}.

\item[(6)] \textbf{Symmetry.} \ Every reversible gate set is symmetric under
interchanging the roles of $0$ and $1$. \ For more, see Section \ref{NATURE}.
\end{itemize}

\subsection{Algorithmic and Complexity Aspects\label{COMPLEXITY}}

Perhaps most relevant to theoretical computer scientists, our classification
theorem leads to new algorithms and complexity results about reversible gates
and circuits: results that follow easily from the classification, but that we
have no idea how to prove otherwise.

Let \textsc{RevGen} (Reversible Generation)\ be the following problem: we are
given as input the truth tables of reversible gates $G_{1},\ldots,G_{K}$, as
well as of a target gate $H$, and wish to decide whether the $G_{i}$'s
generate $H$. \ Then we obtain a linear-time algorithm for \textsc{RevGen}.
\ Here, of course, \textquotedblleft linear\textquotedblright\ means linear in
the sizes of the truth tables, which is $n2^{n}$ for an $n$-bit gate.
\ However, if just a tiny amount of \textquotedblleft summary
data\textquotedblright\ about each gate $G$ is provided---namely, the possible
values of $\left\vert G\left(  x\right)  \right\vert -\left\vert x\right\vert
$, where $\left\vert \cdot\right\vert $\ is the Hamming weight, as well as
which affine transformation $G$ performs if it is affine---then the algorithm
actually runs in $O\left(  n^{\omega}\right)  $ time, where $\omega$\ is the
matrix multiplication exponent.

We have implemented this algorithm;\ code is available for download at
\cite{schaefer:code}. \ For more details see Section \ref{ALGSEC}.

Our classification theorem also implies the first general upper bounds (i.e.,
bounds that hold for all possible gate sets) on the number of gates and
ancilla bits needed to implement reversible transformations. \ In particular,
we show (see Section \ref{COMPRESS}) that if a set of reversible gates
generates an $n$-bit transformation $F$ at all, then it does so via a circuit
with at most $2^{n}\operatorname*{poly}\left(  n\right)  $\ gates and
$O(1)$\ ancilla bits.\ \ These bounds are close to optimal.

By contrast, let us consider the situation for these problems without the
classification theorem. \ Suppose, for example, that we want to know whether a
reversible transformation $H:\left\{  0,1\right\}  ^{n}\rightarrow\left\{
0,1\right\}  ^{n}$\ can be synthesized using gates $G_{1},\ldots,G_{K}$. \ If
we knew some upper bound on the number of ancilla bits that might be needed by
the generating circuit, then if nothing else, we could of course solve this
problem by brute force. \ The trouble is that, without the classification, it
is not obvious how to prove \textit{any }upper bound on the number of
ancillas---not even, say, $\operatorname*{Ackermann}\left(  n\right)  $.
\ This makes it unclear, \textit{a priori}, whether \textsc{RevGen} is even
\textit{decidable}, never mind its complexity!

One \textit{can} show on abstract grounds that \textsc{RevGen}\ is decidable,
but with an astronomical running time. \ To explain this requires a short
digression. \ In universal algebra, there is a body of theory (see
e.g.\ \cite{lau}), which grew out of Post's original work \cite{post}, about
the general problem of classifying closed classes of functions (clones) of
various kinds. \ The upshot is that every clone is characterized by an
\textit{invariant} that all functions in the clone preserve: for example,
affineness for the $\operatorname*{NOT}$\ and $\operatorname*{XOR}$ functions,
or monotonicity for the $\operatorname*{AND}$\ and $\operatorname*{OR}$
functions. \ The clone can then be shown to contain \textit{all} functions
that preserve the invariant. \ (There is a formal definition of
\textquotedblleft invariant,\textquotedblright\ involving polymorphisms, which
makes this statement not a tautology, but we omit it.) \ Alongside the lattice
of clones of functions, there is a dual lattice of \textit{coclones} of
invariants, and there is a Galois connection relating the two: as one adds
more functions, one preserves fewer invariants, and vice versa.

In response to an inquiry by us, Emil Je\v{r}\'{a}bek recently showed
\cite{jerabek} that the clone/coclone duality can be adapted to the setting of
reversible gates. \ This means that we know, even without a classification
theorem, that every closed class of reversible transformations is uniquely
determined by the invariants that it preserves.

Unfortunately, this elegant characterization does not give rise to feasible
algorithms. \ The reason is that, for an $n$-bit gate $G:\left\{  0,1\right\}
^{n}\rightarrow\left\{  0,1\right\}  ^{n}$, the invariants could in principle
involve all $2^{n}$\ inputs, as well arbitrary polymorphisms mapping those
inputs into a commutative monoid. \ Thus the number of polymorphisms one needs
to consider grows at least like $2^{2^{2^{n}}}$. \ Now, the word problem for
commutative monoids is decidable, by reduction to the ideal membership problem
(see, e.g., \cite[p. 55]{khar}). \ And by putting these facts together, one
can derive an algorithm for \textsc{RevGen}\ that uses doubly-exponential
space and triply-exponential time, as a function of the truth table sizes: in
other words, $\exp\left(  \exp\left(  \exp\left(  \exp\left(  n\right)
\right)  \right)  \right)  $ time, as a function of $n$. \ We believe it
should also be possible to extract $\exp\left(  \exp\left(  \exp\left(
\exp\left(  n\right)  \right)  \right)  \right)  $\ upper bounds on the number
of gates and ancillas from this algorithm, although we have not verified the details.

\subsection{Proof Ideas\label{TECHNIQUES}}

We hope we have made the case that the classification theorem improves the
complexity situation for reversible circuit synthesis! \ Even so, some people
might regard classifying all possible reversible gate sets as a complicated,
maybe worthwhile, but fundamentally tedious exercise. \ Can't such problems be
automated via computer search? \ On the contrary, there are specific aspects
of reversible computation that make this classification problem both unusually
rich, and unusually hard to reduce to any finite number of cases.

We already discussed the astronomical number of possible invariants that even
a tiny reversible gate (say, a $3$-bit gate) might satisfy, and the
hopelessness of enumerating them by brute force. \ However, even if we could
cut down the number of invariants to something reasonable, there would still
be the problem\ that the size, $n$, of a reversible gate can be arbitrarily
large---and as one considers larger gates, one can discover more and more
invariants. \ Indeed, that is precisely what happens in our case, since the
Hamming weight mod $k$ invariant can only be \textquotedblleft
noticed\textquotedblright\ by considering gates on $k$ bits or more. \ There
are also \textquotedblleft sporadic\textquotedblright\ affine classes that can
only be found by considering $6$-bit gates.

Of course, it is not hard just to \textit{guess} a large number of reversible
gate classes (affine transformations, parity-preserving and parity-flipping
transformations, etc.), prove that these classes are all distinct, and then
prove that each one can be generated by a simple set of gates (e.g.,
$\operatorname*{CNOT}$ or $\operatorname*{Fredkin}+\operatorname*{NOT}$).
\ Also, once one has a sufficiently powerful gate (say, the
$\operatorname*{CNOT}$\ gate), it is often straightforward to classify all the
classes \textit{containing} that gate. \ So for example, it is relatively easy
to show that $\operatorname*{CNOT}$, together with any non-affine gate,
generates all reversible transformations.

As usual with classification problems, the hard part is to rule out exotic
additional classes: most of the work, one might say, is not about what is
there, but about what isn't there. \ It is one thing to synthesize some random
$1000$-bit reversible transformation using only $\operatorname*{Toffoli}%
$\ gates, but quite another to synthesize a $\operatorname*{Toffoli}$\ gate
using only the random $1000$-bit transformation!

Thinking about this brings to the fore the central issue: that in reversible
computation, it is not enough to output some desired string $F\left(
x\right)  $; one needs to output nothing else \textit{besides} $F\left(
x\right)  $. \ And hence, for example, it does not suffice to look inside the
random $1000$-bit reversible gate $G$, to show that it contains a
$\operatorname*{NAND}$\ gate, which is computationally universal. \ Rather,
one needs to deal with \textit{all} of $G$'s outputs, and show that one can
eliminate the undesired ones.

The way we do that involves another characteristic property of reversible
circuits: that they can have \textquotedblleft global conserved
quantities,\textquotedblright\ such as Hamming weight. \ Again and again, we
need to prove that if a reversible gate $G$ \textit{fails} to conserve some
quantity, such as the Hamming weight mod $k$, then that fact alone implies
that we can use $G$ to implement a desired behavior. \ This is where
elementary algebra and number theory come in.

There are two aspects to the problem. \ First, we need to understand something
about the possible quantities that a reversible gate can conserve. \ For
example, we will need the following three results:

\begin{itemize}
\item No non-conservative reversible gate can conserve inner products mod $k$,
unless $k=2$.

\item No reversible gate can change Hamming weight mod $k$ by a fixed, nonzero
amount, unless $k=2$.

\item No nontrivial linear gate can conserve Hamming weight mod $k$, unless
$k=2$ or $k=4$.
\end{itemize}

We prove each of these statements in Section \ref{HAMMING}, using arguments
based on complex polynomials. \ In Appendix \ref{ALTPROOF}, we give
alternative, more \textquotedblleft combinatorial\textquotedblright\ proofs
for the second and third statements.

Next, using our knowledge about the possible conserved quantities, we need
procedures that take any gate $G$ that fails to conserve some quantity, and
that use $G$ to implement a desired behavior (say, making a single copy of a
bit, or changing an inner product by exactly $1$). \ We then leverage that
behavior to generate a desired gate (say, a $\operatorname{Fredkin}$\ gate).
\ The two core tasks turn out to be the following:

\begin{itemize}
\item Given any non-affine gate, we need to construct a
$\operatorname*{Fredkin}$\ gate. \ We do this in Sections \ref{CONSERV}\ and
\ref{MOD}.

\item Given any non-orthogonal linear gate, we need to construct a
$\operatorname*{CNOTNOT}$\ gate, a parity-preserving version of
$\operatorname*{CNOT}$\ that maps $x,y,z$\ to $x,y\oplus x,z\oplus x$. \ We do
this in Section \ref{TOCNOTNOT}.
\end{itemize}

In both of these cases, our solution involves $3$-dimensional lattices: that
is, subsets of $\mathbb{Z}^{3}$\ closed under integer linear combinations.
\ We argue, in essence, that the only possible obstruction to the desired
behavior is a \textquotedblleft modularity obstruction,\textquotedblright\ but
the assumption about the gate $G$ rules out such an obstruction.

We can illustrate this with an example that ends up \textit{not} being needed
in the final classification proof, but that we worked out earlier in this
research.\footnote{In general, after completing the classification proof, we
were able to go back and simplify it substantially, by removing results---for
example, about the generation of $\operatorname{CNOT}$\ gates---that were
important for working out the lattice in the first place, but which then
turned out to be subsumed (or which \textit{could} be subsumed, with modest
additional effort) by later parts of the classification. \ Our current proof
reflects these simplifications.}\ \ Let $G$\ be any gate that does not
conserve (or anti-conserve) the Hamming weight mod $k$ for any $k\geq2$,
and\ suppose we want to use $G$\ to construct a $\operatorname*{CNOT}$\ gate.

\begin{figure}[h]
\centering
\begin{tikzpicture}[>=latex,scale=0.6]
\draw[very thin,color=gray] (-3.5,-0.5) grid (6.5,8.5);
\draw[->] (0,-0.5) -- (0,8.5);
\draw[<->] (-3.5,0) -- (6.5,0);
\draw[->,very thick,color=green] (0,0) -- (6,4);
\draw[->,very thick,color=blue] (0,0) -- (-3,2);
\draw[->,very thick,color=red] (0,0) -- (-1,8);
\node [below] at (1,0) {\phantom{\tiny (1,0)}};
\node [below] at (2,0) {\phantom{\tiny (2,0)}};
\node at (1.5, -1.5) {Generators};
\end{tikzpicture} \hspace{20px} \begin{tikzpicture}[>=latex,scale=0.6]
\draw[very thin,color=gray] (-0.5,-0.5) grid (7.5,8.5);
\draw[->] (0,-0.5) -- (0,8.5);
\draw[->] (-0.5,0) -- (7.5,0);
\draw[->,very thick,color=green] (1,0) -- (7,4);
\draw[->,very thick,color=blue] (7,4) -- (4,6);
\draw[->,very thick,color=blue] (4,6) -- (1,8);
\draw[->,very thick,color=red] (1,8) -- (2,0);
\node [below] at (1,0) {\tiny (1,0)};
\node [below] at (2,0) {\tiny (2,0)};
\node at (3.5, -1.5) {Copying Sequence};
\end{tikzpicture}
\caption{Moving within first quadrant of lattice to construct a COPY gate}%
\label{fig:lattice}%
\end{figure}
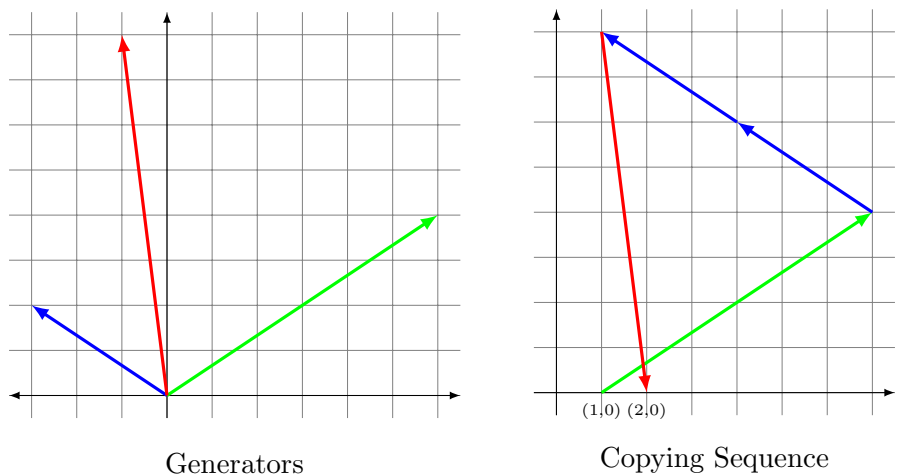

Then we examine how $G$ behaves on restricted inputs:\ in this case, on inputs
that consist entirely of some number of copies of $x$\ and $\overline{x}$,
where $x\in\left\{  0,1\right\}  $\ is a bit, as well as constant $0$ and $1$
bits. \ For example, perhaps $G$ can increase the number of copies of $x$\ by
$5$ while decreasing the number of copies of $\overline{x}$\ by $7$, and can
\textit{also} decrease the number of copies of $x$\ by $6$\ without changing
the number of copies of $\overline{x}$. \ Whatever the case, the set of
possible behaviors generates some lattice: in this case, a lattice in
$\mathbb{Z}^{2}$ (see Figure \ref{fig:lattice}). \ We need to argue that the
lattice contains a distinguished point encoding the desired \textquotedblleft
copying\textquotedblright\ behavior. \ In the case of the $\operatorname{CNOT}%
$ gate, the point is\ $\left(  1,0\right)  $, since we want one more copy of
$x$\ and no more copies of $\overline{x}$. \ Showing that the lattice contains
$\left(  1,0\right)  $, in turn, boils down to arguing that a certain system
of Diophantine linear equations must have a solution. \ One can do this,
finally, by using the assumption that $G$\ does not conserve or anti-conserve
the Hamming weight mod $k$ for any $k$.

To generate the $\operatorname*{Fredkin}$\ gate, we instead use the Chinese
Remainder Theorem to combine gates that change the inner product mod $p$ for
various primes $p$ into a gate that changes the inner product between two
inputs by exactly $1$; while to generate the $\operatorname*{CNOTNOT}$\ gate,
we exploit the assumption that our generating gates are linear. \ In all these
cases, it is crucial that we know, from Section \ref{HAMMING}, that certain
quantities \textit{cannot} be conserved by any reversible gate.

There are a few parts of the classification proof (for example, Section
\ref{NOT}, on affine gate sets) that basically \textit{do} come down to
enumerating cases, but we hope to have given a sense for the interesting parts.

\subsection{Related Work\label{RELATED}}

Surprisingly, the general question of classifying reversible gates such as
$\operatorname*{Toffoli}$\ and $\operatorname*{Fredkin}$\ appears never to
have been asked, let alone answered, prior to this work.

In the reversible computing literature, there are hundreds of papers on
synthesizing reversible circuits (see \cite{saeedi} for a survey), but most of
them focus on practical considerations: for example, trying to minimize the
number of $\operatorname*{Toffoli}$\ gates or other measures of interest,
often using software optimization tools. \ We found only a tiny amount of work
relevant to the classification problem: notably, an unpublished preprint by
Lloyd \cite{lloyd:gate}, which shows that every non-affine reversible gate is
computationally universal, if one does not care what garbage\ is generated in
addition to the desired output. \ Lloyd's result was subsequently rediscovered
by Kerntopf et al.\ \cite{kerntopf}\ and De Vos and Storme \cite{devos}. \ We
will reprove this result for completeness in Section \ref{GARBSEC}, as we use
it as one ingredient in our proof.

There is also work by Morita et al.\ \cite{morita} that uses brute-force
enumeration to classify certain reversible computing elements with $2$, $3$,
or $4$ wires, but the notion of \textquotedblleft reversible
gate\textquotedblright\ there is very different from the standard one (the
gates are for routing a single \textquotedblleft billiard
ball\textquotedblright\ element rather than for transforming bit strings,\ and
they have internal state). \ Finally, there is work by Strazdins
\cite{strazdins}, not motivated by reversible computing,\ which considers
classifying reversible Boolean functions, but which imposes a separate
requirement on each output bit that it belong to one of the classes from
Post's original lattice, and which thereby misses all the reversible gates
that conserve \textquotedblleft global\textquotedblright\ quantities, such as
the $\operatorname*{Fredkin}$\ gate.\footnote{Because of different rules
regarding constants, developed with Post's lattice rather than reversible
computing in mind, Strazdins also includes classes that we do not (e.g.,
functions that always map $0^{n}$\ or $1^{n}$\ to themselves, but are
otherwise arbitrary). \ To use our notation, his $13$-class lattice ends up
intersecting our infinite lattice in just five classes: $\left\langle
\varnothing\right\rangle $, $\left\langle \operatorname*{NOT}\right\rangle $,
$\left\langle \operatorname*{CNOTNOT},\operatorname*{NOT}\right\rangle $,
$\left\langle \operatorname*{CNOT}\right\rangle $, and $\left\langle
\operatorname*{Toffoli}\right\rangle $.}

\section{Notation and Definitions\label{DEF}}

$\mathbb{F}_{2}$ means the field of $2$ elements. \ $\left[  n\right]
$\ means $\left\{  1,\ldots,n\right\}  $. \ We denote by $e_{1},\ldots,e_{n}%
$\ the standard basis for the vector space $\mathbb{F}_{2}^{n}$: that is,
$e_{1}=\left(  1,0,\ldots,0\right)  $, etc.

Let $x=x_{1}\ldots x_{n}$\ be an $n$-bit string. \ Then $\overline{x}$ means
$x$ with all $n$ of its bits inverted. \ Also,$\ x\oplus y$ means bitwise XOR,
$x,y$ or $xy$\ means concatenation, $x^{k}$\ means the concatenation of $k$
copies of $x$, and $\left\vert x\right\vert $ means the Hamming weight. \ The
\textit{parity} of $x$\ is $\left\vert x\right\vert \operatorname{mod}2$.
\ The \textit{inner product} of $x$\ and $y$\ is the integer $x\cdot
y=x_{1}y_{1}+\cdots+x_{n}y_{n}$. \ Note that%
\[
x\cdot\left(  y\oplus z\right)  \equiv x\cdot y+x\cdot z\left(
\operatorname{mod}2\right)  ,
\]
but the above need not hold if we are not working mod $2$.

By $\operatorname{gar}\left(  x\right)  $, we mean garbage depending on $x$:
that is, \textquotedblleft scratch work\textquotedblright\ that a reversible
computation generates along the way to computing some desired function
$f\left(  x\right)  $. \ Typically, the garbage later needs to be
\textit{uncomputed}. \ Uncomputing, a term introduced by Bennett
\cite{bennett}, simply means running an entire computation in reverse, after
the output $f\left(  x\right)  $\ has been safely stored.

\subsection{Gates\label{GATES}}

By a \textit{(reversible) gate}, throughout this paper we will mean a
reversible transformation $G$ on the set of $k$-bit strings: that is, a
permutation of $\left\{  0,1\right\}  ^{k}$, for some fixed $k$. \ Formally,
the terms `gate' and `reversible transformation' will mean the same thing;
`gate'\ just connotes a reversible transformation that is particularly small
or simple. \ 

A gate is \textit{nontrivial} if it does something other than permute its
input bits,\ and \textit{non-degenerate} if it does something other than
permute its input bits and/or apply $\operatorname*{NOT}$'s to some subset of them.

A gate $G$\ is \textit{conservative} if it satisfies $\left\vert G\left(
x\right)  \right\vert =\left\vert x\right\vert $\ for all $x$. \ A gate is
\textit{mod-}$k$\textit{-respecting} if there exists a $j$ such that%
\[
\left\vert G\left(  x\right)  \right\vert \equiv\left\vert x\right\vert
+j\left(  \operatorname{mod}k\right)
\]
for all $x$. \ It's \textit{mod-}$k$\textit{-preserving} if moreover $j=0$.
\ It's \textit{mod-preserving} if it's mod-$k$-preserving for some $k\geq
2$,\ and \textit{mod-respecting} if it's mod-$k$-respecting for some $k\geq2$.

As special cases, a mod-$2$-respecting gate is also called
\textit{parity-respecting}, a mod-$2$-preserving gate is called
\textit{parity-preserving},\ and a gate $G$\ such that%
\[
\left\vert G\left(  x\right)  \right\vert \not \equiv \left\vert x\right\vert
\left(  \operatorname{mod}2\right)
\]
for all $x$ is called \textit{parity-flipping}. \ In Theorem \ref{noshifter},
we will prove that parity-flipping gates are the \textit{only} examples of
mod-respecting gates that are not mod-preserving.

The \textit{respecting number} of a gate $G$, denoted $k\left(  G\right)  $,
is the largest $k$ such that $G$ is mod-$k$-respecting. \ (By convention, if
$G$\ is conservative then $k\left(  G\right)  =\infty$, while if $G$ is
non-mod-respecting then $k\left(  G\right)  =1$.) \ We have the following fact:

\begin{proposition}
\label{chinese}$G$ is mod-$\ell$-respecting\ if and only if $\ell$\ divides
$k\left(  G\right)  $.
\end{proposition}

\begin{proof}
If $\ell$\ divides $k\left(  G\right)  $, then certainly $G$ is mod-$\ell
$-respecting. \ Now, suppose $G$ is mod-$\ell$-respecting\ but $\ell$\ does
not divide $k\left(  G\right)  $. \ Then $G$\ is both mod-$\ell$%
-respecting\ and mod-$k\left(  G\right)  $-respecting. \ So by the Chinese
Remainder Theorem, $G$ is mod-$\operatorname{lcm}\left(  \ell,k\left(
G\right)  \right)  $-respecting. \ But this contradicts the definition of
$k\left(  G\right)  $.
\end{proof}

A gate $G$\ is \textit{affine} if it implements an affine transformation over
$\mathbb{F}_{2}$: that is, if there exists an invertible matrix $A\in
\mathbb{F}_{2}^{k\times k}$, and a vector $b\in\mathbb{F}_{2}^{k}$, such that
$G\left(  x\right)  =Ax\oplus b$ for all $x$. \ A gate is \textit{linear} if
moreover $b=0$. \ A gate is \textit{orthogonal} if it satisfies%
\[
G\left(  x\right)  \cdot G\left(  y\right)  \equiv x\cdot y\left(
\operatorname{mod}2\right)
\]
for all $x,y$. \ (We will observe, in Lemma \ref{orthoglin}, that every
orthogonal gate is linear.) \ Also, if $G\left(  x\right)  =Ax\oplus b$\ is
affine, then the \textit{linear part of }$G$ is the linear transformation
$G^{\prime}\left(  x\right)  =Ax$. \ We call $G$ orthogonal in its linear
part, mod-$k$-preserving in its linear part, etc.\ if $G^{\prime}$ satisfies
the corresponding invariant. \ A gate that is orthogonal in its linear part is
also called an \textit{isometry}.

Given two gates $G$ and $H$, their \textit{tensor product}, $G\otimes H$, is a
gate that applies $G$ and $H$ to disjoint sets of bits. \ We will often use
the tensor product to produce a single gate that combines the properties of
two previous gates. \ Also, we denote by $G^{\otimes t}$\ the tensor product
of $t$ copies of $G$.

\subsection{Gate Classes\label{CLASSES}}

Let $S=\left\{  G_{1},G_{2},\ldots\right\}  $ be a set of gates, possibly on
different numbers of bits and possibly infinite. \ Then $\left\langle
S\right\rangle =\left\langle G_{1},G_{2},\ldots\right\rangle $, \textit{the
class of reversible transformations generated by} $S$, can be defined as the
smallest set of reversible transformations $F:\left\{  0,1\right\}
^{n}\rightarrow\left\{  0,1\right\}  ^{n}$ that satisfies the following
closure properties:

\begin{enumerate}
\item[(1)] \textbf{Base case.} $\ \left\langle S\right\rangle $ contains $S$,
as well as the identity function $F\left(  x_{1}\ldots x_{n}\right)
=x_{1}\ldots x_{n}$\ for all $n\geq1$.

\item[(2)] \textbf{Composition rule.} \ If $\left\langle S\right\rangle
$\ contains $F\left(  x_{1}\ldots x_{n}\right)  $\ and $G\left(  x_{1}\ldots
x_{n}\right)  $, then $\left\langle S\right\rangle $\ also contains $F\left(
G\left(  x_{1}\ldots x_{n}\right)  \right)  $.

\item[(3)] \textbf{Swapping rule.} \ If $\left\langle S\right\rangle
$\ contains $F\left(  x_{1}\ldots x_{n}\right)  $, then $\left\langle
S\right\rangle $\ also contains all possible functions $\sigma\left(  F\left(
x_{\tau\left(  1\right)  }\ldots x_{\tau\left(  n\right)  }\right)  \right)
$\ obtained by permuting $F$'s input and output bits.

\item[(4)] \textbf{Extension rule.} \ If $\left\langle S\right\rangle
$\ contains $F\left(  x_{1}\ldots x_{n}\right)  $, then $\left\langle
S\right\rangle $\ also contains the function%
\[
G\left(  x_{1}\ldots x_{n},b\right)  :=\left(  F\left(  x_{1}\ldots
x_{n}\right)  ,b\right)  ,
\]
in which $b$\ occurs as a \textquotedblleft dummy\textquotedblright\ bit.

\item[(5)] \textbf{Ancilla rule.} \ If $\left\langle S\right\rangle
$\ contains a function $F$\ that satisfies%
\[
F\left(  x_{1}\ldots x_{n},a_{1}\ldots a_{k}\right)  =\left(  G\left(
x_{1}\ldots x_{n}\right)  ,a_{1}\ldots a_{k}\right)  ~~\forall x_{1}\ldots
x_{n}\in\left\{  0,1\right\}  ^{n},
\]
for some smaller function $G$ and fixed \textquotedblleft
ancilla\textquotedblright\ string\ $a_{1}\ldots a_{k}\in\left\{  0,1\right\}
^{k}$\ that do not depend on $x$, then $\left\langle S\right\rangle $\ also
contains $G$. \ (Note that, if the $a_{i}$'s are set to other values, then $F$
need not have the above form.)
\end{enumerate}

Note that because of reversibility, the set of $n$-bit\ transformations in
$\left\langle S\right\rangle $\ (for any $n$) always forms a group. \ Indeed,
if $\left\langle S\right\rangle $\ contains $F$, then clearly $\left\langle
S\right\rangle $\ contains all the iterates $F^{2}\left(  x\right)  =F\left(
F\left(  x\right)  \right)  $, etc. \ But since there must be some positive
integer $m$ such that $F^{m}\left(  x\right)  =x$, this means that
$F^{m-1}\left(  x\right)  =F^{-1}\left(  x\right)  $. \ Thus, we do not need a
separate rule stating that $\left\langle S\right\rangle $\ is closed under inverses.

We say $S$ \textit{generates} the reversible transformation $F$\ if
$F\in\left\langle S\right\rangle $. \ We also say that $S$\ generates
$\left\langle S\right\rangle $. \ If $\left\langle S\right\rangle $\ equals
the set of all permutations of $\left\{  0,1\right\}  ^{n}$, for all $n\geq1$,
then we call $S$\ \textit{universal}.

Given an arbitrary set $\mathcal{C}$\ of reversible transformations, we call
$\mathcal{C}$ a \textit{reversible gate class} (or \textit{class} for short)
if $\mathcal{C}$ is closed under rules (2)-(5) above: in other words, if there
exists an $S$ such that $\mathcal{C}=\left\langle S\right\rangle $.

A \textit{reversible circuit} for the function $F$, over the gate set $S$, is
an explicit procedure for generating $F$ by applying gates in $S$, and thereby
showing that $F\in\left\langle S\right\rangle $. \ An example is shown in
Figure~\ref{circuitfig}. \ Reversible circuit diagrams are read from left to
right, with each bit that occurs in the circuit (both input and ancilla bits)
represented by a horizontal line, and each gate represented by a vertical line.

If every gate $G\in S$\ satisfies some invariant, then we can also describe
$S$ and $\left\langle S\right\rangle $ as satisfying that invariant. \ So for
example, the set $\left\{  \operatorname*{CNOTNOT},\operatorname*{NOT}%
\right\}  $ is affine and parity-respecting, and so is the class that it
generates. \ Conversely, $S$ violates an invariant if any $G\in S$\ violates it.

Just as we defined the respecting number $k\left(  G\right)  $\ of a gate, we
would like to define the respecting number $k\left(  S\right)  $\ of an entire
gate set. \ To do so, we need a proposition about the behavior of $k\left(
G\right)  $ under tensor products.

\begin{figure}[h]
\centering
\begin{minipage}[c]{.3\textwidth}
\Qcircuit @C=1.5em @R=1em {
x_1 & & \ctrl{4} & \qw & \ctrl{4} & \qw \\
x_2 & & \qswap & \qw & \qswap & \qw \\
x_3 & & \qw & \qswap & \qw & \qw \\
x_4 & & \qw & \qswap & \qw & \qw \\
0 & & \qswap & \ctrl{-2} & \qswap & \qw \\
}
\end{minipage}
\caption{Generating a Controlled-Controlled-Swap gate from Fredkin}%
\label{circuitfig}%
\end{figure}
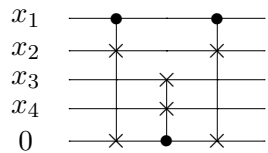

\begin{proposition}
\label{kgtensor}For all gates $G$ and $H$,%
\[
k\left(  G\otimes H\right)  =\gcd\left(  k\left(  G\right)  ,k\left(
H\right)  \right)  .
\]

\end{proposition}

\begin{proof}
Letting $\gamma=\gcd\left(  k\left(  G\right)  ,k\left(  H\right)  \right)  $,
clearly $G\otimes H$\ is mod-$\gamma$-respecting. \ To see that $G\otimes
H$\ is not mod-$\ell$-respecting for any $\ell>\gamma$: by definition, $\ell
$\ must fail to divide either $k\left(  G\right)  $ or $k\left(  H\right)  $.
\ Suppose it\ fails to divide $k\left(  G\right)  $\ without loss of
generality. \ Then $G$ cannot be mod-$\ell$-respecting, by Proposition
\ref{chinese}. \ But if we consider pairs of inputs to $G\otimes H$\ that
differ only on $G$'s input, then this implies that $G\otimes H$\ is not
mod-$\ell$-respecting either.
\end{proof}

If $S=\left\{  G_{1},G_{2},\ldots\right\}  $, then because of Proposition
\ref{kgtensor},\ we can define $k\left(  S\right)  $\ as $\gcd\left(  k\left(
G_{1}\right)  ,k\left(  G_{2}\right)  ,\ldots\right)  $. \ For then not only
will every transformation in $\left\langle S\right\rangle $\ be mod-$k\left(
S\right)  $-respecting, but there will exist transformations in $\left\langle
S\right\rangle $\ that are not mod-$\ell$-respecting\ for any $\ell>k\left(
S\right)  $.

We then have that $S$ is mod-$k$-respecting if and only if $k$ divides
$k\left(  S\right)  $, and mod-respecting if and only if $S$ is mod-$k$%
-respecting for some $k\geq2$.

\subsection{Alternative Kinds of Generation\label{ALTGEN}}

We now discuss four alternative notions of what it can mean for a reversible
gate set to \textquotedblleft generate\textquotedblright\ a transformation.
\ Besides being interesting in their own right, some of these notions will
also be used in the proof of our main classification theorem.

\textbf{Partial Gates.} \ A \textit{partial reversible gate} is an injective
function $H:D\rightarrow\left\{  0,1\right\}  ^{n}$, where $D$ is some subset
of $\left\{  0,1\right\}  ^{n}$. \ Such an $H$ is \textit{consistent} with a
full reversible gate $G$ if $G\left(  x\right)  =H\left(  x\right)
$\ whenever $x\in D$. \ Also, we say that a reversible gate set $S$
\textit{generates} $H$ if $S$ generates any $G$ with which $H$ is consistent.
\ As an example, $\operatorname{COPY}$\ is the $2$-bit partial reversible gate
defined by the following relations:%
\[
\operatorname{COPY}\left(  00\right)  =00,~~~~~\operatorname{COPY}\left(
10\right)  =11.
\]
If a gate set $S$\ can implement the above behavior, using ancilla bits that
are returned to their original states by the end, then we say $S$
\textquotedblleft generates $\operatorname{COPY}$\textquotedblright;\ the
behavior on inputs $01$\ and $11$\ is irrelevant. \ Note that
$\operatorname{COPY}$\ is consistent with $\operatorname{CNOT}$. \ One can
think of $\operatorname{COPY}$\ as a bargain-basement $\operatorname{CNOT}$,
but one that might be bootstrapped up to a full $\operatorname{CNOT}$\ with
further effort.

\textbf{Generation With Garbage.}\ \ Let $D\subseteq\left\{  0,1\right\}
^{m}$, and $H:D\rightarrow\left\{  0,1\right\}  ^{n}$ be some function, which
need not be injective or surjective, or even have the same number of input and
output bits. \ Then we say that a reversible gate set $S$ \textit{generates
}$H$\textit{ with garbage} if there exists a reversible transformation
$G\in\left\langle S\right\rangle $, as well as an ancilla string $a$ and a
function $\operatorname*{gar}$, such that $G\left(  x,a\right)  =\left(
H\left(  x\right)  ,\operatorname*{gar}\left(  x\right)  \right)  $ for all
$x\in D$. \ As an example, consider the ordinary $2$-bit $\operatorname{AND}%
$\ function, from $\left\{  0,1\right\}  ^{2}$\ to $\left\{  0,1\right\}  $.
\ Since $\operatorname{AND}$\ destroys information, clearly no reversible gate
can generate it in the usual sense, but many reversible gates can generate
$\operatorname{AND}$\ with garbage: for instance, the $\operatorname{Toffoli}%
$\ and $\operatorname{Fredkin}$\ gates, as we saw in Section \ref{REVINTRO}.

\textbf{Encoded Universality.} \ This is a concept borrowed from quantum
computing \cite{bkdlw}. \ In our setting, encoded universality means that
there is some way of encoding $0$'s and $1$'s by longer strings, such that our
gate set can implement any desired transformation on the encoded bits. \ Note
that, while this is a weaker notion of universality than the ability to
generate arbitrary permutations of $\left\{  0,1\right\}  ^{n}$, it is
stronger than \textquotedblleft merely\textquotedblright\ computational
universality, because it still requires a transformation to be performed
reversibly, with no garbage left around. \ Formally, given a reversible gate
set $S$, we say that $S$ \textit{supports encoded universality} if there are
$k$-bit strings $\alpha\left(  0\right)  $ and $\alpha\left(  1\right)  $ such
that for every $n$-bit reversible transformation $F\left(  x_{1}\ldots
x_{n}\right)  =y_{1}\ldots y_{n}$, there exists a transformation
$G\in\left\langle S\right\rangle $\ that satisfies%
\[
G\left(  \alpha\left(  x_{1}\right)  \ldots\alpha\left(  x_{n}\right)
\right)  =\alpha\left(  y_{1}\right)  \ldots\alpha\left(  y_{n}\right)
\]
for all $x\in\left\{  0,1\right\}  ^{n}$. \ Also, we say that $S$
\textit{supports affine encoded universality} if this is true for every affine
$F$.

As a well-known example, the $\operatorname{Fredkin}$\ gate is not universal
in the usual sense, because it preserves Hamming weight. \ But it is easy to
see that $\operatorname{Fredkin}$\ supports encoded universality, using the
so-called \textit{dual-rail encoding}, in which every $0$ bit is encoded as
$01$, and every $1$ bit is encoded as $10$. \ In Section \ref{ENCODED}, we
will show, as a consequence of our classification theorem, that \textit{every}
reversible gate set (except for degenerate sets) supports either encoded
universality or affine encoded universality.

\textbf{Loose Generation.} \ Finally, we say that a gate set $S$%
\ \textit{loosely generates} a reversible transformation $F:\left\{
0,1\right\}  ^{n}\rightarrow\left\{  0,1\right\}  ^{n}$, if there exists a
transformation $G\in\left\langle S\right\rangle $, as well as ancilla strings
$a$\ and $b$, such that%
\[
G\left(  x,a\right)  =\left(  F\left(  x\right)  ,b\right)
\]
for all $x\in\left\{  0,1\right\}  ^{n}$. \ In other words, $G$ is allowed to
change the ancilla bits, so long as they change in a way that is independent
of the input $x$. \ Under this rule, one could perhaps tell by examining the
ancilla bits \textit{that} $G$\ was applied, but one could not tell to which
input. \ This suffices for some applications of reversible computing, though
not for others.\footnote{For example, if $G$ were applied to a quantum
superposition, then it would still maintain coherence among all the inputs to
which it was applied---though perhaps not between those inputs and other
inputs in the superposition to which it was \textit{not} applied.}

\section{Stating the Classification Theorem\label{THEOREM}}

In this section we state our main result, and make a few preliminary remarks
about it. \ First let us define the gates that appear in the classification theorem.

\begin{itemize}
\item $\operatorname*{NOT}$ is the $1$-bit gate that maps $x$\ to
$\overline{x}$.

\item $\operatorname*{NOTNOT}$, or $\operatorname{NOT}^{\otimes2}$, is the
$2$-bit gate that maps $xy$\ to $\overline{x}\overline{y}$.
\ $\operatorname*{NOTNOT}$\ is a parity-preserving variant of
$\operatorname*{NOT}$.

\item $\operatorname*{CNOT}$ (Controlled-NOT) is the $2$-bit gate that maps
$x,y$\ to $x,y\oplus x$. $\operatorname*{CNOT}$\ is affine.

\item $\operatorname*{CNOTNOT}$ is the $3$-bit gate that maps $x,y,z$\ to
$x,y\oplus x,z\oplus x$. \ $\operatorname*{CNOTNOT}$\ is affine and parity-preserving.

\item $\operatorname*{Toffoli}$ (also called Controlled-Controlled-NOT, or
CCNOT) is the $3$-bit gate that maps $x,y,z$\ to $x,y,z\oplus xy$.

\item $\operatorname*{Fredkin}$ (also called Controlled-SWAP, or CSWAP) is the
$3$-bit gate that maps $x,y,z$\ to $x,y\oplus x\left(  y\oplus z\right)
,z\oplus x\left(  y\oplus z\right)  $. \ In other words, it swaps $y$\ with
$z$\ if $x=1$, and does nothing if $x=0$. \ $\operatorname*{Fredkin}$ is
conservative: it never changes the Hamming weight.

\item $\operatorname*{C}_{k}$ is a $k$-bit gate that maps $0^{k}$\ to $1^{k}%
$\ and $1^{k}$\ to $0^{k}$, and all other $k$-bit strings to themselves.
\ $\operatorname*{C}_{k}$\ preserves the Hamming weight mod $k$. \ Note that
$\operatorname*{C}_{1}=\operatorname*{NOT}$, while $\operatorname*{C}_{2}$\ is
equivalent to $\operatorname*{NOTNOT}$, up to a bit-swap.

\item $\operatorname*{T}_{k}$ is a $k$-bit gate (for even $k$) that maps
$x$\ to $\overline{x}$\ if $\left\vert x\right\vert $\ is odd, or to $x$ if
$\left\vert x\right\vert $\ is even. \ A different definition is%
\[
\operatorname*{T}\nolimits_{k}\left(  x_{1}\ldots x_{k}\right)  =\left(
x_{1}\oplus b_{x},\ldots,x_{k}\oplus b_{x}\right)  ,
\]
where $b_{x}:=x_{1}\oplus\cdots\oplus x_{k}$. \ This shows that
$\operatorname*{T}_{k}$\ is linear. \ Indeed, we also have%
\[
\operatorname*{T}\nolimits_{k}\left(  x\right)  \cdot\operatorname*{T}%
\nolimits_{k}\left(  y\right)  \equiv x\cdot y+\left(  k+2\right)  b_{x}%
b_{y}\equiv x\cdot y\left(  \operatorname{mod}2\right)  ,
\]
which shows that $\operatorname*{T}\nolimits_{k}$\ is orthogonal. \ Note also
that, if $k\equiv2\left(  \operatorname{mod}4\right)  $, then
$\operatorname*{T}\nolimits_{k}$\ preserves Hamming weight mod $4$: if
$\left\vert x\right\vert $\ is even then $\left\vert \operatorname*{T}%
\nolimits_{k}\left(  x\right)  \right\vert =\left\vert x\right\vert $, while
if $\left\vert x\right\vert $\ is odd then%
\[
\left\vert \operatorname*{T}\nolimits_{k}\left(  x\right)  \right\vert \equiv
k-\left\vert x\right\vert \equiv2-\left\vert x\right\vert \equiv\left\vert
x\right\vert \left(  \operatorname{mod}4\right)  .
\]

\item $\operatorname*{F}_{k}$ is a $k$-bit gate (for even $k$) that maps
$x$\ to $\overline{x}$\ if $\left\vert x\right\vert $\ is even, or to $x$ if
$\left\vert x\right\vert $\ is odd. \ A different definition is%
\[
\operatorname*{F}\nolimits_{k}\left(  x_{1}\ldots x_{k}\right)  =\overline
{\operatorname*{T}\nolimits_{k}\left(  x_{1}\ldots x_{k}\right)  }=\left(
x_{1}\oplus b_{x}\oplus1,\ldots,x_{k}\oplus b_{x}\oplus1\right)
\]
where $b_{x}$\ is as above. \ This shows that $\operatorname*{F}_{k}$\ is
affine. \ Indeed, if $k$ is a multiple of $4$, then $\operatorname*{F}_{k}%
$\ preserves Hamming weight mod $4$: if $\left\vert x\right\vert $ is odd then
$\left\vert \operatorname*{F}_{k}\left(  x\right)  \right\vert =\left\vert
x\right\vert $, while if $\left\vert x\right\vert $\ is even then%
\[
\left\vert \operatorname*{F}\nolimits_{k}\left(  x\right)  \right\vert \equiv
k-\left\vert x\right\vert \equiv\left\vert x\right\vert \left(
\operatorname{mod}4\right)  .
\]
Since $\operatorname*{F}_{k}$\ is equal to $\operatorname*{T}\nolimits_{k}$ in
its linear part, $\operatorname*{F}_{k}$\ is also an isometry.
\end{itemize}

We can now state the classification theorem.

\begin{theorem}
[Main Result]\label{main}Every set of reversible gates generates one of the
following classes:

\begin{enumerate}
\item The trivial class (which contains only bit-swaps).

\item The class of all transformations (generated by $\operatorname*{Toffoli}$).

\item The class of all conservative transformations (generated by
$\operatorname*{Fredkin}$).

\item For each $k\geq3$, the class of all mod-$k$-preserving transformations
(generated by $\operatorname*{C}_{k}$).

\item The class of all affine transformations (generated by
$\operatorname*{CNOT}$).

\item The class of all parity-preserving affine transformations (generated by
$\operatorname*{CNOTNOT}$).

\item The class of all mod-$4$-preserving affine transformations (generated by
$\operatorname*{F}_{4}$).

\item The class of all orthogonal linear transformations (generated by
$\operatorname*{T}_{4}$).

\item The class of all mod-$4$-preserving orthogonal linear transformations
(generated by $\operatorname*{T}_{6}$).

\item Classes 1, 3, 7, 8, or 9 augmented by a $\operatorname*{NOTNOT}$\ gate
(note: 7 and 8 become equivalent this way).

\item Classes 1, 3, 6, 7, 8, or 9\ augmented by a $\operatorname*{NOT}$\ gate
(note: 7 and 8 become equivalent this way).
\end{enumerate}

Furthermore, all the above classes are distinct except when noted otherwise,
and they fit together in the lattice diagram shown in Figure \ref{lattice}%
.\footnote{Let us mention that $\operatorname*{Fredkin}+\operatorname*{NOTNOT}%
$\ generates the class of all parity-preserving transformations, while
$\operatorname*{Fredkin}+\operatorname*{NOT}$\ generates the class of all
parity-respecting transformations. \ We could have listed the
parity-preserving transformations as a special case of the mod-$k$-preserving
transformations: namely, the case $k=2$. \ If we had done so, though, we would
have had to include the caveat that $\operatorname*{C}_{k}$\ only generates
all mod-$k$-preserving transformations when $k\geq3$\ (when $k=2$, we also
need $\operatorname*{Fredkin}$\ in the generating set).\ \ And in any case,
the parity-respecting class would still need to be listed separately.}
\end{theorem}

\begin{figure}[ptb]
\begin{center}
\label{lattice} \begin{tikzpicture}[>=latex]
\tikzstyle{class}=[circle, thick, minimum size=1.2cm, text width=1.0cm, align=center, draw, font=\tiny]
\tikzstyle{nonaffine}=[class, fill=blue!20]
\tikzstyle{affine}=[class, fill=green!20]
\tikzstyle{orthogonal}=[class, fill=yellow!20]
\tikzstyle{inf4}=[class, fill=yellow!20] 
\tikzstyle{automorphism}=[class,fill=red!20]
\matrix[row sep=0.8cm,column sep=0.8cm] {
& & & \node (ALL) [nonaffine]{$\top$}; & & & & \\
& & \node (CNOT) [affine]{$\CNOT$}; & & \node (FRNOT) [nonaffine]{$\Fredkin$ \\ $+\NOT$}; & & & \\
& & & \node (CNOTNOTN) [affine]{$\CNOTNOT$ \\ $+\NOT$}; & & \node (MOD2) [nonaffine]{$\gateMOD{2}$}; & \node (MOD3) {$\cdots$}; & \\
& & \node (F4NOT) [orthogonal]{$\F4+\NOT$}; & & \node (CNOTNOT) [affine]{$\CNOTNOT$}; & & \node (MOD4) [nonaffine]{$\gateMOD{4}$}; & \node (MOD6) {$\cdots$}; \\
& \node (T6NOT) [inf4]{$\T6+\NOT$}; & & \node (F4NOTNOT) [orthogonal]{$\F4+\NOTNOT$}; & & & \node (MOD8) [nonaffine]{$\gateMOD{8}$}; & \node (MOD12) {$\cdots$}; \\
\node (NOT) [automorphism]{$\NOT$}; & & \node (T6NOTNOT) [inf4]{$\T6+\NOTNOT$}; & \node (T4) [orthogonal]{$\T4$}; & \node (F4) [orthogonal]{$\F4$}; & & \node (MODELLIPSIS) {$\vdots$}; & \node (MODELLIPSIS2) {$\cdots$}; \\
& \node (NOTNOT) [automorphism]{$\NOTNOT$}; & & \node (T6) [inf4]{$\T6$}; &  & \node (FREDKIN) [nonaffine]{$\Fredkin$}; & & \\
& & \node (NONE) [automorphism]{$\bot$}; & & & & \\
};
\path[draw,->] (ALL) edge (FRNOT)
(FRNOT) edge (MOD2)
(ALL) edge[bend left=10] (MOD3)
(MOD2) edge (MOD4)
(MOD2) edge (MOD6)
(MOD3) edge (MOD6)
(MOD4) edge (MOD8)
(MOD4) edge (MOD12)
(MOD6) edge (MOD12)
(MOD8) edge (MODELLIPSIS)
(MOD8) edge (MODELLIPSIS2)
(MOD12) edge (MODELLIPSIS2)
(MODELLIPSIS) edge (FREDKIN)
(MODELLIPSIS2) edge (FREDKIN)
(FREDKIN) edge (NONE)
(ALL) edge (CNOT)
(CNOT) edge (CNOTNOTN)
(CNOTNOTN) edge (CNOTNOT)
(F4NOT) edge (F4NOTNOT)
(F4NOTNOT) edge (F4)
(T6NOT) edge (T6NOTNOT)
(T6NOTNOT) edge (T6)
(NOT) edge (NOTNOT)
(NOTNOT) edge (NONE)
(FRNOT) edge (CNOTNOTN)
(CNOTNOTN) edge (F4NOT)
(F4NOT) edge (T6NOT)
(T6NOT) edge (NOT)
(MOD2) edge (CNOTNOT)
(CNOTNOT) edge (F4NOTNOT)
(F4NOTNOT) edge (T6NOTNOT)
(T6NOTNOT) edge (NOTNOT)
(MOD4) edge (F4)
(F4) edge (T6)
(T6) edge (NONE)
(F4NOTNOT) edge (T4)
(T4) edge (T6);
\tikzstyle{class}=[rectangle, thick, minimum size=0.2cm, align=center, draw, font=\tiny]
\tikzstyle{nonaffine}=[class, fill=blue!20]
\tikzstyle{affine}=[class, fill=green!20]
\tikzstyle{orthogonal}=[class, fill=yellow!20]
\tikzstyle{automorphism}=[class,fill=red!20]
\tikzstyle{legend}=[font=\tiny]
\node (NONAFFINE) [nonaffine, below=1cm of FREDKIN]{};
\node (NONAFFINELABEL) [legend,right=2mm of NONAFFINE]{Non-affine};
\node (AFFINE) [affine, below=1mm of NONAFFINE]{};
\node (AFFINELABEL) [legend, right=2mm of AFFINE]{Affine};
\node (ISOMETRY) [orthogonal, below=1mm of AFFINE]{};
\node (ISOMETRYLABEL) [legend, right=2mm of ISOMETRY]{Isometry};
\node (AUTOMORPHISM) [automorphism, below=1mm of ISOMETRY]{};
\node (AUTOMORPHISMLABEL) [legend, right=2mm of AUTOMORPHISM]{Degenerate};
\node (box) [draw, rectangle, inner sep=4mm, fit = (NONAFFINE) (AUTOMORPHISMLABEL)] {};
\end{tikzpicture}
\end{center}
\caption{The inclusion lattice of reversible gate classes}%
\end{figure}
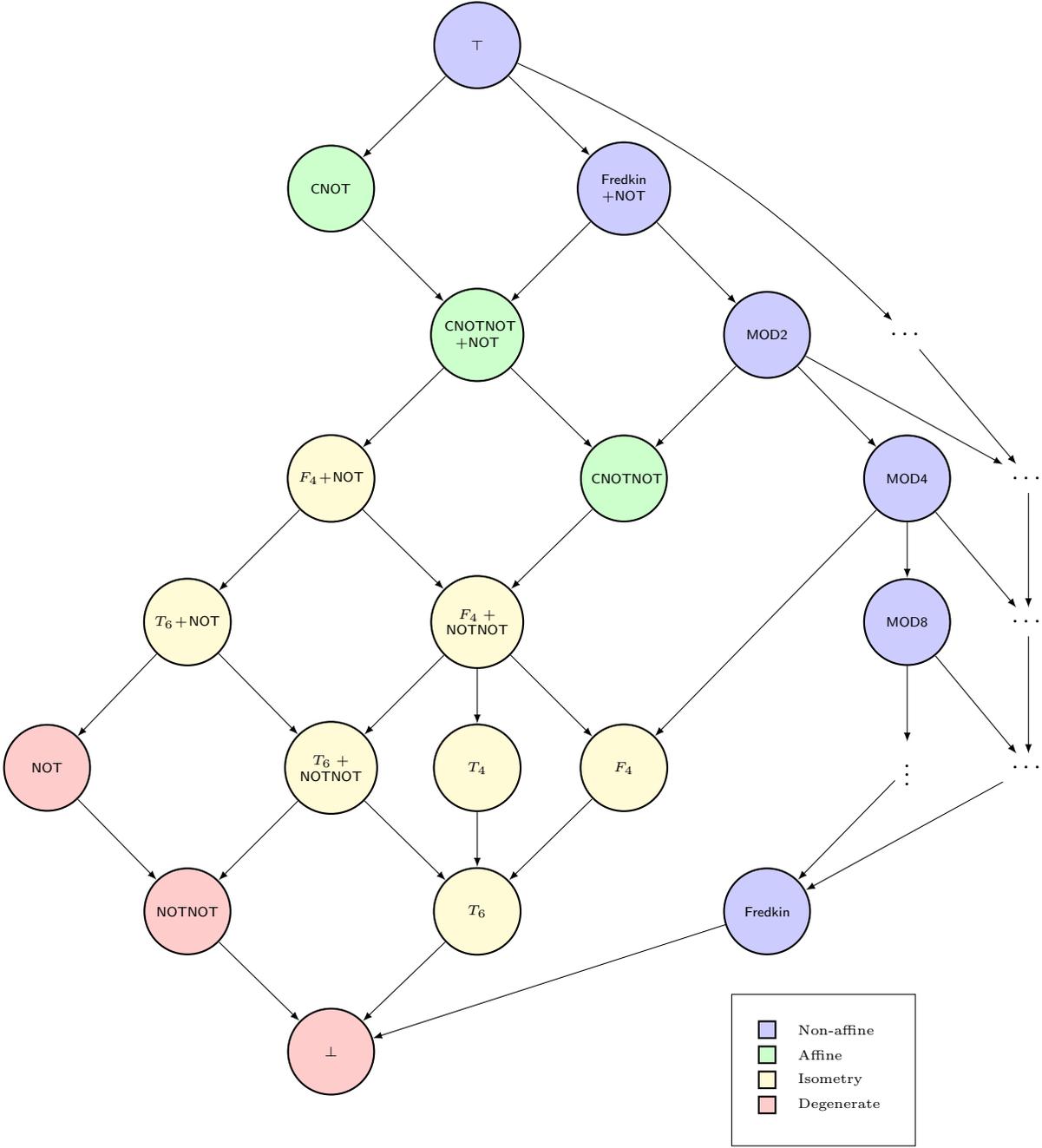

Let us make some comments about the structure of the lattice. \ The lattice
has a countably infinite number of classes, with the one infinite part given
by the mod-$k$-preserving classes. \ The mod-$k$-preserving classes are
partially ordered by divisibility, which means, for example, that the lattice
is not planar.\footnote{For consider the graph with the integers 2, 3, 4, 5,
6, 7, 8, 9, 10, 12, 14, 15, 18, 20, 21, 24, and 28 as its vertices,\ and with
an edge between each pair whose ratio is a prime. \ One can check that this
graph contains $K_{3,3}$ as a minor.} \ While there are infinite descending
chains in the lattice, there is no infinite ascending chain. \ This means
that, if we start from some reversible gate class and then add new gates that
extend its power, we must terminate after finitely many steps with the class
of all reversible transformations.

In Appendix \ref{LOOSE}, we will prove that if we allow loose generation, then
the only change to Theorem \ref{main}\ is that every $\mathcal{C}%
+\operatorname*{NOTNOT}$\ class collapses with the corresponding
$\mathcal{C}+\operatorname*{NOT}$ class.

\section{Consequences of the Classification\label{CONSEQ}}

To illustrate the power of the classification theorem, in this section we use
it to prove four general implications for reversible computation. \ While
these implications are easy to prove with the classification in hand, we do
not know how to prove any of them without it.

\subsection{Nature of the Classes\label{NATURE}}

Here is one immediate (though already non-obvious) corollary of Theorem
\ref{main}.

\begin{corollary}
\label{finitecor}Every reversible gate class $\mathcal{C}$\ is finitely
generated: that is, there exists a finite set $S$ such that $\mathcal{C}%
=\left\langle S\right\rangle $.
\end{corollary}

Indeed, we have something stronger.

\begin{corollary}
\label{singlegate}Every reversible gate class $\mathcal{C}$\ is generated by a
single gate $G\in\mathcal{C}$.
\end{corollary}

\begin{proof}
This is immediate for all the classes listed in Theorem \ref{main}, except the
ones involving $\operatorname*{NOT}$\ or $\operatorname*{NOTNOT}$\ gates.
\ For classes of the form $\mathcal{C}=\left\langle G,\operatorname*{NOT}%
\right\rangle $ or $\mathcal{C}=\left\langle G,\operatorname*{NOTNOT}%
\right\rangle $, we just need a single gate $G^{\prime}$ that is clearly
generated by $\mathcal{C}$, and clearly \textit{not} generated by a smaller
class. \ We can then appeal to Theorem \ref{main} to assert that $G^{\prime}$
\textit{must} generate $\mathcal{C}$. \ For each of the relevant
$G$'s---namely, $\operatorname*{Fredkin}$, $\operatorname*{CNOTNOT}$,
$\operatorname*{F}_{4}$, and $\operatorname*{T}_{6}$---one such $G^{\prime}%
$\ is the tensor product, $G\otimes\operatorname*{NOT}$\ or $G\otimes
\operatorname*{NOTNOT}$.
\end{proof}

We also wish to point out a non-obvious symmetry property that follows from
the classification theorem. \ Given an $n$-bit reversible transformation $F$,
let $F^{\ast}$, or the \textit{dual} of $F$, be $F^{\ast}\left(  x_{1}\ldots
x_{n}\right)  :=\overline{F\left(  \overline{x_{1}\ldots x_{n}}\right)  }$.
\ The dual can be thought of as $F$\ with the roles of $0$ and $1$
interchanged: for example, $\operatorname*{Toffoli}^{\ast}\left(  xyz\right)
$\ flips $z$\ if and only if $x=y=0$. \ Also, call a gate $F$
\textit{self-dual} if $F^{\ast}=F$, and call a reversible gate class
$\mathcal{C}$ \textit{dual-closed} if $F^{\ast}\in\mathcal{C}$\ whenever
$F\in\mathcal{C}$. \ Then:

\begin{corollary}
\label{dualclosed}Every reversible gate class $\mathcal{C}$\ is dual-closed.
\end{corollary}

\begin{proof}
This is obvious for all the classes listed in Theorem \ref{main}\ that include
a $\operatorname*{NOT}$\ or $\operatorname*{NOTNOT}$\ gate. \ For the others,
we simply need to consider the classes one by one: the notions of
\textquotedblleft conservative,\textquotedblright\ \textquotedblleft
mod-$k$-respecting,\textquotedblright\ and \textquotedblleft mod-$k$%
-preserving\textquotedblright\ are manifestly the same after we interchange
$0$ and $1$. \ This is less manifest for the notion of \textquotedblleft
orthogonal,\textquotedblright\ but one can check that $\operatorname*{T}%
\nolimits_{k}$\ and $\operatorname*{F}_{k}$\ are self-dual for all even $k$.
\end{proof}

\subsection{Linear-Time Algorithm\label{ALGSEC}}

If one wanted, one could interpret this entire paper as addressing a
straightforward \textit{algorithms} problem: namely, the \textsc{RevGen}%
\ problem defined in Section \ref{COMPLEXITY}, where we are given as input a
set of reversible gates $G_{1},\ldots,G_{K}$, as well as a target reversible
transformation $H$, and we want to know whether the $G_{i}$'s generate $H$.
\ From that perspective, our contribution is to reduce the known upper bound
on the complexity of \textsc{RevGen}: from recursively-enumerable (!), or
triply-exponential time if we use Je\v{r}\'{a}bek's recent clone/coclone
duality for reversible gates \cite{jerabek}, all the way down to linear time.

\begin{theorem}
\label{alg}There is a linear-time algorithm for \textsc{RevGen}.
\end{theorem}

\begin{proof}
It suffices to give a linear-time algorithm that takes as input the truth
table of a single reversible transformation $G:\left\{  0,1\right\}
^{n}\rightarrow\left\{  0,1\right\}  ^{n}$, and that decides which class it
generates. \ For we can then compute $\left\langle G_{1},\ldots,G_{K}%
\right\rangle $ by taking the least upper bound of $\left\langle
G_{1}\right\rangle ,\ldots,\left\langle G_{K}\right\rangle $, and can also
solve the membership problem by checking whether%
\[
\left\langle G_{1},\ldots,G_{K}\right\rangle =\left\langle G_{1},\ldots
,G_{K},H\right\rangle .
\]
The algorithm is as follows: first, make a single pass through $G$'s truth
table, in order to answer the following two questions.

\begin{itemize}
\item Is $G$ affine, and if so, what is its matrix representation, $G\left(
x\right)  =Ax\oplus b$?

\item What is $W\left(  G\right)  :=\left\{  \left\vert G\left(  x\right)
\right\vert -\left\vert x\right\vert :x\in\left\{  0,1\right\}  ^{n}\right\}
$?
\end{itemize}

In any reasonable RAM model, both questions can easily be answered in
$O\left(  n2^{n}\right)  $ time, which is the number of bits in $G$'s truth table.

If $G$ is non-affine, then Theorem \ref{main}\ implies that we can determine
$\left\langle G\right\rangle $\ from $W\left(  G\right)  $\ alone. \ If
$G$\ is affine, then Theorem \ref{main}\ implies we can determine
$\left\langle G\right\rangle $\ from $\left(  A,b\right)  $ alone, though it
is also convenient to use $W\left(  G\right)  $. \ We need to take the gcd of
the numbers in $W\left(  G\right)  $, check whether $A$\ is orthogonal, etc.,
but the time needed for these operations is only $\operatorname*{poly}\left(
n\right)  $, which is negligible compared to the input size of $n2^{n}$.
\end{proof}

We have implemented the algorithm described in Theorem \ref{alg},\ and Java
code is available for download \cite{schaefer:code}.

\subsection{Compression of Reversible Circuits\label{COMPRESS}}

We now state a \textquotedblleft complexity-theoretic\textquotedblright%
\ consequence of Theorem \ref{main}.

\begin{theorem}
\label{compressthm}Let $R$ be a reversible circuit, over any gate set $S$,
that maps $\left\{  0,1\right\}  ^{n}$ to $\left\{  0,1\right\}  ^{n}$, using
an unlimited number of gates and ancilla bits. \ Then there is another
reversible circuit, over the same gate set $S$, that applies the same
transformation as $R$ does, and that uses only $2^{n}\operatorname*{poly}%
(n)$\ gates and $O(1)$\ ancilla bits.\footnote{Here the big-$O$'s suppress
constant factors that depend on the gate set in question.}
\end{theorem}

\begin{proof}
If $S$\ is one of the gate sets listed in Theorem \ref{main}, then this
follows immediately by examining the reversible circuit constructions in
Section \ref{CONSTRUC}, for each class in the classification. \ Building, in
relevant parts, on results by others \cite{shende,benorcleve}, we will take
care in Section \ref{CONSTRUC} to ensure that each non-affine circuit
construction uses at most $2^{n}\operatorname*{poly}(n)$\ gates and
$O(1)$\ ancilla bits, while each affine construction uses at most $O(n^{2})$
gates and $O(1)$\ ancilla bits (most actually use no ancilla bits).

Now suppose $S$ is \textit{not} one of the sets listed in Theorem \ref{main},
but some other set that generates one of the listed classes. \ So for example,
suppose $\left\langle S\right\rangle =\left\langle \operatorname*{Fredkin}%
,\operatorname*{NOT}\right\rangle $. \ Even then, we know that $S$
\textit{generates} $\operatorname*{Fredkin}$\ and $\operatorname*{NOT}$, and
the number of gates and ancillas needed to do so is just some constant,
independent of $n$. \ Furthermore, each time we need a
$\operatorname*{Fredkin}$\ or $\operatorname*{NOT}$,\ we can reuse the same
ancilla bits, by the assumption that those bits are returned to their original
states. \ So we can simply simulate the appropriate circuit construction from
Section \ref{CONSTRUC}, using only a constant factor more gates and $O\left(
1\right)  $\ more ancilla bits than the original construction.
\end{proof}

As we said in Section \ref{COMPLEXITY}, without the classification theorem, it
is not obvious how to prove \textit{any upper bound whatsoever} on the number
of gates or ancillas, for arbitrary gate sets $S$. \ Of course, any circuit
that uses $T$ gates also uses at most $O\left(  T\right)  $\ ancillas; and
conversely, any circuit that uses $M$ ancillas needs at most $\left(
2^{n+M}\right)  !$ gates, for counting reasons. \ But the best upper bounds on
either quantity that follow from clone theory and the ideal membership problem
appear to have the form $\exp\left(  \exp\left(  \exp\left(  \exp\left(
n\right)  \right)  \right)  \right)  $.

A constant number of ancilla bits \textit{is} sometimes needed, and not only
for the trivial reasons that our gates might act on more than $n$ bits,\ or
only (e.g.)\ be able to map $0^{n}$\ to $0^{n}$ if no ancillas are available.

\begin{proposition}
[Toffoli \cite{toffoli}]\label{needancilla}If no ancillas are allowed, then
there exist reversible transformations of $\left\{  0,1\right\}  ^{n}$\ that
cannot be generated by any sequence of reversible gates on $n-1$\ bits or fewer.
\end{proposition}

\begin{proof}
For all $k\geq1$, any $\left(  n-k\right)  $-bit gate induces an even
permutation of $\left\{  0,1\right\}  ^{n}$---since each cycle is repeated
$2^{k}$\ times, once for every setting of the $k$ bits on which the gate
doesn't act. \ But there are also odd permutations of $\left\{  0,1\right\}
^{n}$.
\end{proof}

It is also easy to show, using a Shannon counting argument, that there exist
$n$-bit reversible transformations that require $\Omega\left(  2^{n}\right)
$\ gates to implement, and $n$-bit affine transformations that require
$\Omega\left(  n^{2}/\log n\right)  $\ gates. \ Thus the bounds in Theorem
\ref{compressthm}\ on the number of gates $T$\ are, for each class, off from
the optimal bounds only by $\operatorname*{polylog}T$ factors.

\subsection{Encoded Universality\label{ENCODED}}

If we only care about which Boolean functions $f:\left\{  0,1\right\}
^{n}\rightarrow\left\{  0,1\right\}  $\ can be computed, and are completely
uninterested in what garbage is output along with $f$, then it is not hard to
see that all reversible gate sets fall into three classes. \ Namely,
non-affine gate sets (such as $\operatorname*{Toffoli}$ and
$\operatorname*{Fredkin}$) can compute all Boolean functions;\footnote{This
was proven by Lloyd \cite{lloyd:gate}, as well as by Kerntopf et
al.\ \cite{kerntopf}\ and De Vos and Storme \cite{devos}; we include a proof
for completeness in Section \ref{GARBSEC}.} non-degenerate affine gate sets
(such as $\operatorname*{CNOT}$ and $\operatorname*{CNOTNOT}$) can compute all
affine functions; and degenerate gate sets (such as $\operatorname*{NOT}$ and
$\operatorname*{NOTNOT}$) can compute only $1$-bit functions. \ However, the
classification theorem lets us make a more interesting statement. \ Recall the
notion of \textit{encoded universality} from Section \ref{ALTGEN}, which
demands that every reversible transformation (or every affine transformation)
be implementable without garbage, once $0$\ and $1$\ are \textquotedblleft
encoded\textquotedblright\ by longer strings $\alpha\left(  0\right)  $\ and
$\alpha\left(  1\right)  $ respectively.

\begin{theorem}
\label{encodethm}Besides the trivial, $\operatorname{NOT}$,\ and
$\operatorname{NOTNOT}$\ classes, every reversible gate class supports encoded
universality if non-affine, or affine encoded universality if affine.
\end{theorem}

\begin{proof}
For $\left\langle \operatorname{Fredkin}\right\rangle $, and for all the
non-affine classes above $\left\langle \operatorname{Fredkin}\right\rangle $,
we use the so-called \textquotedblleft dual-rail encoding,\textquotedblright%
\ where $0$ is encoded by $01$\ and $1$\ is encoded by $10$. \ Given three
encoded bits, $x\overline{x}y\overline{y}z\overline{z}$, we can simulate a
$\operatorname{Fredkin}$ gate by applying one $\operatorname{Fredkin}$ to
$xyz$ and another to $x\overline{y}\overline{z}$, and can also simulate a
$\operatorname{CNOT}$ by applying a $\operatorname{Fredkin}$\ to
$xy\overline{y}$. \ But $\operatorname{Fredkin}+\operatorname{CNOT}%
$\ generates everything.

The dual-rail encoding also works for simulating all affine transformations
using an $\operatorname{F}_{4}$ gate. \ For note that%
\begin{align*}
\operatorname{F}_{4}\left(  xy\overline{y}1\right)   &  =\left(
1,\overline{x\oplus y},x\oplus y,x\right) \\
&  =\left(  x,x\oplus y,\overline{x\oplus y},1\right)  ,
\end{align*}
where we used that we can permute bits for free. \ So given two encoded bits,
$x\overline{x}y\overline{y}$, we can simulate a $\operatorname{CNOT}$\ from
$x$\ to $y$\ by applying $\operatorname{F}_{4}$\ to $x$, $y$, $\overline{y}$,
and one ancilla bit initialized to $1$.

For $\left\langle \operatorname{CNOTNOT}\right\rangle $, we use a repetition
encoding, where $0$ is encoded by $00$ and $1$ is encoded by $11$. \ Given two
encoded bits, $xxyy$, we can simulate a $\operatorname{CNOT}$ from $x$ to
$y$\ by applying a $\operatorname{CNOTNOT}$ from either copy of $x$\ to both
copies of $y$. \ This lets us perform all affine transformations on the
encoded subspace.

The repetition encoding also works for $\left\langle \operatorname{T}%
_{4}\right\rangle $. \ For notice that%
\begin{align*}
\operatorname{T}_{4}\left(  xyy0\right)   &  =\left(  0,x\oplus y,x\oplus
y,x\right) \\
&  =\left(  x,x\oplus y,x\oplus y,0\right)  .
\end{align*}
Thus, to simulate a $\operatorname{CNOT}$ from $x$\ to $y$, we use one copy of
$x$, both copies of $y$, and one ancilla bit initialized to $0$.

Finally, for $\left\langle \operatorname{T}_{6}\right\rangle $, we encode $0$
by $0011$\ and $1$\ by $1100$. \ Notice that%
\begin{align*}
\operatorname{T}_{6}\left(  xyy\overline{y}\overline{y}0\right)   &  =\left(
0,x\oplus y,x\oplus y,\overline{x\oplus y},\overline{x\oplus y},x\right) \\
&  =\left(  x,x\oplus y,x\oplus y,\overline{x\oplus y},\overline{x\oplus
y},0\right)  .
\end{align*}
So given two encoded bits, $xx\overline{x}\overline{x}yy\overline{y}%
\overline{y}$, we can simulate a $\operatorname{CNOT}$ from $x$\ to $y$ by
using one copy of $x$, all four copies of $y$ and\ $\overline{y}$, and one
ancilla bit initialized to $0$.
\end{proof}

In the proof of Theorem \ref{encodethm}, notice that, every time we simulated
$\operatorname{Fredkin}\left(  xyz\right)  $\ or $\operatorname{CNOT}\left(
xy\right)  $, we had to examine only a single bit in the encoding of the
control bit $x$. \ Thus, Theorem \ref{encodethm}\ actually yields a stronger
consequence: that given an ordinary, unencoded input string $x_{1}\ldots
x_{n}$, we can use any non-degenerate reversible gate first to
\textit{translate} $x$ into its encoded version $\alpha\left(  x_{1}\right)
\ldots\alpha\left(  x_{n}\right)  $, and then to perform arbitrary
transformations or affine transformations on the encoding.

\section{Structure of the Proof\label{STRUCTURE}}

The proof of Theorem \ref{main} naturally divides into four components.
\ First, we need to verify that all the gates mentioned in the theorem really
do satisfy the invariants that they are claimed to satisfy---and as a
consequence, that any reversible transformation they generate also satisfies
the invariants. \ This is completely routine.

Second, we need to verify that all pairs of classes that Theorem
\ref{main}\ says are distinct, \textit{are} distinct. \ We handle this in
Theorem \ref{distinct}\ below (there are only a few non-obvious cases).

Third, we need to verify that the \textquotedblleft gate
definition\textquotedblright\ of each class coincides with its
\textquotedblleft invariant definition\textquotedblright---i.e., that each
gate really does generate all reversible transformations that satisfy its
associated invariant. \ For example, we need to show that
$\operatorname*{Fredkin}$\ generates all conservative transformations, that
$\operatorname*{C}_{k}$\ generates all transformations that preserve Hamming
weight mod $k$, and that $\operatorname*{T}_{4}$\ generates all orthogonal
linear transformations. \ Many of these results are already known, but for
completeness, we prove all of them in Section \ref{CONSTRUC}, by giving
explicit constructions of reversible circuits.\footnote{The upshot of the
Galois connection for clones \cite{jerabek} is that, if we could prove that a
list of invariants for a given gate set $S$ was the \textit{complete} list of
invariants satisfied by $S$, then this second part of the proof would be
unnecessary: it would follow automatically that $S$ generates all reversible
transformations that satisfy the invariants. \ But this begs the question: how
do we prove that a list of invariants for $S$\ is complete? \ In each case,
the easiest way we could find to do this, was just by explicitly describing
circuits of $S$-gates to generate all transformations that satisfy the stated
invariants.}

Finally, we need to show that there are no \textit{additional} reversible gate
classes, besides the ones listed in Theorem \ref{main}. \ This is by far the
most interesting part, and occupies the majority of the paper. \ The
organization is as follows:

\begin{itemize}
\item In Section \ref{HAMMING},\ we collect numerous results about what
reversible transformations can and cannot do to Hamming weights mod $k$\ and
inner products mod $k$, in both the affine and the non-affine cases; these
results are then drawn on in the rest of the paper. \ (Some of them are even
used for the circuit constructions in Section \ref{CONSTRUC}.)

\item In Section \ref{NONAFFINE}, we complete the classification of all
non-affine gate sets. \ In Section \ref{ABOVEFREDKIN}, we show that the only
classes that contain a $\operatorname*{Fredkin}$\ gate are $\left\langle
\operatorname*{Fredkin}\right\rangle $\ itself,$\ \left\langle
\operatorname*{Fredkin},\operatorname*{NOTNOT}\right\rangle $, $\left\langle
\operatorname*{Fredkin},\operatorname*{NOT}\right\rangle $, $\left\langle
\operatorname*{C}_{k}\right\rangle $ for $k\geq3$, and $\left\langle
\operatorname*{Toffoli}\right\rangle $. \ Next, in Section \ref{CONSERV}, we
show that every nontrivial conservative gate generates
$\operatorname*{Fredkin}$. \ Then, in Section \ref{MOD},\ we build on the
result of Section \ref{MOD}\ to\ show that every non-affine gate set generates
$\operatorname*{Fredkin}$.

\item In Section \ref{AFFINE}, we complete the classification of all affine
gate sets. \ For simplicity, we start with \textit{linear} gate sets only.
\ In Section \ref{SWAMPLAND}, we show that every nontrivial mod-$4$-preserving
linear gate generates $\operatorname{T}_{6}$, and that every nontrivial,
\textit{non}-mod-$4$-preserving orthogonal gate generates $\operatorname{T}%
_{4}$. \ Next, in Section \ref{TOCNOTNOT}, we show that every non-orthogonal
linear gate generates $\operatorname{CNOTNOT}$. \ Then, in Section
\ref{TOCNOT}, we show that every non-parity-preserving linear gate generates
$\operatorname{CNOT}$. \ Since $\operatorname{CNOT}$\ generates all linear
transformations, completes the classification of linear gate sets. \ Finally,
in Section \ref{NOT}, we \textquotedblleft put back the affine
part,\textquotedblright\ showing that it can lead to only $8$ additional
classes besides the linear classes $\left\langle \varnothing\right\rangle $,
$\left\langle \operatorname{T}_{6}\right\rangle $, $\left\langle
\operatorname{T}_{4}\right\rangle $, $\left\langle \operatorname{CNOTNOT}%
\right\rangle $, and $\left\langle \operatorname{CNOT}\right\rangle $.
\end{itemize}

\begin{theorem}
\label{distinct}All pairs of classes asserted to be distinct by Theorem
\ref{main}, are distinct.
\end{theorem}

\begin{proof}
In each case, one just needs to observe that the gate that generates a given
class A, satisfies some invariant violated by the gate that generates another
class B. \ (Here we are using the \textquotedblleft gate
definitions\textquotedblright\ of the classes, which will be proven equivalent
to the invariant definitions in Section \ref{CONSTRUC}.) \ So for example,
$\left\langle \operatorname*{Fredkin}\right\rangle $\ cannot contain
$\operatorname*{CNOT}$\ because $\operatorname*{Fredkin}$\ is conservative;
conversely, $\left\langle \operatorname*{CNOT}\right\rangle $\ cannot contain
$\operatorname*{Fredkin}$\ because $\operatorname*{CNOT}$\ is affine.

The only tricky classes are those involving $\operatorname*{NOT}$\ and
$\operatorname*{NOTNOT}$ gates:\ indeed, these classes \textit{do} sometimes
coincide, as noted in Theorem \ref{main}. \ However, in all cases where the
classes are distinct, their distinctness is witnessed by the following invariants:

\begin{itemize}
\item $\left\langle \operatorname*{Fredkin},\operatorname*{NOT}\right\rangle $
and $\left\langle \operatorname*{Fredkin},\operatorname*{NOTNOT}\right\rangle
$\ are conservative in their linear part.

\item $\left\langle \operatorname*{CNOTNOT},\operatorname*{NOT}\right\rangle $
is parity-preserving in its linear part.

\item $\left\langle \operatorname*{F}_{4},\operatorname*{NOT}\right\rangle
=\left\langle \operatorname*{T}_{4},\operatorname*{NOT}\right\rangle $ and
$\left\langle \operatorname*{F}_{4},\operatorname*{NOTNOT}\right\rangle
=\left\langle \operatorname*{T}_{4},\operatorname*{NOTNOT}\right\rangle $ are
orthogonal in their linear part (isometries).

\item $\left\langle \operatorname*{T}_{6},\operatorname*{NOT}\right\rangle $
and $\left\langle \operatorname*{T}_{6},\operatorname*{NOTNOT}\right\rangle
$\ are orthogonal and mod-$4$-preserving in their linear part.
\end{itemize}

As a final remark, even if a reversible transformation is implemented with the
help of ancilla bits, as long as the ancilla bits start and end in the same
state $a_{1}\ldots a_{k}$, they have no effect on any of the invariants
discussed above, and for that reason are irrelevant.
\end{proof}

\section{Hamming Weights and Inner Products\label{HAMMING}}

The purpose of this section is to collect various mathematical results about
what a reversible transformation $G:\left\{  0,1\right\}  ^{n}\rightarrow
\left\{  0,1\right\}  ^{n}$ can and cannot do to the Hamming weight of its
input, or to the inner product of two inputs. \ That is, we study the possible
relationships that can hold between $\left\vert x\right\vert $\ and
$\left\vert G\left(  x\right)  \right\vert $, or between $x\cdot y$\ and
$G\left(  x\right)  \cdot G\left(  y\right)  $ (especially modulo various
positive integers $k$). \ Not only are these results used heavily in the rest
of the classification, but some of them might be of independent interest.

\subsection{Ruling Out Mod-Shifters\label{MODSHIFTERS}}

Call a reversible transformation a \textit{mod-shifter} if it always shifts
the Hamming weight mod $k$ of its input string by some fixed, nonzero amount.
\ When $k=2$, clearly mod-shifters exist: indeed, the humble
$\operatorname*{NOT}$\ gate satisfies $\left\vert \operatorname*{NOT}\left(
x\right)  \right\vert \equiv\left\vert x\right\vert +1\left(
\operatorname{mod}2\right)  $\ for all $x\in\left\{  0,1\right\}  $, and
likewise for any other parity-flipping gate. \ However, we now show that
$k=2$\ is the \textit{only} possibility: mod-shifters do not exist for any
larger $k$.

\begin{theorem}
\label{noshifter}There are no mod-shifters for $k\geq3$. \ In other words: let
$G$ be a reversible transformation on $n$-bit strings, and suppose%
\[
\left\vert G\left(  x\right)  \right\vert \equiv\left\vert x\right\vert
+j\left(  \operatorname{mod}k\right)
\]
for all $x\in\left\{  0,1\right\}  ^{n}$. \ Then either $j=0$\ or $k=2$.
\end{theorem}

\begin{proof}
Suppose the above equation holds for all $x$. \ Then introducing a new complex
variable $z$, we have%
\[
z^{\left\vert G\left(  x\right)  \right\vert }\equiv z^{\left\vert
x\right\vert +j}\left(  \operatorname{mod}\left(  z^{k}-1\right)  \right)
\]
(since working mod $z^{k}-1$\ is equivalent to setting $z^{k}=1$). \ Since the
above is true for all $x$,%
\begin{equation}
\sum_{x\in\left\{  0,1\right\}  ^{n}}z^{\left\vert G\left(  x\right)
\right\vert }\equiv\sum_{x\in\left\{  0,1\right\}  ^{n}}z^{\left\vert
x\right\vert }z^{j}\left(  \operatorname{mod}\left(  z^{k}-1\right)  \right)
. \label{damn}%
\end{equation}
By reversibility, we have%
\[
\sum_{x\in\left\{  0,1\right\}  ^{n}}z^{\left\vert G\left(  x\right)
\right\vert }=\sum_{x\in\left\{  0,1\right\}  ^{n}}z^{\left\vert x\right\vert
}=\left(  z+1\right)  ^{n}.
\]
Therefore equation (\ref{damn}) simplifies to%
\[
\left(  z+1\right)  ^{n}\left(  z^{j}-1\right)  \equiv0\left(
\operatorname{mod}\left(  z^{k}-1\right)  \right)  .
\]
Now, since $z^{k}-1$\ has no repeated roots, it can divide $\left(
z+1\right)  ^{n}\left(  z^{j}-1\right)  $\ only if it divides $\left(
z+1\right)  \left(  z^{j}-1\right)  $. \ For this we need either $j=0$,
causing $z^{j}-1=0$, or else $j=k-1$ (from degree considerations). \ But it is
easily checked that the equality%
\[
z^{k}-1=\left(  z+1\right)  \left(  z^{k-1}-1\right)
\]
holds only if $k=2$.
\end{proof}

In Appendix \ref{ALTPROOF}, we provide an alternative proof of Theorem
\ref{noshifter}, using linear algebra. \ The alternative proof is longer, but
perhaps less mysterious.

\subsection{\label{IPMODK}Inner Products Mod \texorpdfstring{$k$}{k}}

We have seen that there exist \textit{orthogonal} gates (such as the
$\operatorname*{T}_{k}$\ gates), which preserve inner products mod $2$. \ In
this section, we first show that no reversible gate that changes Hamming
weights can preserve inner products mod $k$ for any $k\geq3$. \ We then
observe that, if a reversible gate is orthogonal, then it must be linear, and
we give necessary and conditions for orthogonality.

\begin{theorem}
\label{iponly2}Let $G$ be a non-conservative $n$-bit reversible gate, and
suppose%
\[
G\left(  x\right)  \cdot G\left(  y\right)  \equiv x\cdot y\left(
\operatorname{mod}k\right)
\]
for all $x,y\in\left\{  0,1\right\}  ^{n}$. \ Then $k=2$.
\end{theorem}

\begin{proof}
As in the proof of Theorem~\ref{noshifter}, we promote the congruence to a
congruence over complex polynomials:
\[
z^{G(x)\cdot G(y)}\equiv z^{x\cdot y}\left(  \operatorname{mod}\left(
z^{k}-1\right)  \right)
\]
Fix a string $x\in\{0,1\}^{n}$ such that $\left\vert G(x)\right\vert
>\left\vert x\right\vert $,\ which must exist because $G$ is non-conservative.
\ Then sum the congruence over all $y$:
\[
\sum_{y\in\left\{  0,1\right\}  ^{n}}z^{G\left(  x\right)  \cdot G\left(
y\right)  }\equiv\sum_{y\in\left\{  0,1\right\}  ^{n}}z^{x\cdot y}\left(
\operatorname{mod}\left(  z^{k}-1\right)  \right)  .
\]
The summation on the right simplifies as follows.
\[
\sum_{y\in\left\{  0,1\right\}  ^{n}}z^{x\cdot y}=\sum_{y\in\left\{
0,1\right\}  ^{n}}\prod_{i=1}^{n}z^{x_{i}y_{i}}=\prod_{i=1}^{n}\sum_{y_{i}%
\in\left\{  0,1\right\}  }z^{x_{i}y_{i}}=\prod_{i=1}^{n}\left(  1+z^{x_{i}%
}\right)  =\left(  1+z\right)  ^{\left\vert x\right\vert }2^{n-\left\vert
x\right\vert }.
\]
Similarly,%
\[
\sum_{y\in\left\{  0,1\right\}  ^{n}}z^{G\left(  x\right)  \cdot G\left(
y\right)  }=\left(  1+z\right)  ^{\left\vert G\left(  x\right)  \right\vert
}2^{n-\left\vert G\left(  x\right)  \right\vert },
\]
since summing over all $y$ is the same as summing over all $G\left(  y\right)
$. So we have
\begin{align*}
\left(  1+z\right)  ^{\left\vert G\left(  x\right)  \right\vert }%
2^{n-\left\vert G\left(  x\right)  \right\vert }  &  \equiv\left(  1+z\right)
^{\left\vert x\right\vert }2^{n-\left\vert x\right\vert }\left(
\operatorname{mod}\left(  z^{k}-1\right)  \right)  ,\\
0  &  \equiv(1+z)^{\left\vert x\right\vert }2^{n-\left\vert G\left(  x\right)
\right\vert }\left(  2^{\left\vert G\left(  x\right)  \right\vert -\left\vert
x\right\vert }-\left(  1+z\right)  ^{\left\vert G\left(  x\right)  \right\vert
-\left\vert x\right\vert }\right)  \left(  \operatorname{mod}\left(
z^{k}-1\right)  \right)  ,
\end{align*}
or equivalently, letting%
\[
p\left(  x\right)  :=2^{\left\vert G\left(  x\right)  \right\vert -\left\vert
x\right\vert }-\left(  1+z\right)  ^{\left\vert G\left(  x\right)  \right\vert
-\left\vert x\right\vert },
\]
we find that $z^{k}-1$\ divides $(1+z)^{\left\vert x\right\vert }p\left(
x\right)  $ as a polynomial. \ Now, the roots of $z^{k}-1$ lie on the unit
circle centered at $0$. \ Meanwhile, the roots of $p\left(  x\right)  $ lie on
the circle in the complex plane of radius $2$, centered at $-1$. \ The only
point of intersection of these two circles is $z=1$, so that is the only root
of $z^{k}-1$ that can be covered by $p\left(  x\right)  $. \ On the other
hand, clearly $z=-1$ is the only root of $(1+z)^{\left\vert x\right\vert }$.
\ Hence, the only roots of $z^{k}-1$ are $1$ and $-1$, so we conclude that
$k=2$.
\end{proof}

We now study reversible transformations that preserve inner products mod $2$.

\begin{lemma}
\label{orthoglin}Every orthogonal gate $G$ is linear.
\end{lemma}

\begin{proof}
Suppose%
\[
G\left(  x\right)  \cdot G\left(  y\right)  \equiv x\cdot y\left(
\operatorname{mod}2\right)  .
\]
Then for all $x,y,z$,%
\begin{align*}
G\left(  x\oplus y\right)  \cdot G\left(  z\right)   &  \equiv\left(  x\oplus
y\right)  \cdot z\\
&  \equiv x\cdot z+y\cdot z\\
&  \equiv G\left(  x\right)  \cdot G\left(  z\right)  +G\left(  y\right)
\cdot G\left(  z\right) \\
&  \equiv\left(  G\left(  x\right)  \oplus G\left(  y\right)  \right)  \cdot
G\left(  z\right)  \left(  \operatorname{mod}2\right)  .
\end{align*}
But if the above holds for all possible $z$, then%
\[
G\left(  x\oplus y\right)  \equiv G\left(  x\right)  \oplus G\left(  y\right)
\left(  \operatorname{mod}2\right)  .
\]

\end{proof}

Theorem \ref{iponly2}\ and Lemma \ref{orthoglin}\ have the following corollary.

\begin{corollary}
\label{nonlinip}Let $G$ be any non-conservative, nonlinear gate. \ Then for
all $k\geq2$, there exist inputs $x,y$\ such that%
\[
G\left(  x\right)  \cdot G\left(  y\right)  \not \equiv x\cdot y\left(
\operatorname{mod}k\right)  .
\]

\end{corollary}

Also:

\begin{lemma}
\label{ipcond}A linear transformation $G(x)=Ax$ is orthogonal if and only if
$A^{T}A$\ is the identity: that is, if $A$'s column vectors\ satisfy
$\left\vert v_{i}\right\vert \equiv1\left(  \operatorname{mod}2\right)  $\ for
all $i$ and $v_{i}\cdot v_{j}\equiv0\left(  \operatorname{mod}2\right)  $\ for
all $i\neq j$.
\end{lemma}

\begin{proof}
This is just the standard characterization of orthogonal matrices; that we are
working over $\mathbb{F}_{2}$\ is irrelevant. \ First, if $G$ preserves inner
products mod $2$ then for all $i\neq j$,%
\begin{align*}
1  &  \equiv e_{i}\cdot e_{i}\equiv\left(  Ae_{i}\right)  \cdot\left(
Ae_{i}\right)  \equiv\left\vert v_{i}\right\vert \left(  \operatorname{mod}%
2\right)  ,\\
0  &  \equiv e_{i}\cdot e_{j}\equiv\left(  Ae_{i}\right)  \cdot\left(
Ae_{j}\right)  \equiv v_{i}\cdot v_{j}\left(  \operatorname{mod}2\right)  .
\end{align*}
Second, if $G$ satisfies the conditions then%
\[
Ax\cdot Ay\equiv(Ax)^{T}Ay\equiv x^{T}(A^{T}A)y\equiv x^{T}y\equiv x\cdot
y\left(  \operatorname{mod}2\right)  .
\]

\end{proof}

\subsection{Why Mod 2 and Mod 4 Are Special\label{MOD2MOD4}}

Recall that $\wedge$\ denotes bitwise AND. \ We first need an
\textquotedblleft inclusion/exclusion formula\textquotedblright\ for the
Hamming weight of a bitwise sum of strings.

\begin{lemma}
\label{inclex}For all $v_{1},\ldots,v_{t}\in\left\{  0,1\right\}  ^{n}$, we
have%
\[
\left\vert v_{1}\oplus\cdots\oplus v_{t}\right\vert =\sum_{\emptyset\subset
S\subseteq\lbrack t]}(-2)^{\left\vert S\right\vert -1}\left\vert
\bigwedge_{i\in S}v_{i}\right\vert .
\]

\end{lemma}

\begin{proof}
It suffices to prove the lemma for $n=1$, since in the general case we are
just summing over all $i\in\left[  n\right]  $. \ Thus, assume without loss of
generality that $v_{1}=\cdots=v_{t}=1$. \ Our problem then reduces to proving
the following identity:
\[
\sum_{i=1}^{t}(-2)^{i-1}\binom{t}{i}=\left\{
\begin{tabular}
[c]{ll}%
$0$ & if $t$ is even\\
$1$ & if $t$ is odd,
\end{tabular}
\ \ \ \right.
\]
which follows straightforwardly from the binomial theorem.
\end{proof}

\begin{lemma}
\label{affcontriv}No nontrivial affine gate $G$\ is conservative.
\end{lemma}

\begin{proof}
Let $G\left(  x\right)  =Ax\oplus b$; then $\left\vert G\left(  0^{n}\right)
\right\vert =\left\vert 0^{n}\right\vert =0$\ implies $b=0^{n}$. \ Likewise,
$\left\vert G\left(  e_{i}\right)  \right\vert =\left\vert e_{i}\right\vert
=1$\ for all $i$\ implies that $A$ is a permutation matrix. \ But then $G$ is trivial.
\end{proof}

\begin{theorem}
\label{k2or4}If $G$ is a nontrivial linear gate that preserves Hamming weight
mod $k$, then either $k=2$ or $k=4$.
\end{theorem}

\begin{proof}
For all $x,y$, we have%
\begin{align*}
\left\vert x\right\vert +\left\vert y\right\vert -2\left(  x\cdot y\right)
&  \equiv\left\vert x\oplus y\right\vert \\
&  \equiv\left\vert G\left(  x\oplus y\right)  \right\vert \\
&  \equiv\left\vert G\left(  x\right)  \oplus G\left(  y\right)  \right\vert
\\
&  \equiv\left\vert G\left(  x\right)  \right\vert +\left\vert G\left(
y\right)  \right\vert -2\left(  G\left(  x\right)  \cdot G\left(  y\right)
\right) \\
&  \equiv\left\vert x\right\vert +\left\vert y\right\vert -2\left(  G\left(
x\right)  \cdot G\left(  y\right)  \right)  \left(  \operatorname{mod}%
k\right)  ,
\end{align*}
where the first and fourth lines used Lemma \ref{inclex}, the second and fifth
lines used that $G$ is mod-$k$-preserving, and the third line used linearity.
\ Hence%
\begin{equation}
2\left(  x\cdot y\right)  \equiv2\left(  G\left(  x\right)  \cdot G\left(
y\right)  \right)  \left(  \operatorname{mod}k\right)  . \label{wow}%
\end{equation}
If $k$ is odd, then equation (\ref{wow}) implies%
\[
x\cdot y\equiv G\left(  x\right)  \cdot G\left(  y\right)  \left(
\operatorname{mod}k\right)  .
\]
But since $G$ is nontrivial and linear, Lemma \ref{affcontriv}\ says that $G$
is non-conservative. \ So by Theorem \ref{iponly2}, the above equation cannot
be satisfied for any odd $k\geq3$. \ Likewise, if $k$ is even, then
(\ref{wow}) implies%
\[
x\cdot y\equiv G\left(  x\right)  \cdot G\left(  y\right)  \left(
\operatorname{mod}\frac{k}{2}\right)  .
\]
Again by Theorem \ref{iponly2}, the above can be satisfied only if $k=2$\ or
$k=4$.
\end{proof}

In Appendix \ref{ALTPROOF}, we provide an alternative proof of Theorem
\ref{k2or4}, one that does not rely on Theorem \ref{iponly2}.

\begin{theorem}
\label{affine4cond} Let $\{o_{i}\}_{i=1}^{n}$ be an orthonormal basis over
$\mathbb{F}_{2}$. An affine transformation $F(x)=Ax\oplus b$ is mod-$4$%
-preserving if and only if $\left\vert b\right\vert \equiv0\left(
\operatorname{mod}4\right)  $,\ and the vectors $v_{i} := A o_{i}$ satisfy
$\left\vert v_{i}\right\vert +2\left(  v_{i}\cdot b\right)  \equiv|o_{i}|
\left(  \operatorname{mod}4\right)  $\ for all $i$ and $v_{i}\cdot v_{j}%
\equiv0\left(  \operatorname{mod}2\right)  $\ for all $i\neq j$.
\end{theorem}

\begin{proof}
First, if $F$\ is mod-$4$-preserving, then%
\[
0\equiv\left\vert F\left(  0^{n}\right)  \right\vert \equiv\left\vert
A0^{n}\oplus b\right\vert \equiv\left\vert b\right\vert \left(
\operatorname{mod}4\right)  ,
\]
and hence%
\[
|o_{i}| \equiv\left\vert F\left(  o_{i}\right)  \right\vert \equiv\left\vert
Ao_{i}\oplus b\right\vert \equiv\left\vert v_{i}\oplus b\right\vert
\equiv\left\vert v_{i}\right\vert +\left\vert b\right\vert -2\left(
v_{i}\cdot b\right)  \equiv\left\vert v_{i}\right\vert +2\left(  v_{i}\cdot
b\right)  \left(  \operatorname{mod}4\right)
\]
for all $i$, and hence%
\begin{align*}
|o_{i} + o_{j}|  &  \equiv\left\vert F\left(  o_{i}\oplus o_{j}\right)
\right\vert \equiv\left\vert v_{i}\oplus v_{j}\oplus b\right\vert
\equiv\left\vert v_{i}\right\vert +\left\vert v_{j}\right\vert +\left\vert
b\right\vert -2\left(  v_{i}\cdot v_{j}\right)  -2\left(  v_{i}\cdot b\right)
-2\left(  v_{j}\cdot b\right)  +4\left\vert v_{i}\wedge v_{j}\wedge
b\right\vert \\
&  \equiv\left\vert v_{i}\right\vert +\left\vert v_{j}\right\vert +2\left(
v_{i}\cdot v_{j}\right)  +2\left(  v_{i}\cdot b\right)  +2\left(  v_{j}\cdot
b\right)  \left(  \operatorname{mod}4\right) \\
&  \equiv|o_{i}| + |o_{j}| +2\left(  v_{i}\cdot v_{j}\right)  \left(
\operatorname{mod}4\right)
\end{align*}
for all $i\neq j$, from which we conclude that $v_{i}\cdot v_{j}\equiv0\left(
\operatorname{mod}2\right)  $.

Second, if $F$ satisfies the conditions,\ then for any $x=\sum_{i\in S}o_{i}$,
we have%
\begin{align*}
\left\vert F\left(  x\right)  \right\vert  &  =\left\vert b\oplus\sum_{i\in
S}v_{i}\right\vert \\
&  =\left\vert b\right\vert +\sum_{i\in S}\left\vert v_{i}\right\vert
-2\sum_{i\in S}\left(  b\cdot v_{i}\right)  -2\sum_{i\in S~<~j\in S}\left(
v_{i}\cdot v_{j}\right)  +4(\cdots)\\
&  \equiv\sum_{i\in S}\left\vert v_{i}\right\vert -2\left(  b\cdot
v_{i}\right) \\
&  \equiv\sum_{i\in S} |o_{i} |\left(  \operatorname{mod}4\right)  ,
\end{align*}
where the second line follows from Lemma~\ref{inclex}. Furthermore, we have
that
\begin{align*}
|x| = \left|  \sum_{i\in S} o_{i}\right|  = \sum_{i\in S} |o_{i}| - 2
\sum_{i\in S < j \in S} (o_{i} \cdot o_{j}) + 4 (\ldots) \equiv\sum_{i\in S}
|o_{i}| \left(  \operatorname{mod}4\right)  ,
\end{align*}
where the last equality follows from the fact that $\{o_{i}\}_{i=1}^{n}$ is an
orthonormal basis. Therefore, we conclude that $|F(x)| \equiv|x| \left(
\operatorname{mod}4\right)  $.
\end{proof}

We note two corollaries of Theorem \ref{affine4cond}\ for later use.

\begin{corollary}
\label{mod4orthog}Any linear transformation $A\in\mathbb{F}_{2}^{n\times n}$
that preserves Hamming weight mod $4$ is also orthogonal.
\end{corollary}

\begin{corollary}
\label{mod4cond}An orthogonal transformation $A\in\mathbb{F}_{2}^{n\times n}%
$\ preserves Hamming weight mod $4$ if and only if all of its columns have
Hamming weight $1$\ mod $4$.
\end{corollary}

\section{Reversible Circuit Constructions\label{CONSTRUC}}

In this section, we show that all the classes of reversible transformations
listed in Theorem \ref{main}, are indeed generated by the gates that we
claimed, by giving explicit synthesis procedures. \ In order to justify
Theorem~\ref{compressthm}, we also verify that\ in each case, only
$O(1)$\ ancilla bits are needed, even though this constraint makes some of the
constructions more complicated than otherwise.

Many of our constructions---those for $\operatorname*{Toffoli}$\ and
$\operatorname*{CNOT}$, for example---have appeared in various forms in the
reversible computing literature, and are included here only for completeness.
\ Others---those for $\operatorname{C}_{k}$ and $\operatorname*{F}_{4}$, for
example---are new as far as we know, but not hard.

\subsection{Non-Affine Circuits\label{NONAFFCIRC}}

We start with the non-affine classes: $\left\langle \operatorname*{Toffoli}%
\right\rangle $, $\left\langle \operatorname*{Fredkin}\right\rangle $,
$\left\langle \operatorname*{Fredkin},\operatorname*{C}_{k}\right\rangle $,
and $\left\langle \operatorname*{Fredkin},\operatorname*{NOT}\right\rangle $.

\begin{theorem}
[variants in \cite{toffoli,shende}]\label{toffolicirc}$\operatorname*{Toffoli}%
$ generates all reversible transformations on $n$ bits, using only $2$ ancilla
bits.\footnote{Notice that we need at least $2$ so that we can generate
$\operatorname*{CNOT}$ and $\operatorname*{NOT}$ using $\operatorname{Toffoli}%
$.}
\end{theorem}

\begin{proof}
Any reversible transformation $F:\left\{  0,1\right\}  ^{n}\rightarrow\left\{
0,1\right\}  ^{n}$ is a permutation of $n$-bit strings, and any permutation
can be written as a product of transpositions. \ So it suffices to show how to
use $\operatorname*{Toffoli}$\ gates to implement an arbitrary transposition
$\sigma_{y,z}$: that is, a mapping that sends $y=y_{1}\ldots y_{n}$\ to
$z=z_{1}\ldots z_{n}$ and $z$ to $y$, and all other $n$-bit strings to themselves.

Given any $n$-bit string $w$, let us define $w$-$\operatorname*{CNOT}$ to be
the $\left(  n+1\right)  $-bit gate that flips its last bit if its first
$n$\ bits are equal to $w$, and that does nothing otherwise. \ (Thus, the
$\operatorname*{Toffoli}$\ gate is $11$-$\operatorname*{CNOT}$, while
$\operatorname*{CNOT}$\ itself is $1$-$\operatorname*{CNOT}$.) \ Given
$y$-$\operatorname*{CNOT}$\ and $z$-$\operatorname*{CNOT}$\ gates, we can
implement the transposition $\sigma_{y,z}$\ as follows on input $x$:

\begin{enumerate}
\item Initialize an ancilla bit, $a=1$.

\item Apply $y$-$\operatorname*{CNOT}\left(  x,a\right)  $.

\item Apply $z$-$\operatorname*{CNOT}\left(  x,a\right)  $.

\item Apply $\operatorname*{NOT}$\ gates to all $x_{i}$'s such that $y_{i}\neq
z_{i}$.

\item For each $i$\ such that $y_{i}\neq z_{i}$, apply $\operatorname*{CNOT}%
\left(  a,x_{i}\right)  $.

\item Apply $z$-$\operatorname*{CNOT}\left(  x,a\right)  $.

\item Apply $y$-$\operatorname*{CNOT}\left(  x,a\right)  $.
\end{enumerate}

Thus, all that remains is to implement $w$-$\operatorname*{CNOT}$ using
$\operatorname*{Toffoli}$. \ Observe that we can simulate any $w$%
-$\operatorname*{CNOT}$\ using $1^{n}$-$\operatorname*{CNOT}$, by negating
certain input bits (namely, those for which $w_{i}=0$) before and after we
apply the $1^{n}$-$\operatorname*{CNOT}$. \ An example of the transposition
$\sigma_{011,101}$ is given in Figure~\ref{fig:transposition}.

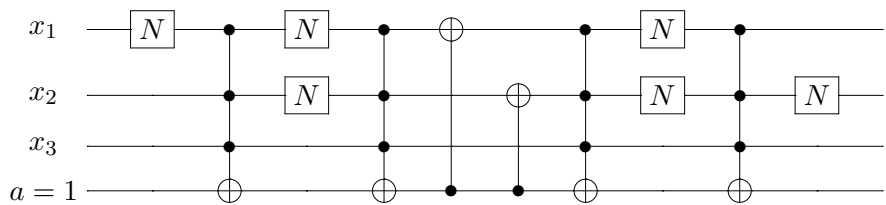
\begin{figure}[h]
\centering
\begin{minipage}[c]{.9\textwidth}
\Qcircuit @C=1.5em @R=1em {
x_1 & & \gate{N} & \ctrl{1} & \gate{N} & \ctrl{1}  & \targ & \qw & \ctrl{1} & \gate{N} & \ctrl{1} & \qw & \qw \\
x_2 & & \qw & \ctrl{1} & \gate{N} & \ctrl{1}  & \qw & \targ & \ctrl{1} & \gate{N} & \ctrl{1} & \gate{N} & \qw \\
x_3 & & \qw & \ctrl{1} & \qw & \ctrl{1}  & \qw & \qw & \ctrl{1} & \qw & \ctrl{1} & \qw & \qw \\
a=1 & & \qw & \targ & \qw & \targ  & \ctrl{-3} & \ctrl{-2} & \targ & \qw & \targ & \qw & \qw \\
}
\end{minipage}
\caption{Generating the transposition $\sigma_{011,101}$}%
\label{fig:transposition}%
\end{figure}

So it suffices to implement $1^{n}$-$\operatorname*{CNOT}$, with control bits
$x_{1}\ldots x_{n}$ and target bit $y$. \ The base case is $n=2$, which we
implement directly using $\operatorname*{Toffoli}$. \ For $n\geq3$, we do the following.

\begin{itemize}
\item Let $a$ be an ancilla.

\item Apply $1^{\left\lceil n/2\right\rceil }$-$\operatorname*{CNOT}\left(
x_{1}\ldots x_{\left\lceil n/2\right\rceil },a\right)  $.

\item Apply $1^{\left\lfloor n/2\right\rfloor +1}$-$\operatorname*{CNOT}%
\left(  x_{\left\lceil n/2\right\rceil +1}\ldots x_{n},a,y\right)  $.

\item Apply $1^{\left\lceil n/2\right\rceil }$-$\operatorname*{CNOT}\left(
x_{1}\ldots x_{\left\lceil n/2\right\rceil },a\right)  $.

\item Apply $1^{\left\lfloor n/2\right\rfloor +1}$-$\operatorname*{CNOT}%
\left(  x_{\left\lceil n/2\right\rceil +1}\ldots x_{n},a,y\right)  $.
\end{itemize}

The crucial point is that this construction works whether the ancilla is
initially $0$ or $1$. \ In other words, we can use \textit{any} bit which is
not one of the inputs, instead of a new ancilla. \ For instance, we can have
one bit dedicated for use in $1^{n}$-$\operatorname*{CNOT}$ gates, which we
use in the recursive applications of $1^{\lceil n/2\rceil}$%
-$\operatorname*{CNOT}$ and $1^{\lfloor n/2\rfloor+1}$-$\operatorname*{CNOT}$,
and the recursive applications within them, and so on.\footnote{The number of
$\operatorname*{Toffoli}$\ gates $T(n)$ needed to implement a\ $1^{n}%
$-$\operatorname*{CNOT}$ (which dominates the cost of a transposition) by this
recursive scheme, is given by the recurrence
\[
T(n)=2T(1+\lfloor n/2\rfloor)+2T(\lceil n/2\rceil)
\]
which we solve to obtain $T\left(  n\right)  =O\left(  n^{2}\right)  $.}

Carefully inspecting the above proof shows that $O\left(  n^{2}2^{n}\right)  $
gates and $3$ ancilla bits suffice to generate any transformation. \ Notice
the main reason we need two of the three ancillas is to apply the
$\operatorname*{NOT}$ gate while the ancilla $a$ is active. \ Case analysis
shows that any circuit constructible from $\operatorname*{NOT}$,
$\operatorname*{CNOT}$, and $\operatorname*{Toffoli}$ is equivalent to a
circuit of $\operatorname*{NOT}$ gates followed by a circuit of
$\operatorname*{CNOT}$ and $\operatorname*{Toffoli}$ gates. \ For example, see
Figure~\ref{fig:not_push}. \ This at most triples the size of the circuit.
\ Therefore, we can construct a circuit that uses only two ancilla bits: apply
the recursive construction, push the $\operatorname*{NOT}$ gates to the front,
and use two ancilla bits to generate the $\operatorname*{NOT}$ gates. \ The
recursive construction itself uses one ancilla bit, plus one more to implement
$\operatorname*{CNOT}$.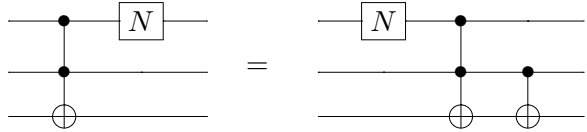
\begin{figure}[h]
\centering
\begin{minipage}[c]{.4\textwidth}
\Qcircuit @C=1.5em @R=1em {
&\ctrl{1} & \gate{N} & \qw \\
&\ctrl{1} & \qw  & \qw \\
&\targ & \qw & \qw \\
}
\end{minipage}\hspace{15px}$=$\hspace{15px}
\begin{minipage}[c]{.4\textwidth}
\Qcircuit @C=1.5em @R=1em {
& \gate{N} & \ctrl{1} & \qw & \qw \\
& \qw &\ctrl{1}  &\ctrl{1} & \qw \\
& \qw &\targ  & \targ & \qw \\
}
\end{minipage}
\caption{Example of equivalent $\operatorname*{Toffoli}$ circuit with
$\operatorname*{NOT}$ gates pushed to the front}%
\label{fig:not_push}%
\end{figure}
\end{proof}

The particular construction above was inspired by a result of Ben-Or and Cleve
\cite{benorcleve}, in which they compute algebraic formulas in a straight-line
computation model using a constant number of registers. \ We note that Toffoli
\cite{toffoli}\ proved a version of Theorem \ref{toffolicirc}, but with
$O\left(  n\right)  $\ ancilla bits rather than $O\left(  1\right)  $. \ More
recently, Shende et al.\ \cite{shende}\ gave a slightly more complicated
construction which uses only $1$ ancilla bit, and also gives explicit bounds
on the number of Toffoli gates required based on the number of fixed points of
the permutation. \ Recall that at least $1$ ancilla bit is needed by
Proposition \ref{needancilla}.

Next, let $\operatorname*{CCSWAP}$, or Controlled-Controlled-SWAP, be the
$4$-bit gate that swaps its last two bits if its first two bits are both $1$,
and otherwise does nothing.

\begin{proposition}
\label{ccswap}$\operatorname*{Fredkin}$ generates $\operatorname*{CCSWAP}$.
\end{proposition}

\begin{proof}
Let $a$ be an ancilla bit initialized to $0$. \ We implement
$\operatorname*{CCSWAP}\left(  x,y,z,w\right)  $\ by applying
$\operatorname*{Fredkin}\left(  x,y,a\right)  $, then $\operatorname*{Fredkin}%
\left(  a,z,w\right)  $, then again $\operatorname*{Fredkin}\left(
x,y,a\right)  $.
\end{proof}

We can now prove an analogue of Theorem \ref{toffolicirc}\ for conservative transformations.

\begin{theorem}
\label{fredkincirc}$\operatorname*{Fredkin}$ generates all conservative
transformations on $n$ bits, using only $5$ ancilla bits.
\end{theorem}

\begin{proof}
In this proof, we will use the \textit{dual-rail representation}, in which
$0$\ is encoded as $01$\ and $1$ is encoded as $10$. \ We will also use
Proposition \ref{ccswap}, that $\operatorname*{Fredkin}$\ generates
$\operatorname*{CCSWAP}$.

As in Theorem \ref{toffolicirc}, we can decompose any reversible
transformation $F:\left\{  0,1\right\}  ^{n}\rightarrow\left\{  0,1\right\}
^{n}$ as a product of transpositions $\sigma_{y,z}$. \ In this case, each
$\sigma_{y,z}$\ transposes two $n$-bit strings $y=y_{1}\ldots y_{n}$\ and
$z=z_{1}\ldots z_{n}$\ of the same Hamming weight.

Given any $n$-bit string $w$, let us define $w$-$\operatorname*{CSWAP}$ to be
the $\left(  n+2\right)  $-bit gate that swaps its last two bits if its first
$n$\ bits are equal to $w$, and that does nothing otherwise. \ (Thus,
$\operatorname*{Fredkin}$\ is $1$-$\operatorname*{CSWAP}$, while
$\operatorname*{CCSWAP}$\ is $11$-$\operatorname*{CSWAP}$.) \ Then given
$y$-$\operatorname*{CSWAP}$\ and $z$-$\operatorname*{CSWAP}$\ gates, where
$\left\vert y\right\vert =\left\vert z\right\vert $, as well as
$\operatorname*{CCSWAP}$\ gates, we can implement the transposition
$\sigma_{y,z}$\ on input $x$\ as follows:

\begin{enumerate}
\item Initialize two ancilla bits (comprising three dual-rail registers) to
$a\overline{a}=01$.

\item Apply $y$-$\operatorname*{CSWAP}\left(  x_{1}\ldots x_{n},a,\overline
{a}\right)  $.

\item Apply $z$-$\operatorname*{CSWAP}\left(  x_{1}\ldots x_{n},a,\overline
{a}\right)  $.

\item Pair off the $i$'s\ such that $y_{i}=1$\ and $z_{i}=0$, with the equally
many $j$'s such that $z_{j}=1$\ and $y_{j}=0$. \ For each such $\left(
i,j\right)  $ pair, apply $\operatorname*{Fredkin}\left(  a,x_{i}%
,x_{j}\right)  $.

\item Apply $z$-$\operatorname*{CSWAP}\left(  x_{1}\ldots x_{n},a,\overline
{a}\right)  $.

\item Apply $y$-$\operatorname*{CSWAP}\left(  x_{1}\ldots x_{n},a,\overline
{a}\right)  $.
\end{enumerate}

The logic here is exactly the same as in the construction of
transpositions\ in Theorem \ref{toffolicirc}; the only difference is that now
we need to conserve Hamming weight.

All that remains is to implement $w$-$\operatorname*{CSWAP}$ using
$\operatorname*{CCSWAP}$. \ First let us show how to implement $1^{n}%
$-$\operatorname*{CSWAP}$ using $\operatorname*{CCSWAP}$. \ Once again, we do
so using a recursive construction. \ For the base case, $n=2$, we just use
$\operatorname*{CCSWAP}$. \ For $n\geq3$, we implement $1^{n}$%
-$\operatorname*{CSWAP}\left(  x_{1},\ldots,x_{n},y,z\right)  $ as follows:

\begin{itemize}
\item Initialize two ancilla bits (comprising one dual-rail register) to
$a\overline{a}=01$.

\item Apply $1^{\left\lceil n/2\right\rceil }$-$\operatorname*{CSWAP}\left(
x_{1}\ldots x_{\left\lceil n/2\right\rceil },a,\overline{a}\right)  $.

\item Apply $1^{\left\lfloor n/2\right\rfloor +1}$-$\operatorname*{CSWAP}%
\left(  x_{\left\lceil n/2\right\rceil +1}\ldots x_{n},a,y,z\right)  $.

\item Apply $1^{\left\lceil n/2\right\rceil }$-$\operatorname*{CSWAP}\left(
x_{1}\ldots x_{\left\lceil n/2\right\rceil },a,\overline{a}\right)  $.

\item Apply $1^{\left\lfloor n/2\right\rfloor +1}$-$\operatorname*{CSWAP}%
\left(  x_{\left\lceil n/2\right\rceil +1}\ldots x_{n},a,y,z\right)  $.
\end{itemize}

The logic is the same as in the construction of $1^{n}$-$\operatorname*{CNOT}%
$\ in Theorem \ref{toffolicirc} except we now use $2$ ancilla bits for the
dual rail representation.

Finally, we need to implement $w$-$\operatorname*{CSWAP}\left(  x_{1}\ldots
x_{n},y,z\right)  $, for arbitrary $w$, using $1^{n}$-$\operatorname*{CSWAP}$.
\ We do so by first constructing $w$-$\operatorname*{CSWAP}$ from
$\operatorname*{NOT}$ gates and $1^{n}$-$\operatorname*{CSWAP}$. \ Observe
that we only use the $\operatorname*{NOT}$ gate on the control bits of the
$\operatorname*{Fredkin}$ gates used during the construction so the
equivalence given in Figure~\ref{fig:fred_equiv} holds (i.e., we can remove
the $\operatorname*{NOT}$ gates).

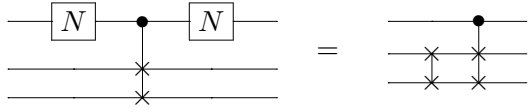
\begin{figure}[h]
\centering
\begin{minipage}[c]{.4\textwidth}
\Qcircuit @C=1.5em @R=1em {
&\gate{N} & \ctrl{2} & \gate{N} & \qw\\
&\qw & \qswap  & \qw & \qw \\
&\qw & \qswap & \qw & \qw\\
}
\end{minipage}\hspace{15px}$=$\hspace{15px}
\begin{minipage}[c]{.4\textwidth}
\Qcircuit @C=1.5em @R=1em {
&\qw & \ctrl{2} & \qw \\
&\qswap & \qswap  & \qw \\
&\qswap \qwx & \qswap & \qw \\
}
\end{minipage}
\caption{Removing $\operatorname*{NOT}$ gates from the
$\operatorname*{Fredkin}$ circuit}%
\label{fig:fred_equiv}%
\end{figure}Hence, we can build a $w$-$\operatorname*{CSWAP}$ out of
$\operatorname*{CCSWAP}$s using only $5$\ ancilla bits: $1$ for
$\operatorname*{CCSWAP}$, $2$ for the $1^{n}$-$\operatorname*{CSWAP}$, and $2$
for a transposition.
\end{proof}

We note that, before the above construction was found by the authors,
unpublished and independent work by Siyao Xu and Qian Yu first showed that
$O(1)$ ancillas were sufficient.

In \cite{fredkintoffoli}, the result that $\operatorname*{Fredkin}$\ generates
all conservative transformations is stated without proof, and credited to
B.\ Silver. \ We do not know how many ancilla bits Silver's construction used.

Next, we prove an analogue of Theorem \ref{toffolicirc}\ for the
mod-$k$-respecting transformations, for all $k\geq2$. \ First, let
$\operatorname{CC}_{k}$, or Controlled-$\operatorname{C}_{k}$, be the $\left(
k+1\right)  $-bit gate that applies $\operatorname{C}_{k}$\ to the final $k$
bits if the first bit is $1$, and does nothing if the first bit is $0$.

\begin{proposition}
\label{ckcck}$\operatorname*{Fredkin}+\operatorname{C}_{k}$ generates
$\operatorname{CC}_{k}$, using $2$ ancilla bits, for all $k\geq2$.
\end{proposition}

\begin{proof}
To implement $\operatorname{CC}_{k}$\ on input bits $x,y_{1}\ldots y_{k}$, we
do the following:

\begin{enumerate}
\item Initialize ancilla bits $a,b$ to $0,1$\ respectively.

\item Use $\operatorname*{Fredkin}$\ gates and swaps to swap $y_{1},y_{2}%
$\ with $a,b$, conditioned on $x=0$.\footnote{In more detail, use
$\operatorname*{Fredkin}$\ gates to swap $y_{1},y_{2}$\ with $a,b$,
conditioned on $x=1$. \ Then swap $y_{1},y_{2}$\ with $a,b$\ unconditionally.}

\item Apply $\operatorname{C}_{k}$\ to $y_{1}\ldots y_{k}$.

\item Repeat step 2.
\end{enumerate}
\end{proof}

Then we have the following.

\begin{theorem}
\label{ckcirc}$\operatorname*{Fredkin}+\operatorname{CC}_{k}$ generates all
mod-$k$-preserving transformations, for $k \geq1$, using only $5$ ancilla bits.
\end{theorem}

\begin{proof}
The proof is exactly the same as that of Theorem \ref{fredkincirc}, except for
one detail. \ Namely, let $y$ and $z$ be $n$-bit strings such that $\left\vert
y\right\vert \equiv\left\vert z\right\vert \left(  \operatorname{mod}k\right)
$. \ Then in the construction of the transposition $\sigma_{y,z}$\ from
$y$-$\operatorname*{CSWAP}$\ and $z$-$\operatorname*{CSWAP}$ gates, when we
are applying step 5, it is possible that $\left\vert y\right\vert -\left\vert
z\right\vert $\ is some nonzero multiple of $k$, say $qk$. \ If so, then we
can no longer pair off each $i$ such that $y_{i}=1$\ and $z_{i}=0$ with a
unique $j$ such that $z_{j}=1$\ and $y_{j}=0$: after we have done that, there
will remain a surplus of `$1$' bits of size $qk$, either in $y$ or in $z$, as
well as a matching surplus of `$0$' bits of size $qk$\ in the other
string.\ \ However, we can get rid of both surpluses using $q$ applications of
a $\operatorname*{CC}_{k}$\ gate (which we have by Proposition \ref{ckcck}),
with $c$\ as the control bit.
\end{proof}

As a special case of Theorem \ref{ckcirc}, note that $\operatorname*{Fredkin}%
+\operatorname{CC}_{1}=\operatorname*{Fredkin}+\operatorname{CNOT}$\ generates
all mod-$1$-preserving transformations---or in other words, all transformations.

We just need one additional fact about the $\operatorname{C}_{k}$\ gate.

\begin{proposition}
\label{ckfredkin}$\operatorname{C}_{k}$ generates $\operatorname*{Fredkin}$,
using $k-2$ ancilla bits, for all $k\geq3$.
\end{proposition}

\begin{proof}
Let $a_{1}\ldots a_{k-2}$\ be ancilla bits initially set to $1$. \ Then to
implement $\operatorname{Fredkin}$ on input bits $x,y,z$, we apply:%
\begin{align*}
&  \operatorname{C}_{k}\left(  x,y,a_{1}\ldots a_{k-2}\right)  ,\\
&  \operatorname{C}_{k}\left(  x,z,a_{1}\ldots a_{k-2}\right)  ,\\
&  \operatorname{C}_{k}\left(  x,y,a_{1}\ldots a_{k-2}\right)  .
\end{align*}

\end{proof}

Combining Theorem \ref{ckcirc}\ with Proposition \ref{ckfredkin}\ now yields
the following.

\begin{corollary}
\label{ckcor}$\operatorname{C}_{k}$ generates all mod-$k$-preserving
transformations for $k\geq3$, using only $k + 3$ ancilla bits.
\end{corollary}

Finally, we handle the parity-flipping case.

\begin{proposition}
\label{cc2}$\operatorname*{Fredkin}+\operatorname*{NOTNOT}$\ (and hence,
$\operatorname*{Fredkin}+\operatorname{NOT}$) generates $\operatorname{CC}%
_{2}$.
\end{proposition}

\begin{proof}
This follows from Proposition \ref{ckcck}, if we recall that $\operatorname{C}%
_{2}$\ is equivalent to $\operatorname*{NOTNOT}$\ up to an irrelevant bit-swap.
\end{proof}

\begin{theorem}
\label{fredkinnotcirc}$\operatorname*{Fredkin}+\operatorname{NOT}$ generates
all parity-respecting transformations on $n$ bits, using only $6$ ancilla bits.
\end{theorem}

\begin{proof}
Let $F$ be any parity-flipping transformation on $n$ bits. \ Then
$F\otimes\operatorname{NOT}$\ is an $\left(  n+1\right)  $-bit
parity-preserving transformation. \ So by Theorem \ref{ckcirc}, we can
implement $F\otimes\operatorname{NOT}$\ using $\operatorname*{Fredkin}%
+\operatorname{CC}_{2}$ (and we have $\operatorname{CC}_{2}$\ by Proposition
\ref{cc2}). \ We can then apply a $\operatorname{NOT}$\ gate to the $\left(
n+1\right)  ^{st}$\ bit to get $F$ alone.
\end{proof}

One consequence of Theorem \ref{fredkinnotcirc}\ is that every parity-flipping
transformation can be constructed from parity-preserving gates and exactly one
$\operatorname{NOT}$ gate.

\subsection{Affine Circuits\label{AFFCIRC}}

It is well-known that $\operatorname*{CNOT}$\ is a \textquotedblleft universal
affine gate\textquotedblright:

\begin{theorem}
\label{cnotcirc}$\operatorname*{CNOT}$ generates all affine transformations,
with only $1$ ancilla bit (or $0$ for linear transformations).
\end{theorem}

\begin{proof}
Let $G\left(  x\right)  =Ax\oplus b$ be the affine transformation that we want
to implement, for some invertible matrix $A\in\mathbb{F}_{2}^{n\times n}$.
\ Then given an input $x=x_{1}\ldots x_{n}$, we first use
$\operatorname*{CNOT}$\ gates (at most $\binom{n}{2}$\ of them) to map $x$\ to
$Ax$, by reversing the sequence of row-operations that maps $A$\ to the
identity matrix in Gaussian elimination. \ Finally, if $b=b_{1}\ldots b_{n}$
is nonzero, then for each $i$ such that $b_{i}=1$, we apply a
$\operatorname*{CNOT}$\ from an ancilla bit that is initialized to $1$.
\end{proof}

A simple modification of Theorem \ref{cnotcirc}\ handles the parity-preserving case.

\begin{theorem}
\label{cnotnotcirc}$\operatorname{CNOTNOT}$ generates all parity-preserving
affine transformations with only $1$ ancilla bit (or $0$ for linear transformations).
\end{theorem}

\begin{proof}
Let $G\left(  x\right)  =Ax\oplus b$ be a parity-preserving affine
transformation. \ We first construct the linear part of\ $G$\ using Gaussian
elimination. \ Notice that for $G$ to be parity-preserving, the columns
$v_{i}$ of $A$ must satisfy $\left\vert v_{i}\right\vert \equiv1\left(
\operatorname{mod}2\right)  $ for all $i$. \ For this reason, the
row-elimination steps come in pairs, so we can implement them using
$\operatorname{CNOTNOT}$. \ Notice further that since $G$\ is
parity-preserving, we must have $\left\vert b\right\vert \equiv0\left(
\operatorname{mod}2\right)  $. \ So we can map $Ax$\ to $Ax\oplus b$, by using
$\operatorname{CNOTNOT}$\ gates plus one ancilla bit set to $1$ to simulate
$\operatorname{NOTNOT}$\ gates.
\end{proof}

Likewise (though, strictly speaking, we will not need this for the proof of
Theorem \ref{main}):

\begin{theorem}
\label{cnnpnotcirc}$\operatorname{CNOTNOT}+\operatorname{NOT}$ generates all
parity-respecting affine transformations using no ancilla bits.
\end{theorem}

\begin{proof}
Use Theorem \ref{cnotnotcirc} to map $x$\ to $Ax$, and then use
$\operatorname{NOT}$\ gates to map $Ax$\ to $Ax\oplus b$.
\end{proof}

We now move on to the more complicated cases of $\left\langle
\operatorname*{F}_{4}\right\rangle $, $\left\langle \operatorname*{T}%
_{6}\right\rangle $, and $\left\langle \operatorname*{T}_{4}\right\rangle $.

\begin{theorem}
\label{f4circ}$\operatorname*{F}_{4}$ generates all mod-$4$-preserving affine
transformations using no ancilla bits.
\end{theorem}

\begin{proof}
Let $F\left(  x\right)  =Ax\oplus b$ be an $n$-bit affine transformation,
$n\geq2$, that preserves Hamming weight mod $4$. \ Using $\operatorname*{F}%
_{4}$\ gates, we will show how to map $F\left(  x\right)  =y_{1}\ldots y_{n}%
$\ to $x=x_{1}\ldots x_{n}$. \ Reversing the construction then yields the
desired map from $x$\ to $F\left(  x\right)  $.

At any point in time, each $y_{j}$\ is some affine function of the $x_{i}$'s.
\ We say that $x_{i}$\ \textquotedblleft occurs in\textquotedblright\ $y_{j}$,
if $y_{j}$\ depends on $x_{i}$. \ At a high level, our procedure will consist
of the following steps, repeated up to $n-3$ times:

\begin{enumerate}
\item Find an\ $x_{i}$\ that does not occur in every $y_{j}$.

\item Manipulate the $y_{j}$'s so that $x_{i}$\ occurs in exactly \textit{one}
$y_{j}$.

\item Argue that no \textit{other} $x_{i^{\prime}}$\ can then occur in that
$y_{j}$. \ Therefore, we have recursively reduced our problem to one involving
a reversible, mod-$4$-preserving, affine function on $n-1$\ variables.
\end{enumerate}

It is not hard to see that the only mod-$4$-preserving affine functions on $3$
or fewer variables, are permutations of the bits. \ So if we can show that the
three steps above can always be carried out, then we are done.

First, since $A$ is invertible, it is not the all-$1$'s matrix, which means
that there must be an $x_{i}$\ that does not occur in every $y_{j}$.

Second, if there are at least three occurrences of $x_{i}$, then apply
$\operatorname*{F}_{4}$ to three positions in which $x_{i}$ occurs, plus one
position in which $x_{i}$ does not occur. \ The result of this is to decrease
the number of occurrences of $x_{i}$ by $2$. \ Repeat until there are at most
two occurrences of $x_{i}$. \ Since $\operatorname*{F}_{4}$\ is mod-$4$%
-preserving and affine, the resulting transformation $F^{\prime}\left(
x\right)  =A^{\prime}x+b^{\prime}$\ must still be mod-$4$-preserving and
affine, so it must still satisfy the conditions of Lemma \ref{affine4cond}.
\ In particular, no column vector of $A^{\prime}$\ can have even Hamming
weight. \ Since two occurrences of $x_{i}$\ would necessitate such a column
vector, we know that $x_{i}$ must occur only once.

Third, if $x_{i}$\ occurs only once in $F^{\prime}\left(  x\right)  $, then
the corresponding column vector $v_{i}$ has exactly one nonzero element.
\ Since $\left\vert v_{i}\right\vert =1$, we know by Lemma~\ref{affine4cond}
that $v_{i}\cdot b\equiv0\left(  \operatorname{mod}2\right)  $, which means
that $b$ has a $0$ in the position where $v_{i}$ has a $1$. \ Now consider the
row of $A^{\prime}$ that includes the nonzero entry of $v_{i}$. \ If any other
column $v_{i^{\prime}}$ is also nonzero in that row, then $v_{i}\cdot
v_{i^{\prime}}\equiv1\left(  \operatorname{mod}2\right)  $, which once again
contradicts the conditions of Lemma~\ref{affine4cond}. \ Thus, no other
$x_{i^{\prime}}$\ occurs in the same $y_{j}$\ that $x_{i}$\ occurs in.
\ Indeed no constant occurs there either, since otherwise $F^{\prime}$\ would
no longer be mod-$4$-preserving. \ So we have reduced to the $\left(
n-1\right)  \times\left(  n-1\right)  $\ case.
\end{proof}

The same argument, with slight modifications, handles $\left\langle
\operatorname*{T}_{4}\right\rangle $\ and $\left\langle \operatorname*{T}%
_{6}\right\rangle $.

\begin{theorem}
\label{t4circ}$\operatorname*{T}_{4}$ generates all orthogonal
transformations, using no ancilla bits.
\end{theorem}

\begin{proof}
The construction is identical to that of Theorem~\ref{f4circ}, except with
$\operatorname*{T}_{4}$\ instead of $\operatorname*{F}_{4}$.\ \ When reducing
the number of occurrences of $x_{i}$ to at most $2$, Lemma \ref{ipcond}%
\ assures us that $\left\vert v_{i}\right\vert \equiv1\left(
\operatorname{mod}2\right)  $.
\end{proof}

\begin{theorem}
\label{t6circ}$\operatorname*{T}_{6}$ generates all mod-$4$-preserving linear
transformations, using no ancilla bits.
\end{theorem}

\begin{proof}
The construction is identical to that of Theorem~\ref{f4circ}, except for the
following change. \ Rather than using $\operatorname*{F}_{4}$\ to reduce the
number of occurrences of some $x_{i}$\ to at most $2$, we now use
$\operatorname*{T}_{6}$ to reduce the number of occurrences of $x_{i}$\ to at
most $4$. \ (If there are $5$ or more occurrences, then $\operatorname*{T}%
_{6}$\ can always decrease the number by $4$.) \ We then appeal to
Corollary~\ref{mod4cond}, which says that $\left\vert v_{i}\right\vert
\equiv1\left(  \operatorname{mod}4\right)  $ for each $i$. \ This implies that
no $x_{i}$ can occur $2$, $3$, or $4$ times in the output vector. \ But that
can only mean that $x_{i}$\ occurs once.
\end{proof}

By Lemma \ref{orthoglin} and Corollary \ref{mod4orthog}, an equivalent way to
state Theorem \ref{t6circ}\ is that $\operatorname*{T}_{6}$ generates
all\ affine transformations that are both mod-$4$-preserving and orthogonal.

All that remains is some \textquotedblleft cleanup work\textquotedblright%
\ (which, again, is not even needed for the proof of Theorem \ref{main}).

\begin{theorem}
\label{cleanupcirc}$\operatorname*{T}_{6}+\operatorname*{NOT}$ generates all
affine transformations that are mod-$4$-preserving (and therefore orthogonal)
in their linear part.

$\operatorname*{T}_{6}+\operatorname*{NOTNOT}$ generates all parity-preserving
affine transformations that are mod-$4$-preserving (and therefore orthogonal)
in their linear part.

$\operatorname*{F}_{4}+\operatorname*{NOT}$ (or equivalently,
$\operatorname*{T}_{4}+\operatorname*{NOT}$) generates all isometries.

$\operatorname*{F}_{4}+\operatorname*{NOTNOT}$ (or equivalently,
$\operatorname*{T}_{4}+\operatorname*{NOTNOT}$) generates all
parity-preserving isometries.

$\operatorname*{NOT}$ generates all degenerate transformations.

$\operatorname*{NOTNOT}$ generates all parity-preserving degenerate transformations.

In none of these cases are any ancilla bits needed.
\end{theorem}

\begin{proof}
As in Theorem \ref{cnnpnotcirc}, we simply apply the relevant construction for
the linear part (e.g., Theorem \ref{t4circ}\ or \ref{t6circ}), then handle the
affine part using $\operatorname*{NOT}$\ or $\operatorname*{NOTNOT}$\ gates.
\end{proof}

\section{The Non-Affine Part\label{NONAFFINE}}

Our goal, in this section, is to prove that there are no non-affine classes
besides the ones listed in Theorem \ref{main}: namely, the conservative
transformations, the parity-respecting transformations,\ the mod-$k$%
-preserving transformations for $k\geq2$, and all transformations.

We will divide our analysis into two parts. \ We first show, in Section
\ref{ABOVEFREDKIN},\ that once a $\operatorname*{Fredkin}$\ gate is available,
matters become fairly simple. \ At that point, the only possibilities are
$\left\langle \operatorname*{Fredkin}\right\rangle $, $\left\langle
\operatorname*{Fredkin},\operatorname*{NOTNOT}\right\rangle $,\ $\left\langle
\operatorname*{Fredkin},\operatorname*{NOT}\right\rangle $, $\left\langle
\operatorname*{C}_{k}\right\rangle $\ for $k\geq3$, and $\left\langle
\operatorname{Toffoli}\right\rangle $. \ Then, in Sections \ref{CONSERV}\ and
\ref{MOD}, we prove the harder result that \textit{every non-affine gate
generates} $\operatorname*{Fredkin}$. \ This, in turn, is broken into three pieces:

\begin{itemize}
\item In Section \ref{GARBSEC}, we reprove a result of Lloyd \cite{lloyd:gate}%
, showing that every non-affine gate is capable of universal computation with garbage.

\item In Section \ref{CONSERV}, we show that every nontrivial conservative
gate generates $\operatorname*{Fredkin}$ (using the result of Section
\ref{GARBSEC}\ as one ingredient).

\item In Section \ref{MOD}, we build on the result of Section \ref{CONSERV},
to show that every non-affine gate\ generates $\operatorname*{Fredkin}$.
\ This requires our first use of lattices, and also draws on some of the
results about inner products and modularity obstructions from Section
\ref{HAMMING}.
\end{itemize}

Summarizing the results of this section, we will obtain the following.

\begin{theorem}
\label{nonaffinedone}Every non-affine gate set generates one of the following
classes: $\left\langle \operatorname*{Fredkin}\right\rangle $, $\left\langle
\operatorname*{C}_{k}\right\rangle $\ for some $k\geq3$,\ $\left\langle
\operatorname*{Fredkin},\operatorname*{NOTNOT}\right\rangle $, $\left\langle
\operatorname*{Fredkin},\operatorname*{NOT}\right\rangle $, or $\left\langle
\operatorname*{Toffoli}\right\rangle $.
\end{theorem}

\subsection{Above Fredkin\label{ABOVEFREDKIN}}

Our goal, in this section, is to classify all reversible gate classes
containing $\operatorname*{Fredkin}$. \ We already know from Theorem
\ref{fredkincirc}\ that $\operatorname*{Fredkin}$\ generates all conservative
transformations. \ We will prove a substantial generalization of that result.
\ First, however, we need a proposition that will also be used later in the
paper. \ Given a reversible transformation $G$, let%
\[
W\left(  G\right)  :=\left\{  \left\vert G\left(  x\right)  \right\vert
-\left\vert x\right\vert :x\in\left\{  0,1\right\}  ^{n}\right\}
\]
be the set of possible changes that $G$ can cause to the Hamming weight of its input.

\begin{proposition}
\label{whatdanielneeds}Let $G$ be any non-conservative gate. \ Then for all
integers $q$, there exists a $t$ such that $q\cdot k\left(  G\right)  \in
W\left(  G^{\otimes t}\right)  $.
\end{proposition}

\begin{proof}
Let $\gamma$\ be the gcd of the elements in $W\left(  G\right)  $. \ Then
clearly $G$\ is mod-$\gamma$-respecting. \ By Proposition \ref{chinese}, this
means that $\gamma$\ must divide $k\left(  G\right)  $.\footnote{Indeed, by
using Theorem \ref{noshifter}, one can show that $\gamma=k\left(  G\right)  $,
except in the special case that $G$\ is parity-flipping, where we have
$\gamma=1$\ and $k\left(  G\right)  =2$.}

Now by reversibility, $W\left(  G\right)  $\ has both positive and negative
elements. \ But this means that we can find any integer multiple of $\gamma
$\ in \textit{some} set of the form%
\[
W\left(  G^{\otimes t}\right)  =\left\{  w_{1}+\cdots+w_{m}:w_{1},\ldots
,w_{m}\in W\left(  G\right)  \right\}  .
\]
Therefore we can find any integer multiple of $k\left(  G\right)  $\ in some
$W\left(  G^{\otimes t}\right)  $\ as well.
\end{proof}

We can now characterize all reversible gate sets that contain
$\operatorname*{Fredkin}$.

\begin{theorem}
\label{fredking}Let $G$ be any gate.\ \ Then$\ \operatorname*{Fredkin}+G$
generates all mod-$k\left(  G\right)  $-preserving transformations (including
in the cases $k\left(  G\right)  =1$, in which case $\operatorname*{Fredkin}%
+G$\ generates all transformations, and \ $k\left(  G\right)  =\infty$, in
which case $\operatorname*{Fredkin}+G$\ generates all conservative transformations).
\end{theorem}

\begin{proof}
Let $k=k\left(  G\right)  $. \ If $k=\infty$\ then we are done by Theorem
\ref{fredkincirc}, so assume $k$ is finite. \ We will assume without loss of
generality that $G$\ is mod-$k$-preserving. \ By Theorem \ref{noshifter}, the
only other possibility is that $G$ is parity-flipping, but in that case we can
simply repeat everything below with $G\otimes G$, which is parity-preserving
and satisfies $k\left(  G\otimes G\right)  =2$, rather than with $G$ itself.

By Theorem \ref{ckcirc}, it suffices to use $\operatorname*{Fredkin}+G$\ to
generate the $\operatorname{CC}_{k}$\ gate. \ Let $H$ be the gate $G\otimes
G^{-1}$, followed by a swap of the two input registers. \ Observe that $H^{2}$
is the identity. \ Also, by Proposition \ref{kgtensor},%
\[
k\left(  H\right)  =\gcd\left(  k\left(  G\right)  ,k\left(  G^{-1}\right)
\right)  =k.
\]
So by Proposition \ref{whatdanielneeds}, there exists a positive integer $t$,
as well as inputs $y=y_{1}\ldots y_{n}$ and $z=z_{1}\ldots z_{n}$ such that
$z=H^{\otimes t}\left(  y\right)  $ (and $y=H^{\otimes t}\left(  z\right)  $,
since $\left(  H^{\otimes t}\right)  ^{2}=I$), and $\left\vert z\right\vert
=\left\vert y\right\vert +k$.

We can assume without loss of generality that $y$ has the form $0^{a}1^{b}%
$---i.e., that its bits are in sorted order. \ We would like to sort the bits
of $z$ as well. \ Notice that, since $\left\vert z\right\vert >\left\vert
y\right\vert $, there is some $i\in\left[  n\right]  $ such that $y_{i}=0$ and
$z_{i}=1$. \ So we can easily design a circuit $U$ of $\operatorname{Fredkin}$
gates, controlled by bit $i$, which reorders the bits of $z$ so that%
\[
z^{\prime}:=U\left(  z\right)  =0^{a-k}1^{b+k}%
\]
whereas $U\left(  y\right)  =y$.

Observe that $H^{\otimes t}$ has a large number of fixed points: we have
$H\left(  u,G\left(  u\right)  \right)  =\left(  u,G\left(  u\right)  \right)
$ for any $u$; hence any string of the form $u_{1},G\left(  u_{1}\right)
,\ldots,u_{t},G\left(  u_{t}\right)  $\ is a fixed point of $H^{\otimes t}$.
\ Call one of these fixed points $w$, and let $w^{\prime}:=U\left(  w\right)
$.

We now consider a circuit $R$\ that applies $U^{-1}$, followed by $H^{\otimes
t}$, followed by $U$. \ This $R$ satisfies the following identities:%
\begin{align*}
R\left(  y\right)   &  =U\left(  H^{\otimes t}\left(  U^{-1}\left(  y\right)
\right)  \right)  =U\left(  H^{\otimes t}\left(  y\right)  \right)  =U\left(
z\right)  =z^{\prime}.\\
R\left(  z^{\prime}\right)   &  =U\left(  H^{\otimes t}\left(  U^{-1}\left(
z^{\prime}\right)  \right)  \right)  =U\left(  H^{\otimes t}\left(  z\right)
\right)  =U\left(  y\right)  =y.\\
R\left(  w^{\prime}\right)   &  =U\left(  H^{\otimes t}\left(  U^{-1}\left(
w^{\prime}\right)  \right)  \right)  =U\left(  H^{\otimes t}\left(  w\right)
\right)  =U\left(  w\right)  =w^{\prime}.
\end{align*}
Using $R$, we now construct $\operatorname{CC}_{k}\left(  x_{1}\ldots x_{k},
c\right)  $. \ Let $A$ and $B$ be two $n$-bit registers, initialized to
$A:=w^{\prime}$ and $B:=0^{a-k}x_{1}\ldots x_{k}1^{b}$. \ Also, let
$q\overline{q}$\ be two ancilla bits in dual-rail representation, initialized
to $q\overline{q}=01$. \ Then to apply $\operatorname{CC}_{k}$, we do the following:

\begin{enumerate}
\item Swap $q$ with $\overline{q}$\ if and only if $x_{1}=\cdots=x_{k}$ and
$c=1$.

\item Swap $A$\ with $B$ if and only if $q=1$.

\item Apply $R$ to the $A$ register.

\item Swap $A$\ with $B$ if and only if $q=1$.

\item Swap $q$ with $\overline{q}$\ if and only if $x_{1}=\cdots=x_{k}$ and
$c=1$.
\end{enumerate}

Here each conditional swap is implemented using $\operatorname*{Fredkin}%
$\ gates; recall from Theorem \ref{fredkincirc}\ that $\operatorname*{Fredkin}%
$\ generates every conservative transformation.

It is not hard to check that the above sequence maps $x_{1}\ldots x_{k}=0^{k}%
$\ to $1^{k}$ and $x_{1}\ldots x_{k}=1^{k}$\ to $0^{k}$ if $c = 1$, otherwise
it maps the inputs to themselves. \ Furthermore, the ancilla bits are returned
to their original states in all cases, since $w^{\prime}$\ is a fixed point of
$R$. \ Therefore we have implemented $\operatorname{CC}_{k}$.
\end{proof}

Theorem \ref{fredking}\ has the following corollary.

\begin{corollary}
\label{fredqueen}Let $S$\ be any non-conservative gate set. \ Then
$\operatorname*{Fredkin}+S$ generates one of the following classes:
$\left\langle \operatorname*{Fredkin},\operatorname*{NOTNOT}\right\rangle $,
$\left\langle \operatorname*{Fredkin},\operatorname*{NOT}\right\rangle $,
$\left\langle \operatorname*{C}_{k}\right\rangle $\ for some $k\geq3$, or
$\left\langle \operatorname{Toffoli}\right\rangle $.
\end{corollary}

\begin{proof}
We know from Proposition \ref{kgtensor}\ that $S$\ generates a single gate $G$
such that $k\left(  G\right)  =k\left(  S\right)  $. \ If $k\left(  S\right)
\geq3$, then Theorem \ref{fredking} implies that $\operatorname*{Fredkin}%
+G$\ generates all $k\left(  S\right)  $-preserving\ transformations, which
equals $\left\langle \operatorname*{C}_{k\left(  S\right)  }\right\rangle
$\ by Corollary \ref{ckcor}. \ If $k\left(  S\right)  =2$\ and $S$\ is
parity-preserving, then Theorem \ref{fredking} implies that
$\operatorname*{Fredkin}+G$\ generates all parity-preserving\ transformations,
which equals $\left\langle \operatorname*{Fredkin},\operatorname*{NOTNOT}%
\right\rangle $\ by Proposition \ref{cc2}. \ If $k\left(  S\right)  =1$, then
Theorem \ref{fredking} implies that $\operatorname*{Fredkin}+G$\ generates all
transformations, which equals $\left\langle \operatorname{Toffoli}%
\right\rangle $\ by Theorem \ref{toffolicirc}.

By Theorem \ref{noshifter}, the one remaining case is that $k\left(  S\right)
=2$\ and some $G\in S$\ is parity-flipping. \ By Theorem \ref{fredking},
certainly $\operatorname*{Fredkin}+G$\ at least generates all
parity-preserving\ transformations. \ Furthermore, let $F$\ be any
parity-flipping transformation. \ Then $F\otimes G^{-1}$\ is
parity-preserving. \ So we can use $\operatorname*{Fredkin}+G$\ to implement
$F\otimes G^{-1}$, then compose with $G$ itself to get $F$. \ Therefore we
generate all parity-flipping transformations, which equals $\left\langle
\operatorname*{Fredkin},\operatorname*{NOT}\right\rangle $\ by Theorem
\ref{fredkinnotcirc}.
\end{proof}

\subsection{Computing with Garbage\label{GARBSEC}}

For completeness, in this section we reprove some lemmas first shown by Seth
Lloyd \cite{lloyd:gate} in an unpublished 1992 technical
report,\footnote{Prompted by the present work, Lloyd has recently\ posted his
1992 report to the arXiv.} and later rediscovered by Kerntopf et
al.\ \cite{kerntopf} and De Vos and Storme \cite{devos}. \ We will use these
lemmas to show the power of non-affine gates.

Recall the notion of \textit{generating with garbage} from Section
\ref{ALTGEN}.

\begin{lemma}
[\cite{lloyd:gate,devos}]\label{lloydnot}Every nontrivial reversible gate
$G$\ generates $\operatorname*{NOT}$\ with garbage.
\end{lemma}

\begin{proof}
Let $G\left(  x_{1}\ldots x_{n}\right)  =y_{1}\ldots y_{n}$\ be nontrivial,
and let $y_{i}=f_{i}\left(  x_{1}\ldots x_{n}\right)  $. \ Then it suffices to
show that at least one $f_{i}$\ is a non-monotone Boolean function. \ For if
$f_{i}$\ is non-monotone, then by definition, there exist two inputs
$x,x^{\prime}\in\left\{  0,1\right\}  ^{n}$, which are identical except that
$x_{j}=1$\ and $x_{j}^{\prime}=0$ at some bit $j$, such that $f_{i}\left(
x\right)  =0$\ and $f_{i}\left(  x^{\prime}\right)  =1$. \ But then, if we set
the other $n-1$\ bits consistent with $x$\ and $x^{\prime}$, we have
$y_{i}=\operatorname*{NOT}\left(  x_{j}\right)  $.

Thus, suppose by contradiction that every $f_{i}$\ is monotone. \ Then
reversibility clearly implies\ that $G\left(  0^{n}\right)  =0^{n}$, and that
the set of strings of Hamming weight $1$ is mapped to itself: that is, there
exists a permutation $\sigma$\ such that $G\left(  e_{j}\right)
=e_{\sigma\left(  j\right)  }$ for all $j$. \ Furthermore, by monotonicity,
for all $j\neq k$ we have $G\left(  e_{j}\oplus e_{k}\right)  \geq
e_{\sigma\left(  j\right)  }\oplus e_{\sigma\left(  k\right)  }$. \ But then
reversibility implies that $G\left(  e_{j}\oplus e_{k}\right)  $\ can only be
$e_{\sigma\left(  j\right)  }\oplus e_{\sigma\left(  k\right)  }$\ itself, and
so on inductively, so that we obtain $G\left(  x_{1}\ldots x_{n}\right)
=x_{\sigma^{-1}\left(  1\right)  }\ldots x_{\sigma^{-1}\left(  n\right)  }%
$\ for all $x\in\left\{  0,1\right\}  ^{n}$. \ But this means that $G$ is
trivial, contradiction.
\end{proof}

\begin{proposition}
[folklore]\label{folkloreprop}For all $n\geq3$, every non-affine Boolean
function on $n$ bits has a non-affine subfunction on $n-1$ bits.
\end{proposition}

\begin{proof}
Let $f:\left\{  0,1\right\}  ^{n}\rightarrow\left\{  0,1\right\}  $ be
non-affine, and let $f_{0}$\ and $f_{1}$\ be the $\left(  n-1\right)  $-bit
subfunctions obtained by restricting $f$'s first input bit to $0$ or $1$
respectively. \ If either $f_{0}$\ or $f_{1}$\ is itself non-affine, then we
are done. \ Otherwise, we have $f_{0}\left(  x\right)  =\left(  a_{0}\cdot
x\right)  \oplus b_{0}$\ and $f_{1}\left(  x\right)  =\left(  a_{1}\cdot
x\right)  \oplus b_{1}$, for some $a_{0},a_{1}\in\left\{  0,1\right\}  ^{n-1}$
and $b_{0},b_{1}\in\left\{  0,1\right\}  $. \ Notice that $f$ is non-affine if
and only if $a_{0}\neq a_{1}$. \ So there is some bit where $a_{0}$\ and
$a_{1}$\ are unequal. \ If we now remove any of the \textit{other} rightmost
$n-1$\ input bits (which must exist since $n-1\geq2$) from $f$, then we are
left with a non-affine function on $n-1$ bits.
\end{proof}

\begin{lemma}
[\cite{lloyd:gate,devos}]\label{lloydand}Every non-affine reversible gate
$G$\ generates the $2$-bit $\operatorname*{AND}$\ gate with garbage.
\end{lemma}

\begin{proof}
Certainly every non-affine gate is nontrivial, so we know from Lemma
\ref{lloydnot}\ that $G$\ generates $\operatorname*{NOT}$ with garbage. \ For
this reason, it suffices to show that $G$\ can generate \textit{some}
non-affine $2$-bit gate with garbage (since all such gates are equivalent to
$\operatorname*{AND}$\ under negating inputs and outputs). \ Let $G\left(
x_{1}\ldots x_{n}\right)  =y_{1}\ldots y_{n}$, and let $y_{i}=f_{i}\left(
x_{1}\ldots x_{n}\right)  $. \ Then some particular $f_{i}$ must be a
non-affine Boolean function. \ So it suffices to show that, by restricting
$n-2$\ of $f_{i}$'s\ input bits, we can get a non-affine function on $2$ bits.
\ But this follows by inductively applying Proposition \ref{folkloreprop}.
\end{proof}

By using Lemma \ref{lloydand}, it is possible to prove directly that the only
classes that contain a $\operatorname*{CNOT}$ gate are $\left\langle
\operatorname*{CNOT}\right\rangle $\ (i.e., all affine transformations) and
$\left\langle \operatorname*{Toffoli}\right\rangle $\ (i.e., all
transformations)---or in other words, that if $G$ is any non-affine gate, then
$\left\langle \operatorname*{CNOT},G\right\rangle =\left\langle
\operatorname*{Toffoli}\right\rangle $. \ However, we will skip this result,
since it is subsumed by our later results.

Recall that $\operatorname*{COPY}$\ is the $2$-bit partial gate that maps
$00$\ to $00$\ and $10$ to $11$.

\begin{lemma}
[\cite{lloyd:gate,kerntopf}]\label{lloydcopy}Every non-degenerate reversible
gate $G$\ generates $\operatorname*{COPY}$\ with garbage.
\end{lemma}

\begin{proof}
Certainly every non-degenerate gate is nontrivial, so we know from Lemma
\ref{lloydnot}\ that $G$\ generates $\operatorname*{NOT}$ with garbage. \ So
it suffices to show that there is some pair of inputs $x,x^{\prime}\in\left\{
0,1\right\}  ^{n}$, which differ only at a single coordinate $i$, such that
$G\left(  x\right)  $\ and $G\left(  x^{\prime}\right)  $\ have Hamming
distance at least $2$. \ For then if we set $x_{i}:=z$, and regard\ the
remaining $n-1$\ coordinates of $x$ as ancillas, we will find at least two
copies of $z$\ or $\overline{z}$\ in $G\left(  x\right)  $, which we can
convert to at least two copies\ of $z$\ using $\operatorname*{NOT}$\ gates.
\ Also, if all of the ancilla bits that receive a copy of $z$\ were initially
$1$, then we can use a $\operatorname*{NOT}$\ gate to reduce to the case where
one of them was initially $0$.

Thus, suppose by contradiction that $G\left(  x\right)  $\ and $G\left(
x^{\prime}\right)  $\ are neighbors on the Hamming cube whenever $x$\ and
$x^{\prime}$\ are neighbors. \ Then starting from $0^{n}$\ and $G\left(
0^{n}\right)  $, we find that every $G\left(  e_{i}\right)  $\ must be a
neighbor of $G\left(  0^{n}\right)  $, every $G\left(  e_{i}\oplus
e_{j}\right)  $\ must be a neighbor of $G\left(  e_{i}\right)  $\ and
$G\left(  e_{j}\right)  $, and so on, so that $G$\ is just a rotation and
reflection of $\left\{  0,1\right\}  ^{n}$. \ But that means $G$ is
degenerate, contradiction.
\end{proof}

\subsection{Conservative Generates Fredkin\label{CONSERV}}

In this section, we prove the following theorem.

\begin{theorem}
\label{consfredkin}Let $G$ be any nontrivial conservative gate. \ Then $G$
generates $\operatorname{Fredkin}$.
\end{theorem}

The proof will be slightly more complicated than necessary, but we will then
reuse parts of it in Section \ref{MOD}, when we show that every non-affine,
\textit{non}-conservative\ gate generates $\operatorname{Fredkin}$.

Given a gate $Q$, let us call $Q$\ \textit{strong quasi-Fredkin} if there
exist control strings $a,b,c,d$\ such that%
\begin{align}
Q\left(  a,01\right)   &  =\left(  a,01\right)  ,\label{q1}\\
Q\left(  b,01\right)   &  =\left(  b,10\right)  ,\label{q2}\\
Q\left(  c,00\right)   &  =\left(  c,00\right)  ,\label{q3}\\
Q\left(  d,11\right)   &  =\left(  d,11\right)  . \label{q4}%
\end{align}

\begin{lemma}
\label{conquasifred}Let $G$ be any nontrivial $n$-bit conservative gate.
\ Then $G$ generates a strong quasi-Fredkin gate.
\end{lemma}

\begin{proof}
By conservativity, $G$ maps unit vectors to unit vectors, say $G\left(
e_{i}\right)  =e_{\pi\left(  i\right)  }$ for some permutation $\pi$. But
since $G$ is nontrivial, there is some input $x\in\left\{  0,1\right\}  ^{n}$
such that $x_{i}=1$, but the corresponding bit $\pi\left(  i\right)  $ in
$G\left(  x\right)  $ is $0$. \ By conservativity, there must also be some bit
$j$ such that $x_{j}=0$, but bit $\pi\left(  j\right)  $ of $G\left(
x\right)  $ is $1$. \ Now permute the inputs to make bit $j$ and bit $i$ the
last two bits, permute the outputs to make bits $\pi\left(  j\right)  $ and
$\pi\left(  i\right)  $ the last two bits, and permute either inputs or
outputs to make $x$ match $G\left(  x\right)  $\ on the first $n-2$ bits.
\ After these permutations are performed, $x$\ has the form $w01$\ for some
$w\in\left\{  0,1\right\}  ^{n-2}$. \ So%
\begin{align*}
G\left(  0^{n-2},01\right)   &  =\left(  0^{n-2},01\right)  ,\\
G\left(  w,01\right)   &  =\left(  w,10\right)  ,\\
G\left(  0^{n-2},00\right)   &  =\left(  0^{n-2},00\right)  ,\\
G\left(  1^{n-2},11\right)   &  =\left(  11^{n-2},11\right)  ,
\end{align*}
where the last two lines again follow from conservativity. \ Hence $G$ (after
these permutations) satisfies the definition of a strong quasi-Fredkin gate.
\end{proof}

Next, call a gate $C$ a \textit{catalyzer} if, for every $x\in\left\{
0,1\right\}  ^{2n}$ with Hamming weight $n$, there exists a \textquotedblleft
program string\textquotedblright\ $p\left(  x\right)  $ such that
\[
C\left(  p\left(  x\right)  ,0^{n}1^{n}\right)  =\left(  p\left(  x\right)
,x\right)  .
\]
In other words, a catalyzer can be used to transform $0^{n}1^{n}$ into any
target string $x$\ of Hamming weight $n$. \ Here $x$\ can be encoded in any
manner of our choice into the auxiliary program string $p\left(  x\right)  $,
as long as $p\left(  x\right)  $ is left unchanged by the transformation.
\ The catalyzer itself cannot depend on $x$.

\begin{lemma}
\label{genswap}Let $Q$\ be a strong quasi-Fredkin gate. \ Then $Q$\ generates
a catalyzer.
\end{lemma}

\begin{proof}
Let $z:=0^{n}1^{n}$ be the string that we wish to transform. \ For all
$i\in\left\{  1,\ldots,n\right\}  $ and $j\in\left\{  n+1,\ldots,2n\right\}
$, let $s_{ij}$\ denote the operation that swaps the $i^{th}$\ and $j^{th}%
$\ bit of $z$. \ Then consider the following list of \textquotedblleft
candidate swaps\textquotedblright:%
\[
s_{1,n+1},\ldots,s_{1,2n},~~s_{2,n+1},\ldots,s_{2,2n},~~\ldots,~~s_{n,n+1}%
,\ldots,s_{n,2n}.
\]
Suppose we go through the list in order from left to right, and for each swap
in the list, get to choose whether to apply it or not. \ It is not hard to see
that, by making these choices, we can map $0^{n}1^{n}$\ to any $x$\ such that
$\left\vert x\right\vert =n$, by pairing off the first $0$ bit that should be
$1$ with the first $1$ bit that should be $0$, the second $0$ bit that should
be $1$ with the second $1$ bit that should be $0$, and so on, and choosing to
swap those pairs of bits and not any other pairs.

Now, let the program string $p\left(  x\right)  $\ be divided into $n^{2}%
$\ registers $r_{1},\ldots,r_{n^{2}}$, each of the same size. \ Suppose that,
rather than applying (or not applying) the\ $t^{th}$\ swap $s_{ij}$\ in the
list, we instead apply the gate $F$, with $r_{t}$\ as the control string, and
$z_{i}$\ and $z_{j}$\ as the target bits. \ Then we claim that we can map $z$
to $x$ as well. \ If the $t^{th}$\ candidate swap is supposed to occur, then
we set $r_{t}:=b$. \ If the $t^{th}$\ candidate swap is not supposed to occur,
then we set $r_{t}$\ to either $a$, $c$, or $d$,\ depending on whether
$z_{i}z_{j}$\ equals\ $01$, $00$, or $11$\ at step $t$\ of the swapping
process. \ Note that, because we know $x$ when designing $p\left(  x\right)
$, we know exactly what $z_{i}z_{j}$\ is going to be at each time step.
\ Also, $z_{i}z_{j}$\ will never equal $10$, because of the order in which we
perform the swaps: we swap each $0$ bit $z_{i}$\ that needs to be swapped with
the first $1$ bit $z_{j}$\ that we can. \ After we have performed the swap,
$z_{i}=1$\ will then only be compared against other $1$\ bits, never against
$0$ bits.
\end{proof}

Finally:

\begin{lemma}
\label{catalyst}Let $G$\ be any non-affine gate, and let $C$ be any catalyzer.
\ Then $G+C$ generates $\operatorname*{Fredkin}$.
\end{lemma}

\begin{proof}
We will actually show how to generate \textit{any} conservative transformation
$F:\left\{  0,1\right\}  ^{n}\rightarrow\left\{  0,1\right\}  ^{n}$.

Since $G$\ is non-affine, Lemmas \ref{lloydnot}, \ref{lloydand}, and
\ref{lloydcopy}\ together imply that we can use $G$ to compute any Boolean
function, albeit possibly with input-dependent garbage.

Let $x\in\{0,1\}^{n}$. \ Then by assumption, $C$ maps $0^{n}1^{n}$ to
$F\left(  x\right)  \overline{F\left(  x\right)  }$ using the program string
$p(F\left(  x\right)  \overline{F\left(  x\right)  })$. \ Now, starting with
$x$ and ancillas $0^{n}1^{n}$, we can clearly use $G$ to produce
\[
x,\operatorname{gar}\left(  x\right)  ,p(F\left(  x\right)  \overline{F\left(
x\right)  }),0^{n}1^{n},
\]
for some garbage $\operatorname{gar}\left(  x\right)  $. \ We can then apply
$C$ to get
\[
x,\operatorname{gar}\left(  x\right)  ,p(F\left(  x\right)  \overline{F\left(
x\right)  }),F\left(  x\right)  ,\overline{F\left(  x\right)  }.
\]
Uncomputing $p(F\left(  x\right)  \overline{F\left(  x\right)  })$ yields
\[
x,F\left(  x\right)  ,\overline{F\left(  x\right)  }.
\]
Notice that since $F$ is conservative, we have $\left\vert x,\overline
{F\left(  x\right)  }\right\vert =n$. \ Therefore, there exists some program
string $p(x,\overline{F\left(  x\right)  })$ that can be used as input to
$C^{-1}$ to map $x,\overline{F\left(  x\right)  }$ to $0^{n}1^{n}$. \ Again,
we can generate this program string using the fact that $G$ is non-affine:%
\[
x,F\left(  x\right)  ,\overline{F\left(  x\right)  },\operatorname{gar}\left(
F\left(  x\right)  \right)  ,p(x,\overline{F\left(  x\right)  }).
\]
Applying $C^{-1}$ and then uncomputing, we get%
\[
F\left(  x\right)  ,0^{n}1^{n}%
\]
which completes the proof.
\end{proof}

By Lemma \ref{affcontriv},\ every nontrivial conservative gate is also
non-affine. \ Therefore, combining Lemmas \ref{conquasifred}, \ref{genswap},
and \ref{catalyst}\ completes the proof of Theorem \ref{consfredkin}, that
every nontrivial conservative gate generates $\operatorname{Fredkin}$.

\subsection{Non-Conservative Generates Fredkin\label{MOD}}

Building on our work in Section \ref{CONSERV}, in this section we handle the
\textit{non}-conservative case, proving the following theorem.

\begin{theorem}
\label{danielthm}Every non-affine, non-conservative gate generates
$\operatorname{Fredkin}$.
\end{theorem}

Thus, let $G$\ be a non-affine, non-conservative gate. \ Starting from $G$, we
will perform a sequence of transformations to produce gates that are
\textquotedblleft gradually closer\textquotedblright\ to
$\operatorname{Fredkin}$. \ Some of these transformations might look a bit
mysterious, but they will culminate in a strong quasi-Fredkin gate, which we
already know from Lemmas \ref{genswap}\ and \ref{catalyst}\ is enough to
generate a $\operatorname{Fredkin}$\ gate (since $G$ is also non-affine).

The first step is to create a non-affine gate with two particular inputs as
fixed points.

\begin{lemma}
\label{fixedpoints}Let $G$ be any non-affine gate on $n$ bits. \ Then $G$
generates a non-affine gate $H$\ on $n^{2}$ bits that acts as the identity on
the inputs $0^{n^{2}}$ and $1^{n^{2}}$.
\end{lemma}

\begin{proof}
We construct $H$ as follows:

\begin{enumerate}
\item Apply $G^{\otimes n}$ to $n^{2}$\ input bits. \ Let $G_{i}$ be the
$i^{th}$ gate in this tensor product.

\item For all $i\in\left[  n-1\right]  $, swap the $i^{th}$ output bit of
$G_{i}$ with the $i^{th}$ output bit of $G_{n}$.

\item Apply $\left(  G^{-1}\right)  ^{\otimes n}$.
\end{enumerate}

It is easy to see that $H$ maps $0^{n^{2}}$ to $0^{n^{2}}$ and $1^{n^{2}}$ to
$1^{n^{2}}$. \ (Indeed, $H$ maps \textit{every} input that consists of an
$n$-bit string repeated $n$ times to itself.) \ To see that $H$ is also
non-affine, first notice that $G^{-1}$ is non-affine. \ But we can cause any
input $x=x_{1}\ldots x_{n}$\ that we like to be fed into the final copy of
$G^{-1}$, by encoding that input \textquotedblleft
diagonally,\textquotedblright\ with each $G_{i}$\ producing $x_{i}$\ as its
$i^{th}$\ output bit.\ \ Therefore $H$ is non-affine.
\end{proof}

As a remark, with all the later transformations we perform, we will want to
maintain the property that the all-$0$ and all-$1$ inputs are fixed points.
\ Fortunately, this will not be hard to arrange.

Let $H$\ be the output of Lemma \ref{fixedpoints}. \ If $H$ is conservative
(i.e., $k\left(  H\right)  =\infty$), then $H$ already generates
$\operatorname{Fredkin}$ by Theorem~\ref{consfredkin}, so we are done. \ Thus,
we will assume in what follows that $k\left(  H\right)  $ is finite. \ We will
further assume that $H$ is mod-$k\left(  H\right)  $-preserving. \ By Theorem
\ref{noshifter}, the only gates $H$\ that are \textit{not} mod-$k\left(
H\right)  $-preserving are the parity-flipping gates---but if $H$ is
parity-flipping, then $H\otimes H$\ is parity-preserving, and we can simply
repeat the whole construction with $H\otimes H$\ in place of $H$.

Now we want to show that we can use $H$\ to decrease the inner product between
a pair of inputs by exactly $1$\ mod $m$, for any $m$ we like.

\begin{lemma}
\label{primeslemma}Let $H$\ be any non-conservative, nonlinear gate. \ Then
for all $m\geq2$, there is a positive integer $t$, and inputs $x,y$, such that%
\[
H^{\otimes t}(x)\cdot H^{\otimes t}(y)-x\cdot y\equiv-1\left(
\operatorname{mod}m\right)  .
\]

\end{lemma}

\begin{proof}
Let $m=p_{1}^{\alpha_{1}}p_{2}^{\alpha_{2}}\ldots p_{s}^{\alpha_{s}}$ where
each $p_{i}$ is a distinct prime. \ By Corollary~\ref{nonlinip}, we know that
for each $p_{i}$, there is some pair of inputs $x_{i},y_{i}$ such that%
\[
H\left(  x_{i}\right)  \cdot H\left(  y_{i}\right)  \not \equiv x_{i}\cdot
y_{i}\left(  \operatorname{mod}p_{i}\right)  .
\]
In other words, letting%
\[
\gamma_{i}:=H\left(  x_{i}\right)  \cdot H\left(  y_{i}\right)  -x_{i}\cdot
y_{i},
\]
we have $\gamma_{i}\not \equiv 0\left(  \operatorname{mod}p_{i}\right)  $ for
all $i\in\left\{  1,\ldots,s\right\}  $. \ Our goal is to find an $\left(
x,y\right)  $\ such that%
\[
H^{\otimes t}\left(  x\right)  \cdot H^{\otimes t}\left(  y\right)  -x\cdot
y\equiv-1\left(  \operatorname{mod}m\right)  .
\]
To do so, it suffices to find nonnegative integers $d_{1},\ldots,d_{s}$\ that
solve the equation
\begin{equation}
\sum_{i=1}^{s}d_{i}\gamma_{i}\equiv-1\left(  \operatorname{mod}m\right)  .
\label{diyi}%
\end{equation}
Here $d_{i}$ represents the number of times the pair $\left(  x_{i}%
,y_{i}\right)  $\ occurs in $\left(  x,y\right)  $. \ By construction, no
$p_{i}$\ divides $\gamma_{i}$,\ and since the $p_{i}$'s\ are distinct primes,
they have no common factor. \ This implies that $\gcd\left(  \gamma_{1}%
,\ldots,\gamma_{s},m\right)  =1$. \ So by the Chinese Remainder Theorem, a
solution to (\ref{diyi}) exists.
\end{proof}

Note also that, if $H$\ maps the all-$0$ and all-$1$ strings to themselves,
then $H^{\otimes t}$\ does so as well.

To proceed further, it will be helpful to introduce some terminology.
\ Suppose that we have two strings $x=x_{1}\ldots x_{n}$ and $y=y_{1}\ldots
y_{n}$. \ For each $i$, the pair $x_{i}y_{i}$ has one of four possible values:
$00$, $01$, $10$, or $11$. \ Let the \textit{type} of $\left(  x,y\right)  $
be an ordered triple $\left(  a,b,c\right)  \in\mathbb{Z}^{3}$, which simply
records the number of occurrences in $\left(  x,y\right)  $\ of each of the
three pairs $01$, $10$, and $11$. \ (It will be convenient not to keep track
of $00$ pairs, since they don't contribute to the Hamming weight of either $x$
or $y$.) \ Clearly, by applying swaps, we can convert between any pairs
$\left(  x,y\right)  $\ and $\left(  x^{\prime},y^{\prime}\right)  $\ of the
same type, provided that $x,y,x^{\prime},y^{\prime}$\ all have the same length
$n$.

Now suppose that, by repeatedly applying a gate $H$, we can convert some input
pair $\left(  x,y\right)  $\ of type $\left(  a,b,c\right)  $\ into some pair
$\left(  x^{\prime},y^{\prime}\right)  $\ of type $\left(  a^{\prime
},b^{\prime},c^{\prime}\right)  $. \ Then we say that $H$ \textit{generates
the slope}%
\[
\left(  a^{\prime}-a,b^{\prime}-b,c^{\prime}-c\right)  .
\]
Note that, if $H$\ generates the slope $\left(  p,q,r\right)  $, then by
inverting the transformation, we can also generate the slope $\left(
-p,-q,-r\right)  $. \ Also, if $H$\ generates the slope $\left(  p,q,r\right)
$\ by acting on the input pair $\left(  x,y\right)  $, and the slope $\left(
p^{\prime},q^{\prime},r^{\prime}\right)  $\ by acting on $\left(  x^{\prime
},y^{\prime}\right)  $, then it generates the slope $\left(  p+p^{\prime
},q+q^{\prime},r+r^{\prime}\right)  $\ by acting on $\left(  xx^{\prime
},yy^{\prime}\right)  $. \ For these reasons, the achievable slopes form a
$3$-dimensional \textit{lattice}---that is, a subset of $\mathbb{Z}^{3}%
$\ closed under integer linear combinations---which we can denote
$\mathcal{L}\left(  H\right)  $.

What we really want is for the lattice $\mathcal{L}\left(  H\right)  $\ to
contain a particular point: $\left(  1,1,-1\right)  $. \ Once we have shown
this, we will be well on our way to generating a strong quasi-Fredkin\ gate.
\ We first need a general fact about slopes.

\begin{lemma}
\label{mreduce}Let $H$\ map the all-$0$ input to itself. \ Then $\mathcal{L}%
\left(  H\right)  $\ contains the points $\left(  k\left(  H\right)
,0,0\right)  $, $\left(  0,k\left(  H\right)  ,0\right)  $, and $\left(
0,0,k\left(  H\right)  \right)  $.
\end{lemma}

\begin{proof}
Recall from Proposition {\ref{whatdanielneeds} }that there exists a $t$, and
an input $w$, such that $\left\vert H^{\otimes t}\left(  w\right)  \right\vert
=\left\vert w\right\vert +k\left(  H\right)  $. \ Thus, to generate the slope
$\left(  k\left(  H\right)  ,0,0\right)  $, we simply need to do the following:

\begin{itemize}
\item Choose an input pair $\left(  x,y\right)  $\ with sufficiently many
$x_{i}y_{i}$\ pairs of the forms $10$\ and $00$.

\item Apply $H^{\otimes t}$\ to a subset of bits on which $x$\ equals $w$, and
$y$\ equals the all-$0$ string.
\end{itemize}

Doing this will increase the number of $10$\ pairs by $k\left(  H\right)  $,
while not affecting the number of $01$\ or $11$ pairs.

To generate the slope $\left(  0,k\left(  H\right)  ,0\right)  $, we do
exactly the same thing, except that we reverse the roles of $x$\ and $y$.

Finally, to generate the slope $\left(  0,0,k\left(  H\right)  \right)  $, we
choose an input pair $\left(  x,y\right)  $\ with sufficiently many
$x_{i}y_{i}$\ pairs of the forms $11$\ and $00$, and then use the same
procedure to increase the number of $11$\ pairs by $k\left(  H\right)  $.
\end{proof}

We can now prove that $\left(  1,1,-1\right)  $\ is indeed in our lattice.

\begin{lemma}
\label{inlattice}Let $H$\ be a mod-$k\left(  H\right)  $-preserving gate that
maps the all-$0$ input to itself, and suppose there exist inputs $x,y$\ such
that%
\[
H(x)\cdot H(y)-x\cdot y\equiv-1\left(  \operatorname{mod}k\left(  H\right)
\right)  .
\]
Then $\left(  1,1,-1\right)  \in\mathcal{L}\left(  H\right)  $.
\end{lemma}

\begin{proof}
The assumption implies directly that $H$\ generates a slope of the form
$\left(  p,q,-1+rk\left(  H\right)  \right)  $, for some integers $p,q,r$.
\ Thus, Lemma \ref{mreduce} implies that $H$ also generates a slope of the
form $\left(  p,q,-1\right)  $, via some gate $G\in\left\langle H\right\rangle
$\ acting on inputs $\left(  x,y\right)  $. \ Now, since $H$ is mod-$k\left(
H\right)  $-preserving, we have $\left\vert G\left(  x\right)  \right\vert
\equiv\left\vert x\right\vert \left(  \operatorname{mod}k\left(  H\right)
\right)  $\ and $\left\vert G\left(  y\right)  \right\vert \equiv\left\vert
y\right\vert \left(  \operatorname{mod}k\left(  H\right)  \right)  $. \ But
this implies that $p\equiv1\left(  \operatorname{mod}k\left(  H\right)
\right)  $\ and $q\equiv1\left(  \operatorname{mod}k\left(  H\right)  \right)
$. \ So, again using Lemma \ref{mreduce}, we can generate the slope $\left(
1,1,-1\right)  $.
\end{proof}

Combining Lemmas \ref{fixedpoints}, \ref{primeslemma}, and \ref{inlattice}, we
can summarize our progress so far as follows.

\begin{corollary}
\label{sofarcor}Let $G$ be any non-affine, non-conservative gate. \ Then
either $G$\ generates $\operatorname{Fredkin}$, or else it generates a gate
$H$\ that maps the all-$0$ and all-$1$ inputs to themselves, and that also
satisfies $\left(  1,1,-1\right)  \in\mathcal{L}\left(  H\right)  $.
\end{corollary}

We now explain the importance of the lattice point $\left(  1,1,-1\right)  $.
\ Given a gate $Q$, let us call $Q$\ \textit{weak quasi-Fredkin} if there
exist strings $a$ and $b$ such that%
\begin{align*}
Q\left(  a,01\right)   &  =\left(  a,01\right)  ,\\
Q\left(  b,01\right)   &  =\left(  b,10\right)  .
\end{align*}
Then:

\begin{lemma}
\label{quasfred}A gate $H$ generates a weak quasi-Fredkin gate if and only if
$\left(  1,1,-1\right)  \in\mathcal{L}\left(  H\right)  $.
\end{lemma}

\begin{proof}
If $H$\ generates a weak quasi-Fredkin gate $Q$, then applying $Q$ to the
input pair $\left(  a,01\right)  $\ and $\left(  b,01\right)  $\ directly
generates the slope $\left(  1,1,-1\right)  $. \ For the converse direction,
if $H$\ generates the slope $\left(  1,1,-1\right)  $, then by definition
there exists a gate $Q\in\left\langle H\right\rangle $, and inputs $x,y$, such
that $\left\vert Q\left(  x\right)  \right\vert =\left\vert x\right\vert
$\ and $\left\vert Q\left(  y\right)  \right\vert =\left\vert y\right\vert $,
while%
\[
Q\left(  x\right)  \cdot Q\left(  y\right)  =x\cdot y-1.
\]
In other words, applying $Q$\ decreases by one\ the number of $1$\ bits on
which $x$\ and $y$ agree, while leaving their Hamming weights the same. \ But
in that case, by permuting input and output bits, we can easily put $Q$ into
the form of a weak quasi-Fredkin gate.
\end{proof}

Next, recall the definition of a strong quasi-Fredkin gate from Section
\ref{CONSERV}. \ Then combining Corollary \ref{sofarcor} with Lemma
\ref{quasfred}, we have the following.

\begin{corollary}
\label{strongquasfred}Let $G$ be any non-affine, non-conservative gate. \ Then
either $G$\ generates $\operatorname{Fredkin}$, or else it generates a strong
quasi-Fredkin gate.
\end{corollary}

\begin{proof}
Combining Corollary \ref{sofarcor} with Lemma \ref{quasfred}, we find that
either $G$ generates $\operatorname{Fredkin}$, or else it generates a weak
quasi-Fredkin gate that maps the all-$0$ and all-$1$ strings to themselves.
\ But such a gate \textit{is} a strong quasi-Fredkin gate, since we can let
$c$\ be the all-$0$ string and $d$ be the all-$1$ string.
\end{proof}

Combining Corollary \ref{strongquasfred} with Lemmas \ref{genswap}\ and
\ref{catalyst} now completes the proof of Theorem \ref{danielthm}: that every
non-affine, non-conservative gate generates $\operatorname{Fredkin}$.
\ However, since every non-affine, conservative gate generates
$\operatorname{Fredkin}$ by Theorem~\ref{consfredkin}, we get the following
even broader corollary.

\begin{corollary}
\label{broadercor}Every non-affine gate\ generates $\operatorname{Fredkin}$.
\end{corollary}

Finally, combined with Corollary \ref{fredqueen}, Corollary \ref{broadercor}%
\ completes the proof of Theorem \ref{nonaffinedone}, that every non-affine
gate set generates either $\left\langle \operatorname*{Fredkin}\right\rangle
$,\ $\left\langle \operatorname*{Fredkin},\operatorname*{NOTNOT}\right\rangle
$, $\left\langle \operatorname*{Fredkin},\operatorname*{NOT}\right\rangle $,
$\left\langle \operatorname*{C}_{k}\right\rangle $ for some $k\geq3$, or
$\left\langle \operatorname*{Toffoli}\right\rangle $.

\section{The Affine Part\label{AFFINE}}

Having completed the classification of the non-affine classes, in this section
we turn our attention to proving that there are no affine classes besides the
ones listed in Theorem \ref{main}: namely, the trivial, $\operatorname{T}_{6}%
$, $\operatorname{T}_{4}$, $\operatorname{F}_{4}$, $\operatorname{CNOTNOT}$,
and $\operatorname{CNOT}$ classes, as well as various extensions of them by
$\operatorname{NOTNOT}$\ and $\operatorname{NOT}$ gates.

To make the problem manageable, we start by restricting attention to the
linear parts of affine transformations (i.e., if a transformation has the form
$G\left(  x\right)  =Ax\oplus b$, we ignore the additive constant $b$). \ We
show that the only possibilities for the linear part are: the identity, all
mod-$4$-preserving orthogonal transformations, all orthogonal transformations,
all parity-preserving linear transformations, or all linear transformations.
\ This result, in turn, is broken into several pieces:

\begin{itemize}
\item In Section \ref{SWAMPLAND}, we show that any mod-$4$-preserving
orthogonal gate generates \textit{all} mod-$4$-preserving orthogonal
transformations, and that any non-mod-$4$-preserving orthogonal gate generates
all orthogonal transformations.

\item In Section \ref{TOCNOTNOT}, we show that every non-orthogonal,
parity-preserving linear gate generates $\operatorname{CNOTNOT}$. \ This again
requires \textquotedblleft slope theory\textquotedblright\ and the analysis of
a $3$-dimensional lattice. \ It also draws on the results of Section
\ref{MOD2MOD4}, which tell us that it suffices to restrict attention to the
case $k\left(  G\right)  =2$.

\item In Section \ref{TOCNOT},\ we show that every non-parity-preserving
linear gate generates $\operatorname{CNOT}$. \ In this case we are lucky that
we only need to analyze a $1$-dimensional lattice (i.e., an ideal in
$\mathbb{Z}$)
\end{itemize}

Finally, in Section \ref{NOT}, we complete the classification by showing that
including the affine parts can yield only the following additional
possibilities: $\operatorname{NOTNOT}$, $\operatorname{NOT}$,
$\operatorname{F}_{4}$, $\operatorname{F}_{4}+\operatorname{NOTNOT}$,
$\operatorname{F}_{4}+\operatorname{NOT}$, $\operatorname{T}_{6}%
+\operatorname{NOTNOT}$, $\operatorname{T}_{6}+\operatorname{NOT}$, or
$\operatorname{CNOTNOT}+\operatorname{NOT}$. \ Summarizing, the results of
this section will imply the following.

\begin{theorem}
\label{affinedone}Any set of affine gates generates one of the following $13$
classes: $\left\langle \varnothing\right\rangle $, $\left\langle
\operatorname{NOTNOT}\right\rangle $, $\left\langle \operatorname{NOT}%
\right\rangle $, $\left\langle \operatorname{T}_{6}\right\rangle $,
$\left\langle \operatorname{T}_{6},\operatorname{NOTNOT}\right\rangle $,
$\left\langle \operatorname{T}_{6},\operatorname{NOT}\right\rangle $,
$\left\langle \operatorname{T}_{4}\right\rangle $, $\left\langle
\operatorname*{F}_{4}\right\rangle $, $\left\langle \operatorname{T}%
_{4},\operatorname{NOTNOT}\right\rangle $, $\left\langle \operatorname{T}%
_{4},\operatorname{NOT}\right\rangle $, $\left\langle \operatorname{CNOTNOT}%
\right\rangle $, $\left\langle \operatorname{CNOTNOT},\operatorname{NOT}%
\right\rangle $, or $\left\langle \operatorname{CNOT}\right\rangle $.
\end{theorem}

Together with Theorem \ref{nonaffinedone}, this will then complete the proof
of Theorem \ref{main}.

\subsection{The T and F Swamplands\label{SWAMPLAND}}

In this section, we wish to characterize the orthogonal classes. \ We first
need a lemma.

\begin{lemma}
\label{treduce}$\operatorname{T}_{4k+2}$ generates $\operatorname{T}_{6}$, and
$\operatorname{T}_{4k}$ generates $\operatorname{T}_{4}$, for all $k\geq1$.
\end{lemma}

\begin{proof}
We first describe how to simulate $\operatorname{T}_{6}\left(  x_{1}\ldots
x_{6}\right)  $, using three applications of $\operatorname{T}_{4k+2}$. \ Let
$b_{x}:=x_{1}\oplus\cdots\oplus x_{6}$. \ Also, let $a$\ be a string of
ancilla bits, initialized to $0^{2k-2}$. \ Then:

\begin{enumerate}
\item Apply $\operatorname{T}_{4k+2}$\ to the string $0^{4k-4}x_{1}\ldots
x_{6}$. \ This yields $b_{x}^{4k-4},x_{1}\oplus b_{x},\ldots,x_{6}\oplus
b_{x}$.

\item Swap out $2k-2$\ of the $b_{x}$\ bits with the ancilla string
$a=0^{2k-2}$, and apply $\operatorname{T}_{4k+2}$\ again. \ This yields%
\[
\operatorname{T}_{4k+2}\left(  0^{2k-2},b_{x}^{2k-2},x_{1}\oplus b_{x}%
,\ldots,x_{6}\oplus b_{x}\right)  =\left(  b_{x}^{2k-2},0^{2k-2},x_{1}\ldots
x_{6}\right)  ,
\]
since the number of `$b_{x}$' entries is even.

\item Swap the $2k-2$ bits that are now $0$ with $a=b_{x}^{2k-2}$, and apply
$\operatorname{T}_{4k+2}$ a third time. \ This returns $a$ to $0^{2k-2}$, and
yields%
\[
\operatorname{T}_{4k+2}\left(  b_{x}^{4k-4},x_{1}\ldots x_{6}\right)
=\operatorname{T}_{4k+2}\left(  0^{4k-4},x_{1}\oplus b_{x},\ldots x_{6}\oplus
b_{x}\right)  .
\]

\end{enumerate}

Thus, we have successfully applied $\operatorname{T}_{6}$ to $x_{1}\ldots
x_{6}$. \ The same sequence of steps can be used to simulate $\operatorname{T}%
_{4}\left(  x_{1}\ldots x_{4}\right)  $ using three applications of
$\operatorname{T}_{4k}$.
\end{proof}

We can now show that there is only one nontrivial orthogonal class that is
also mod-$4$-preserving: namely, $\left\langle \operatorname{T}_{6}%
\right\rangle $.

\begin{theorem}
\label{gettingt6}Any nontrivial mod-$4$-preserving linear gate $G$\ generates
$\operatorname{T}_{6}$.
\end{theorem}

\begin{proof}
Let $G\left(  x\right)  =Ax$, for some $A\in\mathbb{F}_{2}^{n\times n}%
$.\ \ Then recall from Corollary~\ref{mod4orthog} that $A$ is orthogonal. \ By
Lemma \ref{ipcond}, this implies that $A^{-1}=A^{T}$, so $G$ can also generate
$A^{T}$.

Let $B$\ be the $\left(  n+1\right)  \times\left(  n+1\right)  $\ matrix that
acts as the identity on the first bit, and as $A$ on bits $2$ through $n+1$.
\ Observe that $B^{T}$ acts as the identity on the first bit, and as $A^{T}%
$\ on bits $2$ through $n+1$. \ Also, since $A$ preserves Hamming weight mod
$4$, so do $A^{T}$, $B$, and $B^{T}$. \ By Corollary \ref{mod4cond}, this
implies that each of $\left(  B^{T}\right)  $'s column vectors must have
Hamming weight $1$ mod $4$. \ Furthermore, since $A$ is nontrivial, there must
be some column of $B^{T}$ with Hamming weight $4k+1$, for some $k\geq1$.
\ Then by swapping rows and columns, we can get $B^{T}$ into the form%
\[
\left(
\begin{array}
[c]{ccc}%
1 & 0 & 0~~~~\cdots~~~~0\\
0 & 1 & \text{---}v_{1}\text{---}\\
\vdots & \vdots & \vdots\\
0 & 1 & \text{---}v_{4k+1}\text{---}\\
0 & 0 & \text{---}v_{4k+2}\text{---}\\
\vdots & \vdots & \vdots\\
0 & 0 & \text{---}v_{n}\text{---}%
\end{array}
\right)  ,
\]
where $v_{1},\ldots,v_{n}$ are row vectors each of length $n-1$. \ Let
$\delta_{ij}$\ equal $1$ if $i=j$\ or $0$\ otherwise.\ \ Then note that by
orthogonality,%
\[
v_{i}\cdot v_{j}=\left\{
\begin{tabular}
[c]{ll}%
$\overline{\delta_{ij}}$ & if $i,j\leq4k+1,$\\
$\delta_{ij}$ & otherwise.
\end{tabular}
\ \right.
\]

Now let $C^{T}$\ be the matrix obtained by swapping the first two columns of
$B^{T}$. \ Then we claim that $C^{T}B$ yields a $\operatorname{T}_{4k+2}$
transformation. \ Since $\operatorname{T}_{4k+2}$ generates $\operatorname{T}%
_{6}$\ by Lemma~\ref{treduce}, we will be done after we have shown this.

We have%
\begin{align*}
C^{T}B  &  =\left(
\begin{array}
[c]{ccc}%
0 & 1~ & 0~~~~\cdots~~~~0\\
1 & 0 & \text{---}v_{1}\text{---}\\
\vdots & \vdots~ & \vdots\\
1 & 0 & \text{---}v_{4k+1}\text{---}\\
0 & 0 & \text{---}v_{4k+2}\text{---}\\
\vdots & \vdots~ & \vdots\\
0 & 0 & \text{---}v_{n}\text{---}%
\end{array}
\right)  \left(
\begin{array}
[c]{ccccccc}%
1 & 0 & \cdots & 0 & 0 & \cdots & 0\\
0 & 1 & \cdots & 1 & 0 & \cdots & 0\\
0 & | & ~ & | & | & ~ & |\\
\vdots & v_{1}^{T} & \cdots & v_{4k+1}^{T}~ & v_{4k+2}^{T} & \cdots &
v_{n}^{T}\\
0 & | & ~ & | & | & ~ & |
\end{array}
\right) \\
&  =\left(
\begin{array}
[c]{cccccc}%
0 & 1 & 1 & 0 & 0 & 0\\
1 & \ddots & 1 & 0 & 0 & 0\\
1 & 1 & 0 & 0 & 0 & 0\\
0 & 0 & 0 & 1 & 0 & 0\\
0 & 0 & 0 & 0 & \ddots & 0\\
0 & 0 & 0 & 0 & 0 & 1
\end{array}
\right)  .
\end{align*}
One can check that the above transformation is actually $\operatorname{T}%
_{4k+2}$ on the first $4k+2$ bits, and the identity on the rest.
\end{proof}

Likewise, there is only one orthogonal class that is \textit{not}
mod-$4$-preserving: namely, $\left\langle \operatorname{T}_{4}\right\rangle $.

\begin{theorem}
\label{gettingt4}Let $G$\ be any nontrivial orthogonal gate that does not
preserve Hamming weight mod $4$. \ Then $G$ generates $\operatorname{T}_{4}$.
\end{theorem}

\begin{proof}
We use essentially the same construction as in Theorem~\ref{gettingt6}. \ The
only change is that Corollary~\ref{mod4cond} now tells us that there must be a
column of $B^{T}$ with Hamming weight $4k+3$ for some $k\geq1$, so we use that
in place of the column with Hamming weight $4k+1$. \ This leads to an $\left(
n+1\right)  \times\left(  n+1\right)  $\ matrix $C^{T}B$, which acts as
$\operatorname{T}_{4k+4}$ on the first $4k+4$\ bits and as the identity on the
rest. \ But $\operatorname{T}_{4k+4}$ generates $\operatorname{T}_{4}$ by
Lemma \ref{treduce}, so we are done.
\end{proof}

\subsection{Non-Orthogonal Linear Generates CNOTNOT\label{TOCNOTNOT}}

In classifying all linear gate sets, our next goal is to show that
\textquotedblleft there is nothing between orthogonal and
parity-preserving.\textquotedblright\ \ In other words:

\begin{theorem}
\label{scottthm}Let $G$ be any non-orthogonal, parity-preserving linear gate.
\ Then $G$ generates $\operatorname{CNOTNOT}$ (or equivalently, all
parity-preserving linear transformations).
\end{theorem}

The main idea of the proof is as follows. \ Let $\operatorname{CPD}$, or
\textit{Copying with a Parity Dumpster}, be the following partial reversible
gate:%
\begin{align*}
\operatorname{CPD}\left(  000\right)   &  =000,\\
\operatorname{CPD}\left(  001\right)   &  =001,\\
\operatorname{CPD}\left(  100\right)   &  =111,\\
\operatorname{CPD}\left(  101\right)   &  =110.
\end{align*}
In other words, $\operatorname{CPD}$\ maps $x0y$ to $x,x,x\oplus y$---copying
$x$, but also XORing $x$\ into the $y$\ \textquotedblleft
dumpster\textquotedblright\ in order to preserve the total parity. \ Notice
that $\operatorname{CPD}$\ is consistent with $\operatorname*{CNOTNOT}$;
indeed, it is simply the restriction of $\operatorname*{CNOTNOT}$\ to inputs
whose second bit is $0$. \ Notice also that, whenever we have a $3$-bit string
of the form $xxy$, we can apply $\operatorname{CPD}$\ in reverse to get
$x,0,x\oplus y$.

Then we will first observe that $\operatorname{CPD}$\ generates
$\operatorname*{CNOTNOT}$. \ We will then apply the theory of types and
slopes, which already made an appearance in Section \ref{MOD}, to show that
any non-orthogonal linear gate generates $\operatorname{CPD}$: in essence,
that there are no modularity or other obstructions to generating it.

\begin{lemma}
\label{cpdlem}Let $G$ be any gate that generates $\operatorname{CPD}$. \ Then
$G$ generates $\operatorname*{CNOTNOT}$ (or equivalently, all
parity-preserving linear transformations).
\end{lemma}

\begin{proof}
Let $F:\left\{  0,1\right\}  ^{n}\rightarrow\left\{  0,1\right\}  ^{n}$\ be
any reversible, parity-preserving linear transformation. \ Then we can
generate the following sequence of states:%
\begin{align*}
x  &  \rightarrow x,\operatorname{gar}\left(  x\right)  ,F\left(  x\right) \\
&  \rightarrow x,\operatorname{gar}\left(  x\right)  ,F\left(  x\right)
,F\left(  x\right)  ,\left\vert x\right\vert \left(  \operatorname{mod}%
2\right) \\
&  \rightarrow x,F\left(  x\right)  ,\left\vert x\right\vert \left(
\operatorname{mod}2\right) \\
&  \rightarrow x,F\left(  x\right)  ,\operatorname{gar}\left(  F\left(
x\right)  \right)  ,x,\left\vert x\right\vert \left(  \operatorname{mod}%
2\right) \\
&  \rightarrow F\left(  x\right)  ,\operatorname{gar}\left(  F\left(
x\right)  \right)  ,x,\left\vert x\right\vert +\left\vert F\left(  x\right)
\right\vert \left(  \operatorname{mod}2\right) \\
&  =F\left(  x\right)  ,\operatorname{gar}\left(  F\left(  x\right)  \right)
,x,0\left(  \operatorname{mod}2\right) \\
&  \rightarrow F\left(  x\right)  ,
\end{align*}
for some garbage strings $\operatorname{gar}\left(  x\right)  $\ and
$\operatorname{gar}\left(  F\left(  x\right)  \right)  $. \ Here the first
line computes $F\left(  x\right)  $ from $x$; the second line applies
$\operatorname{CPD}$ to copy $F\left(  x\right)  $\ (using a single
\textquotedblleft dumpster\textquotedblright\ bit for each bit of $F\left(
x\right)  $); the third line uncomputes $F\left(  x\right)  $; the fourth line
computes a second copy of $x$ from $F\left(  x\right)  $; the fifth line
applies $\operatorname{CPD}$ in reverse to erase one of the copies of
$x$\ (reusing same dumpster bit from before); and the sixth line uncomputes
$x$. \ Also, $\left\vert x\right\vert +\left\vert F\left(  x\right)
\right\vert \equiv0\left(  \operatorname{mod}2\right)  $\ follows because $F$
is parity-preserving.
\end{proof}

So, given a non-orthogonal, parity-preserving linear gate $G$, we now need to
show how to implement $\operatorname{CPD}$.

For the rest of this section, we will consider a situation where we are given
an $n$-bit string, with the initial state $xy0^{n-2}$ (where $x$\ and $y$\ are
two arbitrary bits), and then we apply a sequence of $\mathbb{F}_{2}$\ linear
transformations to the string. \ Here we do \textit{not} assume that ancilla
bits initialized to $1$\ are available, though ancilla bits initialized to $0$
are fine. \ As a result, at every time step, every bit in our string will be
either $x$, $y$, $x\oplus y$, or $0$. \ Because we are studying only the
linear case here, not the affine case, we do not need to worry about the
possibilities $x\oplus1$,\ $y\oplus1$, etc., which would considerably
complicate matters. \ (We will handle the affine case in Section \ref{NOT}.)

By analogy to Section \ref{MOD}, let us define the \textit{type} of a string
$z\left(  x,y\right)  \in\left\{  0,1\right\}  ^{n}$\ to be $\left(
a,b,c\right)  $, if $z$\ contains $a$\ copies of $x$ and $b$\ copies of
$y$\ and $c$ copies of $x\oplus y$. \ Since any string of type $\left(
a,b,c\right)  $\ can be transformed into any other string of type $\left(
a,b,c\right)  $ using bit-swaps, the type of $z$ is its only relevant
property. \ As before, if by repeatedly applying a linear gate\ $G$, we can
map some string of type $\left(  a,b,c\right)  $\ into some string of type
$\left(  a^{\prime},b^{\prime},c^{\prime}\right)  $, then we say that $G$
\textit{generates the slope}%
\[
\left(  a^{\prime}-a,b^{\prime}-b,c^{\prime}-c\right)  .
\]
Again, if $G$\ generates the slope $\left(  p,q,r\right)  $, then $G^{-1}%
$\ generates the slope $\left(  -p,-q,-r\right)  $. \ Also, if $G$ generates
the slope $\left(  p,q,r\right)  $\ using the string $z$, and the slope
$\left(  p^{\prime},q^{\prime},r^{\prime}\right)  $\ using the string
$z^{\prime}$, then it generates the slope $\left(  p+p^{\prime},q+q^{\prime
},r+r^{\prime}\right)  $\ using the string $zz^{\prime}$. \ For these reasons,
the set of achievable slopes forms a $3$-dimensional lattice, which we denote
$\mathcal{L}\left(  G\right)  \subseteq\mathbb{Z}^{3}$. \ Moreover, this is a
lattice with a strong symmetry property:

\begin{proposition}
\label{sym3}$\mathcal{L}\left(  G\right)  $ is symmetric under all $6$
permutations of the $3$\ coordinates.
\end{proposition}

\begin{proof}
Clearly we can interchange the roles of $x$\ and $y$. \ However, we can also,
e.g., define $x^{\prime}:=x$\ and $y^{\prime}:=x\oplus y$, in which case
$x^{\prime}\oplus y^{\prime}=y$. \ In the triple $\left(  x,y,x\oplus
y\right)  $, each element is the XOR of the other two.
\end{proof}

Just like before, our entire question will boil down to whether or not the
lattice $\mathcal{L}\left(  G\right)  $\ contains a certain point. \ In this
case, the point is $\left(  1,-1,1\right)  $. \ The importance of the $\left(
1,-1,1\right)  $ point comes from the following lemma.

\begin{lemma}
\label{specialpt}Let $G$ be any linear gate. \ Then $G$ generates
$\operatorname{CPD}$, if and only if $\left(  1,-1,1\right)  \in
\mathcal{L}\left(  G\right)  $.
\end{lemma}

\begin{proof}
If $G$ generates $\operatorname{CPD}$, then it maps $x0y$, which has type
$\left(  1,1,0\right)  $,\ to $x,x,x\oplus y$, which has type $\left(
2,0,1\right)  $. \ This amounts to generating the slope $\left(
1,-1,1\right)  $.

Conversely, suppose $\left(  1,-1,1\right)  \in\mathcal{L}\left(  G\right)  $.
\ Then there is some gate $H\in\left\langle G\right\rangle $, and some string
of the form $z=x^{a}y^{b}\left(  x\oplus y\right)  ^{c}$, such that%
\[
H\left(  z\right)  =x^{a+1}y^{b-1}\left(  x\oplus y\right)  ^{c+1}.
\]
But the very fact that $G$ generates such an $H$ implies that $G$ is
non-degenerate, and if $G$\ is non-degenerate, then Lemma \ref{lloydcopy}%
\ implies that, starting from $xy0^{n-2}$, we can use $G$ to increase the
numbers of $x$, $y$, and $x\oplus y$ simultaneously without bound. \ That is,
there is some $Q\in\left\langle G\right\rangle $\ such that (omitting the $0$
bits)%
\[
Q\left(  xy\right)  =x^{a^{\prime}}y^{b^{\prime}}\left(  x\oplus y\right)
^{c^{\prime}},
\]
where $a^{\prime}>a$ and $b^{\prime}>b$ and $c^{\prime}>c$. \ So then the
procedure to implement $\operatorname{CPD}$\ is to apply $Q$,\ then $H$, then
$Q^{-1}$.
\end{proof}

Thus, our goal now is to show that, if $G$\ is any non-orthogonal,
parity-preserving linear gate, then $\left(  1,-1,1\right)  \in\mathcal{L}%
\left(  G\right)  $. \ Observe that, if $k\left(  G\right)  =4$, then
Corollary \ref{mod4orthog}\ implies that $G$ is orthogonal, contrary to
assumption. \ By Theorem \ref{k2or4}, this means that the only remaining
possibility is $k\left(  G\right)  =2$. \ This has the following consequence
for the lattice $\mathcal{L}\left(  G\right)  $.

\begin{proposition}
\label{prop200}If $G$\ is a linear gate with $k\left(  G\right)  \leq2$, then
$\mathcal{L}\left(  G\right)  $ contains all even points (i.e., all $\left(
p,q,r\right)  $ such that $p\equiv q\equiv r\equiv0\left(  \operatorname{mod}%
2\right)  $).
\end{proposition}

\begin{proof}
By Proposition \ref{whatdanielneeds}, we must be able to use $G$ to map
$10^{n-1}$\ to $1110^{n-3}$. \ Since $0^{n}$ is mapped to itself by any linear
transformation, this implies that $G$ can map $x0^{n-1}$\ to $xxx0^{n-3}%
$,\ which means that it generates the slope $\left(  2,0,0\right)  $. \ So
$\left(  2,0,0\right)  \in\mathcal{L}\left(  G\right)  $. \ By Proposition
\ref{sym3}, then, $\mathcal{L}\left(  G\right)  $ also contains the points
$\left(  0,2,0\right)  $\ and $\left(  0,0,2\right)  $. \ But these three
generate all the even points.
\end{proof}

Proposition \ref{prop200}\ has the following immediate corollary.

\begin{corollary}
\label{cor200}Let $G$\ be a linear gate with $k\left(  G\right)  \leq2$, and
suppose $\mathcal{L}\left(  G\right)  $\ contains any point $\left(
p,q,r\right)  $\ such that $p\equiv q\equiv r\equiv1\left(  \operatorname{mod}%
2\right)  $. \ Then $\mathcal{L}\left(  G\right)  $\ contains $\left(
1,-1,1\right)  $.
\end{corollary}

Thus, it remains only to prove the following lemma.

\begin{lemma}
\label{abc1}Let $G$\ be any parity-preserving, non-orthogonal linear gate.
\ Then $\mathcal{L}\left(  G\right)  $\ contains a point $\left(
p,q,r\right)  $\ such that $p\equiv q\equiv r\equiv1\left(  \operatorname{mod}%
2\right)  $.
\end{lemma}

\begin{proof}
In the proof of Theorem \ref{scottthm}, this is the first place where we use
the linearity of $G$ in an essential way---i.e., not just to deduce that
$k\left(  G\right)  \in\left\{  2,4\right\}  $, or to avoid dealing with bits
of the form $x\oplus1$,\ $y\oplus1$, etc. \ It is also the first place where
we use the non-orthogonality of $G$, other than to rule out the possibility
that $k\left(  G\right)  =4$; and the first place where we use that $G$ is parity-preserving.

Let us view $G$ as an $n\times n$ matrix over $\mathbb{F}_{2}$. \ Then the
fact that $G$ is parity-preserving means that every column of $G$ has odd
Hamming weight. \ Also, the fact that $G$ is non-orthogonal means that it must
have two columns with an odd inner product. \ Assume without loss of
generality that these are the first and second columns. \ Let the first two
columns of $G$ consist of:%
\begin{align*}
&  a\text{ rows of the form }1,0,\\
&  b\text{ rows of the form }1,1,\\
&  c\text{ rows of the form }0,1,\\
&  d\text{ rows of the form }0,0,
\end{align*}
where $a,b,c,d$\ are nonnegative integers summing to $n$. \ Then from the
above, we have that $a+b$ and $b+c$ and $b$ are all odd, from which it follows
that $a$ and $c$ are even.

Now consider applying $G$ to the input $xy0^{n-2}$. \ The result will contain:%
\begin{align*}
&  a\text{ copies of }x\text{,}\\
&  c\text{ copies of }y\text{,}\\
&  b\text{ copies of }x\oplus y\text{.}%
\end{align*}
This means that we've mapped a string of type $\left(  1,1,0\right)  $\ to a
string of type $\left(  a,c,b\right)  $, thereby generating the slope $\left(
a-1,c-1,b\right)  $. \ But this is the desired odd point in $\mathcal{L}%
\left(  G\right)  $.
\end{proof}

Combining Lemma \ref{cpdlem}, Lemma \ref{specialpt}, Corollary \ref{cor200},
and Lemma \ref{abc1}\ now completes the proof of Theorem \ref{scottthm}.

\subsection{Non-Parity-Preserving Linear Generates CNOT\label{TOCNOT}}

To complete the classification of linear gate sets, our final task is to prove
the following theorem.

\begin{theorem}
\label{scottthm2}Let $G$ be any non-parity-preserving linear gate. \ Then $G$
generates $\operatorname{CNOT}$ (or equivalently, all linear transformations).
\end{theorem}

Recall that $\operatorname{COPY}$\ is the partial gate that maps $x0$\ to
$xx$. \ We will first show how to use $G$ to generate $\operatorname{COPY}$,
and then use $\operatorname{COPY}$\ to generate $\operatorname{CNOT}$.

Note that since $G$\ is linear, it cannot be parity-flipping. \ So since $G$
is non-parity-preserving, it is also non-parity-respecting, and $k\left(
G\right)  $\ must be finite and odd. \ But by Theorem \ref{k2or4}, this means
that $k\left(  G\right)  =1$: in other words, $G$ is non-mod-respecting.

Let $z$ be an $n$-bit string that consists entirely of copies of $x$\ and $0$.
\ Let the \textit{type} of $z$ be the number of copies of $x$. \ Clearly we
can map any $z$ to any other $z$ of the same type using swaps, so the type of
$z$\ is its only relevant property. \ Also, we say that a gate $G$
\textit{generates the slope} $p$, if by applying $G$ repeatedly, we can map
some input $z$\ of type $a$ to some input $z^{\prime}$\ of type $a+p$. \ Note
that if $G$ generates the slope $p$, then by reversibility, it also generates
the slope $-p$. \ Also, if $G$\ generates the slope $p$\ by mapping $z$\ to
$z^{\prime}$, and the slope $q$ by mapping $w$\ to $w^{\prime}$, then it
generates the slope $p+q$\ by mapping $zw$\ to $z^{\prime}w^{\prime}$. \ For
these reasons, the set of achievable slopes forms an ideal in $\mathbb{Z}$
(i.e., a $1$-dimensional lattice), which we can denote $\mathcal{L}\left(
G\right)  $. \ The question of whether $G$ generates $\operatorname{COPY}%
$\ can then be rephrased as the question of whether $\mathcal{L}\left(
G\right)  $\ contains $1$---or equivalently, of whether $\mathcal{L}\left(
G\right)  =\mathbb{Z}$.

\begin{lemma}
\label{lg1}A linear gate $G$ generates $\operatorname{COPY}$\ if and only if
$1\in\mathcal{L}\left(  G\right)  $.
\end{lemma}

\begin{proof}
If $G$\ generates $\operatorname{COPY}$, then clearly $1\in\mathcal{L}\left(
G\right)  $. \ For the converse direction, suppose $1\in\mathcal{L}\left(
G\right)  $. \ Then $G$ can be used to map an input of type $a$\ to an input
of type $a+1$, for some $a$. \ Hence $G$\ can also be used to map inputs of
type $b$\ to inputs of type $b+1$, for all $b\geq a$. \ This also implies that
$G$\ is non-degenerate, so by Lemma \ref{lloydcopy}, it can be used to
increase the number of copies of $x$\ without bound. \ So to copy a bit $x$,
we first apply some gate $H\in\left\langle G\right\rangle $\ to map $x$ to
$x^{b}$\ for some $b\geq a$, then map $x^{b}$\ to $x^{b+1}$, and finally apply
$H^{-1}$\ to map $x^{b+1}$\ to $x^{2}$.
\end{proof}

Now, the question of whether $1\in\mathcal{L}\left(  G\right)  $\ is easily answered.

\begin{lemma}
\label{lgz}Let $G$ be any non-mod-respecting linear gate. \ Then
$\mathcal{L}\left(  G\right)  =\mathbb{Z}$.
\end{lemma}

\begin{proof}
This follows almost immediately from\ Proposition \ref{whatdanielneeds},
together with the fact that $k\left(  G\right)  =1$. \ We simply need to
observe that, if $x=1$, then the number of copies of $x$\ corresponds to the
Hamming weight.
\end{proof}

Finally, we show that $\operatorname{COPY}$\ suffices for $\operatorname{CNOT}%
$.

\begin{lemma}
\label{finishit}Let $G$\ be any linear gate that generates
$\operatorname{COPY}$. \ Then $G$ generates $\operatorname{CNOT}$.
\end{lemma}

\begin{proof}
We will actually prove that $G$ generates \textit{any} linear transformation
$F$. \ Observe that, if $G$ generates $\operatorname{COPY}$, then it must be
non-degenerate. \ Therefore, by copying bits whenever needed, and using $G$ to
do computation on them, clearly we can map the input $x$ to a string of the
form%
\[
x,\operatorname{gar}\left(  x\right)  ,F\left(  x\right)  ,
\]
for some garbage string $\operatorname{gar}\left(  x\right)  $. \ Since $G$
generates $\operatorname{COPY}$, we can then make one copy of $F\left(
x\right)  $, mapping the above to%
\[
x,\operatorname{gar}\left(  x\right)  ,F\left(  x\right)  ,F\left(  x\right)
.
\]
Next we can uncompute the computation of $F$ to get%
\[
x,F\left(  x\right)  .
\]
By reversibility, we can then map the above to%
\[
x,F\left(  x\right)  ,\operatorname{gar}\left(  F\left(  x\right)  \right)
,x.
\]
By inverting $\operatorname{COPY}$, we can then implement $xx\rightarrow x$,
to map the above to%
\[
F\left(  x\right)  ,\operatorname{gar}\left(  F\left(  x\right)  \right)  ,x.
\]
Finally, we can uncompute the computation of $x$ to get $F\left(  x\right)  $\ alone.
\end{proof}

Combining Lemmas \ref{lg1}, \ref{lgz}, and \ref{finishit}\ now completes the
proof of Theorem \ref{scottthm2}. \ Then combining Theorems \ref{cnotcirc},
\ref{cnotnotcirc}, \ref{t4circ}, \ref{t6circ}, \ref{gettingt6},
\ref{gettingt4},\ \ref{scottthm}, and \ref{scottthm2}, we can summarize our
progress on the linear case as follows.

\begin{corollary}
\label{scottcor}Every set of linear gates generates either $\left\langle
\varnothing\right\rangle $, $\left\langle \operatorname{T}_{6}\right\rangle $,
$\left\langle \operatorname{T}_{4}\right\rangle $, $\left\langle
\operatorname{CNOTNOT}\right\rangle $, or $\left\langle \operatorname{CNOT}%
\right\rangle $.
\end{corollary}

\subsection{Adding Back the NOTs\label{NOT}}

Now that we have completed the classification of the \textit{linear} gate
classes, the final step that remains is to take care of the affine parts. \ We
first give some useful lemmas for manipulating affine gates.

\begin{lemma}
\label{notlem}$\operatorname{NOT}^{\otimes k}$ generates
$\operatorname{NOTNOT}$ for all $k\geq1$, as well as $\operatorname{NOT}$\ if
$k$ is odd.
\end{lemma}

\begin{proof}
To implement $\operatorname*{NOTNOT}\left(  x,y\right)  $, apply
$\operatorname{NOT}^{\otimes k}$\ to $x,a_{1}\ldots a_{k-1}$ and then to
$y,a_{1}\ldots a_{k-1}$. \ To implement $\operatorname*{NOT}\left(  x\right)
$,\ let $\ell:=\frac{k-1}{2}$. \ Apply $\operatorname{NOT}^{\otimes k}$\ to
$x,a_{1}\ldots a_{\ell},b_{1}\ldots b_{\ell}$, then $x,a_{1}\ldots a_{\ell
},c_{1}\ldots c_{\ell}$,\ then $x,b_{1}\ldots b_{\ell},c_{1}\ldots c_{\ell}$.
\end{proof}

More generally:

\begin{lemma}
\label{makenotnot}Let $G$ be any gate of the form $\operatorname{NOT}\otimes
H$. \ Then $G$ generates $\operatorname{NOTNOT}$.
\end{lemma}

\begin{proof}
To implement $\operatorname*{NOTNOT}\left(  x,y\right)  $, first apply $G$\ to
$x,a$\ where $a$ is some ancilla string; then apply $G^{-1}$\ to $y,a$.
\end{proof}

Also:

\begin{lemma}
\label{notslin}Let $G\left(  x\right)  =Ax\oplus b$\ be an affine gate. \ Then
$G+\operatorname*{NOTNOT}$\ generates $A$ itself.
\end{lemma}

\begin{proof}
First we use $G^{\otimes2}$\ to map $x,0^{n}$\ to $Ax\oplus b,b$; then we use
$\operatorname*{NOTNOT}$\ gates to map $Ax\oplus b,b$\ to $Ax,0^{n}$.
\end{proof}

By combining Lemmas \ref{makenotnot}\ and \ref{notslin}, we obtain the following.

\begin{corollary}
[Cruft Removal]\label{cruftcor}Let $G\left(  x\right)  =Ax\oplus b$\ be an
$n$-bit affine gate. \ Suppose $A$\ applies a linear transformation
$A^{\prime}$\ to the first $m$ bits of $x$, and acts as the identity on the
remaining $n-m$\ bits. \ Then $G$ generates an $m$-bit gate of the form
$H\left(  x\right)  =A^{\prime}x\oplus c$.
\end{corollary}

\begin{proof}
If $b_{i}=0$\ for all $i>m$, then we are done. \ Otherwise, we can use Lemma
\ref{makenotnot} to generate $\operatorname*{NOTNOT}$, and then Lemma
\ref{notslin} to generate $H\left(  x\right)  =A^{\prime}x$.
\end{proof}

\begin{lemma}
\label{addnots}Let $S$ be any class of parity-preserving linear or affine
gates. \ Then there are no classes between $\left\langle S\right\rangle $ and
$\left\langle S+\operatorname{NOT}\right\rangle $ other than $\left\langle
S+\operatorname{NOTNOT}\right\rangle $.
\end{lemma}

\begin{proof}
Let $G$ be a transformation that is generated by $S+\operatorname{NOT}$\ but
not by $S$. \ Then we need to show how to generate $\operatorname{NOT}$ or
$\operatorname{NOTNOT}$ themselves using $S+G$.

We claim that$\ G$ acts as%
\[
G\left(  x\right)  =V\left(  x\right)  \oplus b,
\]
where $V\left(  x\right)  $ is some parity-preserving affine transformation
generated by $S$, and $b$ is some nonzero string. \ First, $V$ must be
generated by $S$ because, given any circuit for $G$ over the set
$S+\operatorname{NOT}$, we can always push the $\operatorname{NOT}$ gates to
the end; this leaves us with a circuit for the \textquotedblleft$S$
part\textquotedblright\ of $G$. \ (This is the one place where we use that
$S$\ is affine.) \ Also, $b$ must be nonzero because otherwise, $G$ would
already be generated by $S$.

Given $x$, suppose we first apply $V^{-1}$ (which must be generated by $S$),
then apply $G$. \ This yields%
\[
G\left(  V^{-1}\left(  x\right)  \right)  =V\left(  V^{-1}\left(  x\right)
\right)  \oplus b=x\oplus b,
\]
which is equivalent to $\operatorname{NOT}^{\otimes k}$\ for some nonzero $k$.
\ By Lemma \ref{notlem}, this generates $\operatorname{NOTNOT}$. \ If
$\left\vert b\right\vert $\ is always even, then since $V$\ is
parity-preserving, clearly we remain within $\left\langle
S+\operatorname{NOTNOT}\right\rangle $. \ If, on the other hand, $\left\vert
b\right\vert $\ is ever odd, then again by Lemma \ref{notlem}, we can generate
$\operatorname{NOT}$.
\end{proof}

We can finally complete the proof of Theorem \ref{affinedone}, characterizing
the possible affine classes.

\begin{proof}
[Proof of Theorem \ref{affinedone}]If we restrict ourselves to the linear part
of the class, then we know from Corollary \ref{scottcor} that the only
possibilities are $\left\langle \operatorname{CNOT}\right\rangle $,
$\left\langle \operatorname{CNOTNOT}\right\rangle $, $\left\langle
\operatorname{T}_{4}\right\rangle $, $\left\langle \operatorname{T}%
_{6}\right\rangle $, and $\left\langle \varnothing\right\rangle $ (i.e., the
trivial class). \ We will handle these possibilities one by one.

\textbf{Linear part is }$\left\langle \operatorname{CNOT}\right\rangle $.
\ Since $\operatorname{CNOT}$ can already generate all affine transformations
(by Theorem \ref{cnotcirc}), using an ancilla bit initialized to $1$, we have
$\left\langle S\right\rangle \subseteq\left\langle \operatorname{CNOT}%
\right\rangle $. \ For the other direction, Corollary \ref{cruftcor} implies
that $S$\ must generate a gate of the form $\operatorname{CNOT}\left(
x\right)  \oplus b$, for some $b\in\left\{  0,1\right\}  ^{2}$. \ However, it
is not hard to see that all such gates can generate $\operatorname{CNOT}$ itself.

\textbf{Linear part is }$\left\langle \operatorname{CNOTNOT}\right\rangle $.
\ Here we clearly have $\left\langle S\right\rangle \subseteq\left\langle
\operatorname{CNOTNOT},\operatorname{NOT}\right\rangle $. \ Meanwhile,
Corollary \ref{cruftcor} again implies that $S$\ generates a gate of the form
$G\left(  x\right)  =\operatorname{CNOTNOT}\left(  x\right)  \oplus b$, for
some $b\in\left\{  0,1\right\}  ^{3}$. \ Suppose the first bit of $b$ is $1$;
this is the bit that corresponds to the control of the $\operatorname{CNOTNOT}%
$. \ Then $G\left(  G\left(  x\right)  \right)  $ generates
$\operatorname{NOTNOT}$, so by Lemma~\ref{notslin}, we can generate
$\operatorname{CNOTNOT}$. \ If, on the other hand, the first bit of $b$ is
$0$, then $G$ generates $\operatorname{NOT}$ or $\operatorname{NOTNOT}$
directly, so we can again use Lemma~\ref{notslin} to generate
$\operatorname{CNOTNOT}$. \ Therefore $\left\langle S\right\rangle $ lies
somewhere between $\left\langle \operatorname{CNOTNOT}\right\rangle $ and
$\left\langle \operatorname{CNOTNOT},\operatorname{NOT}\right\rangle $.\ \ But
since $\operatorname{CNOTNOT}$\ already generates\ $\operatorname{NOTNOT}$,
Lemma~\ref{addnots} says that the only possibilities are $\left\langle
\operatorname{CNOTNOT}\right\rangle $ and $\left\langle \operatorname{CNOTNOT}%
,\operatorname{NOT}\right\rangle $.

\textbf{Linear part is }$\left\langle \operatorname{T}_{4}\right\rangle $.
\ In this case $\left\langle S\right\rangle \subseteq\left\langle
\operatorname{T}_{4},\operatorname{NOT}\right\rangle $. \ Again Corollary
\ref{cruftcor} implies that $S$ generates a gate of the form $G\left(
x\right)  =\operatorname{T}_{4}\left(  x\right)  \oplus b$, for some
$b\in\left\{  0,1\right\}  ^{4}$. \ If $b=1111$, then $S$\ generates
$\operatorname*{F}_{4}$.\ \ So $\left\langle S\right\rangle $ lies somewhere
between $\left\langle \operatorname*{F}_{4}\right\rangle $\ and $\left\langle
\operatorname*{F}_{4},\operatorname{NOT}\right\rangle =\left\langle
\operatorname{T}_{4},\operatorname{NOT}\right\rangle $, but then
Lemma~\ref{addnots} ensures that $\left\langle \operatorname*{F}%
_{4}\right\rangle $, $\left\langle \operatorname*{F}_{4},\operatorname{NOTNOT}%
\right\rangle =\left\langle \operatorname{T}_{4},\operatorname{NOTNOT}%
\right\rangle $, and $\left\langle \operatorname{T}_{4},\operatorname{NOT}%
\right\rangle $\ are the only possibilities. \ Likewise, if $b=0000$, then
$S$\ generates $\operatorname{T}_{4}$, so $\left\langle \operatorname{T}%
_{4}\right\rangle $, $\left\langle \operatorname{T}_{4},\operatorname{NOTNOT}%
\right\rangle $, and $\left\langle \operatorname{T}_{4},\operatorname{NOT}%
\right\rangle $ are the only possibilities.

Next suppose $\left\vert b\right\vert $\ is odd. \ Then $G\left(  G\left(
x\right)  \right)  =\operatorname{NOT}^{\otimes4}\left(  x\right)  $, which
generates $\operatorname{NOTNOT}$\ by Lemma \ref{notlem}. \ So by Lemma
\ref{notslin}, we generate $\operatorname{T}_{4}$\ as well. \ Thus we have at
least $\left\langle \operatorname{T}_{4},\operatorname{NOTNOT}\right\rangle $.
\ But since $G$\ itself is parity-flipping, $\left\langle S\right\rangle $\ is
not parity-preserving, leaving $\left\langle \operatorname{T}_{4}%
,\operatorname{NOT}\right\rangle $\ as the only possibility by Lemma
\ref{addnots}. \ Finally suppose $\left\vert b\right\vert =2$: without loss of
generality, $b=1100$. \ Let $Q$\ be an operation that swaps the first two bits
of $x$ with the last two bits. \ Then $G\left(  Q\left(  G\left(  x\right)
\right)  \right)  $ is equivalent to $\operatorname{NOT}^{\otimes4}\left(
x\right)  $ up to swaps, so again we have at least $\left\langle
\operatorname{T}_{4},\operatorname{NOTNOT}\right\rangle $, leaving
$\left\langle \operatorname{T}_{4},\operatorname{NOTNOT}\right\rangle $ and
$\left\langle \operatorname{T}_{4},\operatorname{NOT}\right\rangle $\ as the
only possibilities.

\textbf{Linear part is }$\left\langle \operatorname{T}_{6}\right\rangle $.
\ In this case $\left\langle S\right\rangle \subseteq\left\langle
\operatorname{T}_{6},\operatorname{NOT}\right\rangle $. \ Again, Corollary
\ref{cruftcor} implies that $S$ generates $G(x)=\operatorname{T}_{6}\left(
x\right)  \oplus b$ for some $b\in\left\{  0,1\right\}  ^{6}$. \ If
$b=000000$, then $S$\ generates $\operatorname{T}_{6}$, so $\left\langle
\operatorname{T}_{6}\right\rangle $, $\left\langle \operatorname{T}%
_{6},\operatorname{NOTNOT}\right\rangle $, and $\left\langle \operatorname{T}%
_{6},\operatorname{NOT}\right\rangle $ are the only possibilities by Lemma
\ref{addnots}. \ If $\left\vert b\right\vert $ is odd, then $G\left(  G\left(
x\right)  \right)  =\operatorname{NOT}^{\otimes6}\left(  x\right)  $. \ By
Lemma \ref{notlem}, this means that $S$\ generates $\operatorname{NOTNOT}$, so
by Lemma \ref{notslin}, it generates $\operatorname{T}_{6}$\ as well. \ But
$G$ is parity-flipping, leaving $\left\langle \operatorname{T}_{6}%
,\operatorname{NOT}\right\rangle $\ as the only possibility by Lemma
\ref{addnots}. \ If $\left\vert b\right\vert $ is $2$ or $4$, then by an
appropriate choice of swap operation $Q$, we can cause $G\left(  Q\left(
G\left(  x\right)  \right)  \right)  $\ to generate $\operatorname{NOTNOT}$,
so again $\left\langle \operatorname{T}_{6},\operatorname{NOTNOT}\right\rangle
$\ and $\left\langle \operatorname{T}_{6},\operatorname{NOT}\right\rangle
$\ are the only possibilities.

Finally, if $b=111111$, then\ $G(x)=\operatorname*{F}_{6}\left(  x\right)  $.
\ In this case we start with the operation
\[
\operatorname*{F}\nolimits_{6}\left(  x00000\right)  =1\overline{x}%
\overline{x}\overline{x}\overline{x}\overline{x}%
\]
Using three of the $\overline{x}$ outputs and three fresh $0$\ ancilla bits,
we then perform
\[
\operatorname*{F}\nolimits_{6}\left(  \overline{x}\overline{x}\overline
{x}000\right)  =111xxx
\]
Next, bringing\ the $xxx$ outputs together with the remaining $\overline
{x}\overline{x}$\ outputs and one fresh $0$\ ancilla bit, we apply%
\[
\operatorname*{F}\nolimits_{6}\left(  xxx\overline{x}\overline{x}0\right)
=11100\overline{x}%
\]
In summary, we have performed a $\operatorname{NOT}\left(  x\right)  $
operation with some garbage still around. \ However, if we repeat this entire
procedure $6$ times, then the Hamming weight of the garbage will be a multiple
of $6$. \ We can remove all this of garbage using the $\operatorname*{F}_{6}$
gate. \ Therefore, we have created a $\operatorname{NOT}^{\otimes6}$\ gate,
which generates $\operatorname{NOTNOT}$ by Lemma \ref{notlem}. \ So again we
can generate $\operatorname{T}_{6}$, leaving $\left\langle \operatorname{T}%
_{6},\operatorname{NOTNOT}\right\rangle $\ and $\left\langle \operatorname{T}%
_{6},\operatorname*{NOT}\right\rangle $\ as the only possibilities by
Lemma~\ref{notslin}.

\textbf{Linear part is }$\left\langle \varnothing\right\rangle $. \ In this
case $\left\langle S\right\rangle \subseteq\left\langle \operatorname{NOT}%
\right\rangle $, so Lemma \ref{addnots}\ implies that the only possibilities
are $\left\langle \varnothing\right\rangle $, $\left\langle
\operatorname{NOTNOT}\right\rangle $, and $\left\langle \operatorname{NOT}%
\right\rangle $.
\end{proof}

\section{Open Problems\label{OPEN}}

As discussed in Section \ref{INTRO},\ the central challenge we leave is to
give a complete classification of all \textit{quantum} gate sets acting on
qubits, in terms of which unitary transformations they can generate or
approximate. \ Here, just like in this paper, one should assume that
qubit-swaps are free, and that arbitrary ancillas are allowed as long as they
are returned to their initial states.

A possible first step---which would build directly on our results here---would
be to classify all possible quantum gate sets within the \textit{stabilizer
group},\ which is a quantum generalization of the group of affine classical
reversible\ transformations. \ Since the stabilizer group is discrete, here at
least there is no need for representation theory, Lie algebras, or any notion
of approximation, but the problem still seems complicated. \ A different step
in the direction we want, which \textit{would} involve Lie algebras, would be
to classify all sets of $1$- and $2$-qubit gates. \ A third step would be to
classify qubit Hamiltonians (i.e., the infinitesimal-time versions of unitary
gates), in terms of which $n$-qubit Hamiltonians they can be used to generate.
\ Here the recent work of Cubitt and Montanaro \cite{cubitt:qma}, which
classifies qubit Hamiltonians in terms of the complexity of approximating
ground state energies, might be relevant. \ Yet a fourth possibility would be
to classify quantum gates under the assumption that intermediate measurements
are allowed. \ Of course, these simplifications can also be combined.

On the classical side, we have left completely open the problem of classifying
reversible gate sets over non-binary alphabets. \ In the non-reversible
setting, it was discovered in the 1950s (see \cite{lau}) that Post's lattice
becomes dramatically different and more complicated when we consider gates
over a $3$-element set rather than Boolean gates: for example, there is now an
uncountable infinity of clones, rather than \textquotedblleft
merely\textquotedblright\ a countable infinity. \ Does anything similar happen
in the reversible case? \ Even for reversible gates over (say) $\left\{
0,1,2\right\}  ^{n}$, we cannot currently give an algorithm to decide whether
a given gate $G$ generates another gate $H$\ any better than the
triple-exponential-time algorithm that comes from clone theory, nor can we
give reasonable upper bounds on the number of gates or ancillas needed in the
generating circuit, nor can we answer basic questions like whether every class
is finitely generated.

Finally, can one reduce the number of gates in each of our circuit
constructions to the limits imposed by Shannon-style counting arguments?
\ What are the tradeoffs, if any, between the number of gates and the number
of ancilla bits?

\section{Acknowledgments}

At the very beginning of this project, Emil Je\v{r}\'{a}bek \cite{jerabek}%
\ brought the $\left\langle \operatorname*{C}_{k}\right\rangle $ and
$\left\langle \operatorname*{T}_{6}\right\rangle $\ classes to our attention,
and also proved that every reversible gate class is characterized by
invariants (i.e., that the \textquotedblleft clone-coclone
duality\textquotedblright\ holds for reversible gates). \ Also, Matthew Cook
gave us encouragement, asked pertinent questions, and helped us understand the
$\left\langle \operatorname*{T}_{4}\right\rangle $\ class. \ We are grateful
to both of them. \ We also thank Adam Bouland, Seth Lloyd, Igor Markov, and
particularly Siyao Xu for helpful discussions.

\bibliographystyle{plain}
\bibliography{thesis}

\section{Appendix: Post's Lattice with Free Constants\label{POST}}

For completeness, in this appendix we prove a `quick-and-dirty'\ version of
Post's 1941 classification theorem \cite{post}, for sets of ordinary
(non-reversible) Boolean logic gates.

\begin{figure}[ptb]
\begin{center}
\begin{tikzpicture}[>=latex]
\tikzstyle{class}=[circle, thick, minimum size=1.2cm, text width=1.0cm, align=center, draw, font=\tiny]
\tikzstyle{all}=[class,fill=blue!20]
\tikzstyle{affine}=[class,fill=green!20]
\tikzstyle{monotone}=[class,fill=yellow!20]
\tikzstyle{none}=[class,fill=red!20]
\tikzstyle{invis}=[class,draw=none]
\matrix[row sep=1cm,column sep=1.4cm] {
& \node (ALL) [all] {$\top$}; & \\
\node (XOR) [affine] {$\mathsf{XOR}$}; & \node (MONO1) {}; & \node (MONO2) {}; \\
\node (NOT) [affine] {$\mathsf{NOT}$}; & \node (AND) [monotone] {$\mathsf{AND}$}; & \node (OR) [monotone] {$\mathsf{OR}$}; \\
& \node (NONE1) [invis] {}; & \node (NONE2) [invis] {}; \\
};
\node (MONO) [monotone] at ($(MONO1)!0.5!(MONO2)$) {$\mathsf{MONO}$};
\node (NONE) [none] at ($(NONE1)!0.5!(NONE2)$) {$\bot$};
\path[draw,->] (ALL) edge (XOR)
(ALL) edge (MONO)
(MONO) edge (AND)
(MONO) edge (OR)
(AND) edge (NONE)
(OR) edge (NONE)
(XOR) edge (NOT)
(NOT) edge (NONE);
\end{tikzpicture}
\par
\label{postlitefig}
\end{center}
\caption{\textquotedblleft Post's Lattice Lite\textquotedblright}%
\end{figure}
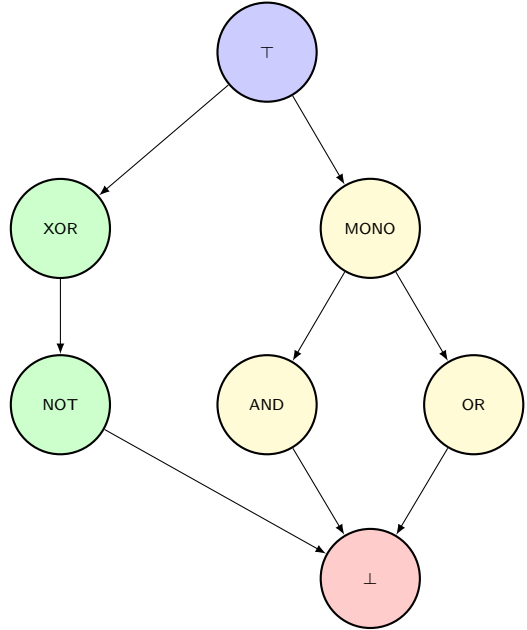

\begin{theorem}
[Post's Lattice Lite]\label{postlite}Assume the constant functions $f=0$ and
$f=1$, as well as the identity function $f\left(  x\right)  =x$, are available
for free. \ Then the only Boolean clones (i.e., classes of Boolean functions
$f:\left\{  0,1\right\}  ^{n}\rightarrow\left\{  0,1\right\}  $\ closed under
composition and addition of dummy variables) are the following:

\begin{enumerate}
\item The trivial class (containing the constant and identity functions).

\item The $\operatorname*{AND}$\ class.

\item The $\operatorname*{OR}$\ class.

\item The class of monotone functions (generated by $\left\{
\operatorname*{AND},\operatorname*{OR}\right\}  $).

\item The $\operatorname*{NOT}$\ class.

\item The class of affine functions (generated by $\left\{
\operatorname*{XOR},\operatorname*{NOT}\right\}  $).

\item The class of all Boolean functions (generated by $\left\{
\operatorname*{AND},\operatorname*{NOT}\right\}  $).
\end{enumerate}
\end{theorem}

\begin{proof}
We take it as known that $\left\{  \operatorname*{AND},\operatorname*{OR}%
\right\}  $\ generates all monotone functions, $\left\{  \operatorname*{XOR}%
,\operatorname*{NOT}\right\}  $ generates all affine functions, and $\left\{
G,\operatorname*{NOT}\right\}  $\ generates all functions, for any $2$-bit
non-affine gate $G$.

Let $\mathcal{C}$ be a Boolean clone that contains the constant $0$ and $1$
functions. \ Then $\mathcal{C}$ is closed under restrictions (e.g., if
$f\left(  x,y\right)  \in\mathcal{C}$, then $f\left(  0,y\right)  $\ and
$f\left(  x,1\right)  $\ are also in $\mathcal{C}$), and that is the crucial
fact we exploit.

First suppose $\mathcal{C}$ contains a non-monotone gate. \ Then certainly we
can construct a $\operatorname{NOT}$ gate by restricting inputs. \ If, in
addition, $\mathcal{C}$\ contains a non-affine gate, then by Proposition
\ref{folkloreprop}, we can construct a $2$-bit\ non-affine gate by restricting
inputs: $\operatorname{AND}$, $\operatorname*{OR}$, $\operatorname*{NAND}$,
$\operatorname*{NOR}$, $\operatorname{IMPLIES}$, or $\operatorname{NOT}\left(
\operatorname{IMPLIES}\right)  $. \ Together with the $\operatorname{NOT}$
gate, this puts us in class 7. \ If, on the other hand, $\mathcal{C}$ contains
only affine gates, then as long as one of those gates depends on at least two
input bits, by restricting inputs we can construct a $2$-bit\ non-degenerate
affine gate: $\operatorname{XOR}$\ or $\operatorname{NOT}\left(
\operatorname{XOR}\right)  $. \ Together with the $\operatorname{NOT}$ gate,
this puts us in class 6. \ If, on the other hand, every gate depends on only
$1$ input bit, then we are in class 5.

Next suppose $\mathcal{C}$ contains only monotone gates. \ Clearly the only
\textit{affine} monotone gates are trivial. \ Thus, as long as one of the
gates is nontrivial, it is non-affine, so Proposition \ref{folkloreprop} again
implies that we can construct a non-affine $2$-bit\ monotone gate by
restricting inputs: $\operatorname{AND}$\ or $\operatorname*{OR}$. \ If we can
construct only $\operatorname{AND}$\ gates, then we are in class 2; if only
$\operatorname*{OR}$ gates, then we are in class 3; if both, then we are in
class 4. \ If, on the other hand, every gate is trivial, then we are in class 1.
\end{proof}

The simplicity of Theorem \ref{postlite}\ underscores how much more
complicated it is to understand reversible gates than non-reversible gates,
when we impose a similar rule in both cases (i.e., that $0$\ and $1$ constant
or ancilla bits are available for free).

\section{Appendix: The Classification Theorem with Loose Ancillas\label{LOOSE}%
}

\begin{theorem}
\label{loosethm}Under the loose ancilla rule, the only change to
Theorem~\ref{main} is that every $\mathcal{C}+\operatorname*{NOTNOT}$ class
collapses with the corresponding $\mathcal{C}+\operatorname*{NOT}$ class.
\end{theorem}

\begin{proof}
That this collapse happens is clear: under the loose ancilla rule, we can
always simulate a $\operatorname*{NOT}$ gate by applying a
$\operatorname*{NOTNOT}$ gate to the desired bit, as well as to a ``dummy"
ancilla bit that will never be used for any other purpose.

To see that no other collapses happen, we must show that the remaining classes
are distinct. \ Under the usual ancilla rule, the classes are distinct because
for any pair of classes we can find an invariant satisfied by one, but not the
other, to separate the two. \ We would like to do the same for loose ancilla
classes, but invariants under the usual rule need not, \textit{a priori}, be
invariants under the loose ancilla rule. \ More concretely, as we have seen, a
gate set that preserves parity under the usual rule need no longer preserve it
under the loose ancilla rule. \ However, we claim that all the \textit{other}
invariants are also loose ancilla invariants.

Suppose $G\left(  x,a\right)  =\left(  H\left(  x\right)  ,b\right)  $ is a
transformation generated under the loose ancilla rule, where $a$ and $b$ are
constants, so that under the loose ancilla rule, we have also generated $H$.
\ We would like to show that any invariant of $G$ must also hold for $H$, so
let us consider the invariants one by one.

\begin{itemize}
\item If $G$ is mod-$k$-respecting then
\[
\left\vert G\left(  x\right)  \right\vert -\left\vert x\right\vert =\left\vert
H\left(  x\right)  \right\vert -\left\vert x\right\vert +\left\vert
a\right\vert -\left\vert b\right\vert ,
\]
is constant modulo $k$, and hence $\left\vert H\left(  x\right)  \right\vert
-\left\vert x\right\vert $ is constant modulo $k$, so $H\left(  x\right)  $ is
mod-$k$-respecting. \ For $k\geq3$, mod-$k$-respecting is equivalent to
mod-$k$-preserving by Theorem~\ref{noshifter}. \ When $k=2$, we have already
seen that $\operatorname*{NOTNOT}$ collapses with $\operatorname*{NOT}$.

\item If $G$ is conservative then $0=\left\vert G\left(  x\right)  \right\vert
-\left\vert x\right\vert =\left\vert H\left(  x\right)  \right\vert
-\left\vert x\right\vert +\left\vert a\right\vert -\left\vert b\right\vert $
as above. \ If we average over all $x$ and appeal to reversibility, then we
see that $\left\vert a\right\vert -\left\vert b\right\vert $ must be $0$, and
hence $H$ is conservative.

\item If $G$ is affine then
\[
G%
\begin{pmatrix}
x\\
a
\end{pmatrix}
=%
\begin{pmatrix}
M_{11} & M_{12}\\
M_{21} & M_{22}%
\end{pmatrix}%
\begin{pmatrix}
x\\
a
\end{pmatrix}
+%
\begin{pmatrix}
c_{1}\\
c_{2}%
\end{pmatrix}
=%
\begin{pmatrix}
H(x)\\
b
\end{pmatrix}
,
\]
so clearly $H\left(  x\right)  =M_{11}x+M_{12}a+c_{1}$ is affine as well.
\ Since $M_{21}x+M_{22}a+c_{2}=b$ for all $x$, we must have $M_{21}=0$. \ But
this means that if the columns of%
\[%
\begin{pmatrix}
M_{11} & M_{12}\\
M_{21} & M_{22}%
\end{pmatrix}
,
\]
the linear part of $G$, have weight $2$, weight $4$, or are orthogonal, then
the same is true of columns of
\[%
\begin{pmatrix}
M_{11}\\
0
\end{pmatrix}
,
\]
and hence the columns of $M_{11}$ itself. \ In short, if the linear part of
$G$ has any of the properties we are interested in, then so does the linear
part of $H$.

\item If $G$ is orthogonal then $c_{1}=0$ and $c_{2}=0$. \ Recall that
$M_{21}=0$, and since a matrix of the form
\[%
\begin{pmatrix}
A & B\\
0 & C
\end{pmatrix}
\]
has an inverse of the same form, and the inverse of an orthogonal matrix is
its transpose, we see that $M_{12}=0$. \ It follows that $H\left(  x\right)
=M_{11}x+M_{12}a+c_{1}$ is actually just $H\left(  x\right)  =M_{11}x$ when
$G$ is orthogonal, therefore $H$ is orthogonal because $M_{11}$ is orthogonal.
\end{itemize}
\end{proof}

\section{Appendix: Number of Gates Generating Each Class\label{NUMBER}}

In this appendix, we count how many $n$-bit gates belong to each of the
classes of Theorem \ref{main}. \ Let us write $\left\langle G\right\rangle
_{n}$ for the set of $n$-bit gates generated by $G$, and $\#\left\langle
G\right\rangle _{n}$ for the number of $n$-bit gates generated by $G$. \ Then
Theorem~\ref{exactcount} gives the exact number of gates in each class, while
Theorem~\ref{approxcount} gives the asymptotics.

\begin{theorem}
\label{exactcount}Let $n\geq1$ be an integer.

\begin{itemize}
\item The total number of gates is
\[
\#\left\langle \operatorname{Toffoli}\right\rangle _{n}=\left(  2^{n}\right)
!
\]
and the non-affine classes break down as follows. \ For $k\geq3$,
\begin{align*}
\#\left\langle \operatorname{Fredkin},\operatorname{NOT}\right\rangle _{n}  &
=2\left(  2^{n-1}!\right)  ^{2}\\
\#\left\langle \operatorname{Fredkin},\operatorname{NOTNOT}\right\rangle _{n}
&  =\left(  2^{n-1}!\right)  ^{2}\\
\#\left\langle \operatorname*{C}\nolimits_{k}\right\rangle _{n}  &
=\prod_{i=0}^{k-1}\left(  \sum_{j-i\equiv0\left(  \operatorname{mod}k\right)
}\binom{n}{j}\right)  !\\
\#\left\langle \operatorname{Fredkin}\right\rangle _{n}  &  =\prod_{i=0}%
^{n}\left(  \binom{n}{i}!\right)
\end{align*}

\item The total number of affine gates is
\[
\#\left\langle \operatorname{CNOT}\right\rangle _{n}=2^{n\left(  n+1\right)
/2}\prod_{i=1}^{n}\left(  2^{i}-1\right)  .
\]

\item The numbers of parity-preserving and parity-respecting gates are:%
\begin{align*}
\#\left\langle \operatorname{CNOTNOT}\right\rangle _{n}  &  =2^{n\left(
n+1\right)  /2-1}\prod_{i=1}^{n-1}\left(  2^{i}-1\right) \\
\#\left\langle \operatorname{CNOTNOT},\operatorname{NOT}\right\rangle _{n}  &
=2^{n\left(  n+1\right)  /2}\prod_{i=1}^{n-1}\left(  2^{i}-1\right)
\end{align*}

\item The numbers of gates in $\left\langle \varnothing\right\rangle $,
$\left\langle \operatorname*{T}_{6}\right\rangle $, and $\left\langle
\operatorname*{T}_{4}\right\rangle $ are:%
\begin{align*}
\#\left\langle \varnothing\right\rangle _{n}  &  =n!\\
\#\left\langle \operatorname*{T}\nolimits_{4}\right\rangle _{n}  &  =%
\begin{cases}
2^{m^{2}}\prod_{i=1}^{m-1}\left(  2^{2i}-1\right)  , & \text{if $n=2m$}\\
2^{m^{2}}\prod_{i=1}^{m}\left(  2^{2i}-1\right)  , & \text{if $n=2m+1$}%
\end{cases}
\\
\#\left\langle \operatorname*{T}\nolimits_{6}\right\rangle _{n}  &  =%
\begin{cases}
1 & \text{if $n=1$}\\
2^{4m^{2}+1}\prod_{i=1}^{2m}\left(  2^{2i}-1\right)  & \text{if $n=4m+2$}\\
2^{4m^{2}+2m+1}\left(  2^{2m+1}+\left(  -1\right)  ^{m}\right)  \prod
_{i=1}^{2m}\left(  2^{2i}-1\right)  & \text{if $n=4m+3$}\\
2^{4m^{2}-2m+1}\left(  2^{2m-1}-\left(  -1\right)  ^{m}\right)  \prod
_{i=1}^{2m-2}\left(  2^{2i}-1\right)  & \text{if $n=4m\geq4$}\\
2^{4m^{2}-2m+1}\left(  2^{2m}-\left(  -1\right)  ^{m}\right)  \prod
_{i=1}^{2m-1}\left(  2^{2i}-1\right)  & \text{if $n=4m+1\geq5$}%
\end{cases}
\end{align*}
Furthermore,
\[
\#\left\langle \operatorname*{F}\nolimits_{4}\right\rangle _{n}=\#\left\langle
\operatorname*{T}\nolimits_{4}\right\rangle _{n}.
\]

\item For any linear class $\left\langle G\right\rangle $ we have
\begin{align*}
\#\left\langle G,\operatorname{NOT}\right\rangle _{n}  &  =\#\left\langle
G\right\rangle _{n}2^{n}\\
\#\left\langle G,\operatorname{NOTNOT}\right\rangle _{n}  &  =\#\left\langle
G\right\rangle _{n}2^{n-1}%
\end{align*}

\end{itemize}
\end{theorem}

Let us count each class in turn. \ To start, note that an $n$-bit reversible
gate is, by definition, a permutation of $\left\{  0,1\right\}  ^{n}$, so
there are $\left(  2^{n}\right)  !$ gates in total.

Parity-preserving gates map even-weight strings to even-weight strings, and
odd-weight strings to odd-weight strings. \ It follows that there are $\left(
\left(  2^{n-1}\right)  !\right)  ^{2}$ parity-preserving gates. \ Clearly
there are exactly twice as many parity-respecting gates, since we can append a
$\operatorname{NOT}$ gate to any parity-preserving gate to get a
parity-flipping gate, and vice versa.

The mod-$k$-preserving gates (for $k\geq3$) also decompose into a product of
permutations, one for each Hamming weight class modulo $k$. \ This leads to
the formula
\[
\#\left\langle \operatorname*{C}\nolimits_{k}\right\rangle _{n}=\prod
_{i=0}^{k-1}\left(  \sum_{j-i\equiv0\left(  \operatorname{mod}k\right)
}\binom{n}{j}\right)  !
\]
Likewise, for conservative gates, we have%
\[
\#\left\langle \operatorname{Fredkin}\right\rangle _{n}=\prod_{i=0}^{n}\left(
\binom{n}{i}!\right)  .
\]

The linear part of an affine gate is an $n\times n$\ invertible matrix $A$.
\ The number of such matrices is well-known to be%
\[
\prod_{i=0}^{n-1}\left(  2^{n}-2^{i}\right)  =2^{n\left(  n-1\right)  /2}%
\prod_{i=1}^{n}\left(  2^{i}-1\right)  .
\]
There are an additional $2^{n}$ choices for the affine part, so%
\[
\#\left\langle \operatorname{CNOT}\right\rangle _{n}=2^{n\left(  n+1\right)
/2}\prod_{i=1}^{n}\left(  2^{i}-1\right)  .
\]

A parity-preserving affine transformation is an affine transformation on the
$\left(  n-1\right)  $-dimensional subspace of even Hamming-weight vectors,
extended to the entire space by defining the transformation on any odd-weight
vector. \ There are $2^{n\left(  n-1\right)  /2}\prod_{i=1}^{n-1}\left(
2^{i}-1\right)  $ affine transformations on $n-1$ dimensions and $2^{n-1}$
choices of odd-weight vector, so there are%
\[
\#\left\langle \operatorname{CNOTNOT}\right\rangle _{n}=2^{n\left(
n+1\right)  /2-1}\prod_{i=1}^{n-1}\left(  2^{i}-1\right)
\]
parity-preserving affine transformations, and twice as many parity-respecting
affine transformations.

We refer to MacWilliams \cite{macwilliams} for the formula (below) for the
number of orthogonal $n\times n$ matrices.
\[
\#\left\langle \operatorname*{T}\nolimits_{4}\right\rangle _{n}=%
\begin{cases}
2^{m^{2}}\prod_{i=1}^{m-1}\left(  2^{2i}-1\right)  , & \text{if $n=2m$,}\\
2^{m^{2}}\prod_{i=1}^{m}\left(  2^{2i}-1\right)  , & \text{if $n=2m+1$.}%
\end{cases}
\]
We now turn our attention to counting $\left\langle \operatorname{T}%
_{6}\right\rangle _{n}$, which is more involved. \ The approach will be
similar to that of MacWilliams \cite{macwilliams}. \ It will help to consider
$\left\langle \operatorname{T}_{4}\right\rangle _{n}$ and $\left\langle
\operatorname{T}_{6}\right\rangle _{n}$ as groups. \ Indeed, $\left\langle
\operatorname{T}_{4}\right\rangle _{n}$ is just the orthogonal group
$\mathrm{O}(n)$ over $\mathbb{F}_{2}$, and $\left\langle \operatorname{T}%
_{6}\right\rangle _{n}$ is a proper subgroup.

The idea is to find a unique representative for each of the cosets of
$\left\langle \operatorname{T}_{6}\right\rangle _{n}$ in $\left\langle
\operatorname{T}_{4}\right\rangle _{n}$. \ Since we know $\#\left\langle
\operatorname{T}_{4}\right\rangle _{n}$ by \cite{macwilliams}, dividing by the
number of unique representatives will give us $\#\left\langle \operatorname{T}%
_{6}\right\rangle _{n}$ as desired.

Recall that by Lemma~\ref{ipcond}, the Hamming weight of each column vector of
an orthogonal matrix is either $1\operatorname{mod}4$\ or $3\operatorname{mod}%
4$. \ If $A\in\left\langle \operatorname{T}_{4}\right\rangle _{n}$ is an
orthogonal matrix with column vectors $a_{1},\ldots,a_{n}$, then the
\textit{characteristic vector} $c\left(  A\right)  $ is an $n$-dimensional
vector whose $i^{th}$ entry, $c_{i}(A)$, is defined as follows:
\[
c_{i}\left(  A\right)  =%
\begin{cases}
1 & \text{if }\left\vert a_{i}\right\vert \equiv3\left(  \operatorname{mod}%
4\right) \\
0 & \text{if }\left\vert a_{i}\right\vert \equiv1\left(  \operatorname{mod}%
4\right)
\end{cases}
.
\]
The following lemma shows that these characteristic vectors can be used as a
representatives for the cosets of $\left\langle \operatorname{T}%
_{6}\right\rangle _{n}$.

\begin{lemma}
\label{uniquerep}Two orthogonal transformations, $A,B\in\left\langle
\operatorname{T}_{4}\right\rangle _{n}$, are in the same coset of
$\left\langle \operatorname{T}_{6}\right\rangle _{n}$ if and only if $c\left(
A\right)  =c\left(  B\right)  $.
\end{lemma}

\begin{proof}
Note that $A$ and $B$ are in the same coset if and only if $T:=BA^{-1}=BA^{T}$
is in $\left\langle \operatorname*{T}\nolimits_{6}\right\rangle _{n}$. \ We
know that $T\in\left\langle \operatorname*{T}\nolimits_{4}\right\rangle _{n}$,
and that $T\left(  a_{i}\right)  =b_{i}$ for all $i$. \ Since $a_{1}%
,\ldots,a_{n}$ is an orthogonal basis, Theorem~\ref{t6circ} says that
$T\in\left\langle \operatorname*{T}\nolimits_{6}\right\rangle _{n}$ if and
only if $T$ is mod-$4$-preserving. \ By Theorem~\ref{affine4cond}, this holds
if and only if $\left\vert a_{i}\right\vert \equiv\left\vert b_{i}\right\vert
\left(  \operatorname{mod}4\right)  $ for all $i$, or equivalently, $c\left(
A\right)  =c\left(  B\right)  $.
\end{proof}

Lemma \ref{uniquerep}\ shows that it suffices to count the number of possible
characteristic vectors. \ Perhaps surprisingly, not every characteristic
vector is achievable;\ the following lemma shows exactly which ones are.

\begin{lemma}
\label{charvec}If $A\in\left\langle \operatorname{T}_{4}\right\rangle _{n}$,
then $\left\vert c\left(  A\right)  \right\vert \equiv0\left(
\operatorname{mod}4\right)  $. \ Furthermore, for every characteristic vector
$c$ such that $\left\vert c\right\vert \equiv0\left(  \operatorname{mod}%
4\right)  $, there exists a matrix $A\in\left\langle \operatorname{T}%
_{4}\right\rangle _{n}$ such that $c(A)=c$.
\end{lemma}

\begin{proof}
Let $A\in\left\langle \operatorname{T}_{4}\right\rangle _{n}$ with column
vectors $a_{1},\ldots,a_{n}$. \ Of course, $A$ might not preserve Hamming
weight mod $4$. \ The main idea of the proof is to promote $A$ to an affine
function $f\left(  x\right)  =Ax\oplus b$ that \textit{does} preserve Hamming
weight mod $4$. \ We know that such a function exists because we can decompose
$A$ into a circuit of $\operatorname{T}_{4}$ gates by Theorem~\ref{t4circ}.
\ Replacing each such gate with $\operatorname{F}_{4}$ will yield a circuit of
the desired form that preserves Hamming weight mod $4$.

Recall from Theorem~\ref{affine4cond} that if $f$ preserves Hamming weight mod
$4$, then $\left\vert a_{i}\right\vert +2\left(  a_{i}\cdot b\right)
\equiv1\left(  \operatorname{mod}4\right)  $. \ Expanding out this condition
we get
\[
a_{i}\cdot b\equiv%
\begin{cases}
1\left(  \operatorname{mod}2\right)  & \text{if }\left\vert a_{i}\right\vert
\equiv3\left(  \operatorname{mod}4\right) \\
0\left(  \operatorname{mod}2\right)  & \text{if }\left\vert a_{i}\right\vert
\equiv1\left(  \operatorname{mod}4\right)
\end{cases}
,
\]
which is equivalent to the condition $A^{T}b=c(A)$. \ Therefore,
\[
\left\vert b\right\vert =\left\vert Ac\left(  a\right)  \right\vert
=\left\vert \sum_{i=1}^{n}a_{i}c_{i}\left(  A\right)  \right\vert \equiv
\sum_{i=1}^{n}c_{i}\left(  A\right)  \left\vert a_{i}\right\vert +2\sum
_{i<j}c_{i}\left(  A\right)  c_{j}\left(  A\right)  \left(  a_{i}\cdot
a_{j}\right)  \equiv\sum_{i=1}^{n}c_{i}\left(  A\right)  \left\vert
a_{i}\right\vert \equiv3\left\vert c\left(  a\right)  \right\vert \left(
\operatorname{mod}4\right)  ,
\]
which implies that $\left\vert b\right\vert \equiv\left\vert c\left(
A\right)  \right\vert \left(  \operatorname{mod}4\right)  $. \ But we know by
Theorem~\ref{affine4cond} that $\left\vert b\right\vert \equiv0\left(
\operatorname{mod}4\right)  $. So $\left\vert c\left(  A\right)  \right\vert
\equiv0\left(  \operatorname{mod}4\right)  $, which completes the first part
of the lemma.

We now need to show that any characteristic vector of Hamming weight divisible
by $4$ is realized by some matrix $A\in\left\langle \operatorname{T}%
_{4}\right\rangle _{n}$. \ Notice that $c\left(  \operatorname{T}_{4}\right)
=\left(  1,1,1,1\right)  $. \ Therefore, by taking an appropriate tensor
product of $\operatorname{T}_{4}$ gates and permuting the rows and columns, we
can achieve any characteristic vector of Hamming weight divisible by $4$.
\end{proof}

\begin{corollary}
$\#\left\langle \operatorname{F}_{4}\right\rangle _{n}=\#\left\langle
\operatorname{T}_{4}\right\rangle _{n}$.
\end{corollary}

\begin{proof}
The condition $A^{T}b=c(A)$ in the proof of Lemma~\ref{charvec} implies that
there is a unique vector $b=Ac(A)$ such that $f(x)=A(x)\oplus b$ is mod-$4$-preserving.
\end{proof}

Combining Lemmas~\ref{uniquerep} and \ref{charvec}, we find that the number of
representatives for the cosets of $\left\langle \operatorname{T}%
_{6}\right\rangle _{n}$ in $\left\langle \operatorname{T}_{4}\right\rangle
_{n}$ equals the number of $n$-bit strings with Hamming weight $4$. \ An
explicit formula for this quantity is given by Knuth \cite[p. 70]{knuth:vol1}.
\ This now completes the proof of Theorem~\ref{exactcount}.

Table~\ref{table:counting} gives the number of generators of each class for
$3\leq n\leq7$.

\begin{table}[h]
\scalebox{.9}{
\parbox{\textwidth}{
\begin{tabular}{l |  l l l l l}
& \multicolumn{1}{c}{$n=3$} & \multicolumn{1}{c}{$n=4$} & \multicolumn{1}{c}{$n=5$} & \multicolumn{1}{c}{$n=6$} & \multicolumn{1}{c}{$n=7$} \\
\hline$\left<\operatorname{Toffoli}\right>$ & $\num{37980}$ & $\num{20919528228864}$ & $2.6313\times 10^{35}$ & $1.2689\times 10^{89}$ & $3.8562\times 10^{215}$  \\
$\left<\operatorname{Fredkin, NOT}\right>$ & $\num{480}$ & $\num{1625691648}$ & $4.3776\times 10^{26}$ & $6.9238\times 10^{70}$ & $1.6100\times 10^{178}$  \\
$\left<\operatorname{Fredkin, NOTNOT}\right>$ & $\num{450}$ & $\num{1624862256}$ & $4.3776\times 10^{26}$ & $6.9238\times 10^{70}$ & $1.6100\times 10^{178}$  \\
$\left<\operatorname{C_{3}}\right>$ & $\num{36}$ & $\num{9953280}$ & $5.7818\times 10^{21}$ & $2.9340\times 10^{60}$ & $5.1283\times 10^{156}$  \\
$\left<\operatorname{C_{4}}\right>$ & $\num{0}$ & $\num{414696}$ & $6.6368\times 10^{18}$ & $5.1015\times 10^{53}$ & $1.2863\times 10^{142}$  \\
$\left<\operatorname{C_{5}}\right>$ & $\num{0}$ & $\num{0}$ & $1.8962\times 10^{17}$ & $1.0352\times 10^{50}$ & $1.1760\times 10^{133}$  \\
$\left<\operatorname{C_{6}}\right>$ & $\num{0}$ & $\num{0}$ & $\num{0}$ & $2.1567\times 10^{48}$ & $4.4602\times 10^{128}$  \\
$\left<\operatorname{C_{7}}\right>$ & $\num{0}$ & $\num{0}$ & $\num{0}$ & $\num{0}$ & $7.0797\times 10^{126}$  \\
$\left<\operatorname{Fredkin}\right>$ & $\num{30}$ & $\num{414696}$ & $1.8962\times 10^{17}$ & $2.1567\times 10^{48}$ & $7.0797\times 10^{126}$  \\
$\left<\operatorname{CNOT}\right>$ & $\num{1152}$ & $\num{301056}$ & $\num{309657600}$ & $\num{1269678735360}$ & $\num{20807658944593920}$  \\
$\left<\operatorname{CNOTNOT,NOT}\right>$ & $\num{72}$ & $\num{10368}$ & $\num{5149440}$ & $\num{10238607360}$ & $\num{82569982279680}$  \\
$\left<\operatorname{CNOTNOT}\right>$ & $\num{72}$ & $\num{10368}$ & $\num{5149440}$ & $\num{10238607360}$ & $\num{82569982279680}$  \\
$\left<\operatorname{T_4, NOT}\right>$ & $\num{0}$ & $\num{192}$ & $\num{9600}$ & $\num{691200}$ & $\num{90316800}$  \\
$\left<\operatorname{T_4, NOTNOT}\right>$ & $\num{0}$ & $\num{144}$ & $\num{8400}$ & $\num{648000}$ & $\num{87494400}$  \\
$\left<\operatorname{T_6, NOT}\right>$ & $\num{0}$ & $\num{0}$ & $\num{0}$ & $\num{23040}$ & $\num{2257920}$  \\
$\left<\operatorname{T_6, NOTNOT}\right>$ & $\num{0}$ & $\num{0}$ & $\num{0}$ & $\num{22320}$ & $\num{2222640}$  \\
$\left<\operatorname{T_4}\right>, \left<\operatorname{F_4}\right>$ & $\num{0}$ & $\num{24}$ & $\num{600}$ & $\num{21600}$ & $\num{1411200}$  \\
$\left<\operatorname{T_6}\right>$ & $\num{0}$ & $\num{0}$ & $\num{0}$ & $\num{720}$ & $\num{35280}$  \\
$\left<\operatorname{NOT}\right>$ & $\num{24}$ & $\num{192}$ & $\num{1920}$ & $\num{23040}$ & $\num{322560}$  \\
$\left<\operatorname{NOTNOT}\right>$ & $\num{18}$ & $\num{168}$ & $\num{1800}$ & $\num{22320}$ & $\num{317520}$  \\
$\left<\varnothing\right>$ & $\num{6}$ & $\num{24}$ & $\num{120}$ & $\num{720}$ & $\num{5040}$  \\
\end{tabular}
}}\caption{Number of $n$-bit generators for each reversible gate class.}%
\label{table:counting}%
\end{table}

\begin{theorem}
\label{approxcount}The asymptotic size of each reversible gate class is as
follows.%
\begin{align*}
\log_{2}\#\left\langle \operatorname{Toffoli}\right\rangle _{n}  &
=n2^{n}-\frac{2^{n}}{\ln2}+\frac{n}{2}+\frac{1}{2}\log_{2}2\pi+O(2^{-n})\\
\log_{2}\#\left\langle \operatorname{Fredkin},\operatorname{NOTNOT}%
\right\rangle _{n}  &  =n2^{n}-\frac{2^{n}}{\ln2}-2^{n}+n\log_{2}n+\log_{2}%
\pi+O(2^{-n})\\
\log_{2}\#\left\langle \operatorname{Fredkin},\operatorname{NOT}\right\rangle
_{n}  &  =\log_{2}\#\left\langle \operatorname{Fredkin},\operatorname{NOTNOT}%
\right\rangle _{n}+1\\
\log_{2}\#\left\langle \operatorname*{C}\nolimits_{k}\right\rangle _{n}  &
=n2^{n}-\frac{2^{n}}{\ln2}-2^{n}\log_{2}k+o( 2^{n} )\\
\log_{2}\#\left\langle \operatorname{Fredkin}\right\rangle _{n}  &
=n2^{n}-\frac{2^{n}}{\ln2}-2^{n}\log_{2}\frac{\pi e\sqrt{n}}{2}+o(2^{n})\\
\log_{2}\#\left\langle \operatorname{CNOT}\right\rangle _{n}  &  =n\left(
n+1\right)  -\alpha+O(2^{-n})\\
\log_{2}\#\left\langle \operatorname{CNOTNOT},\operatorname{NOT}\right\rangle
_{n}  &  =n\left(  n-1\right)  -\alpha+O(2^{-n})\\
\log_{2}\#\left\langle \operatorname{CNOTNOT}\right\rangle _{n}  &  =\log
_{2}\#\left\langle \operatorname{CNOTNOT},\operatorname{NOT}\right\rangle
_{n}-1\\
\log_{2}\#\left\langle \varnothing\right\rangle _{n}  &  =n\log_{2}n-\frac
{n}{\ln2}+\frac{1}{2}\log_{2}2\pi+O\mathopen{}\left(  \frac{1}{n}\right)
\mathclose{}\\
\log_{2}\#\left\langle \operatorname*{T}\nolimits_{4}\right\rangle _{n}  &
=\frac{n(n-1)}{2}-\beta+O( 2^{-n} )\\
\log_{2}\#\left\langle \operatorname*{T}\nolimits_{6}\right\rangle _{n}  &
=\frac{n^{2}-3n+4}{2}-\beta+O\mathopen{}\left(  2^{-n/2}\right)  \mathclose{}
,
\end{align*}
where
\begin{align*}
\alpha &  =-\sum_{i=1}^{\infty}\log_{2}(1-2^{-i})\approx1.7919,\\
\beta &  =-\sum_{i=1}^{\infty}\log_{2}(1-2^{-2i})\approx0.53839.
\end{align*}
Recall that $\#\left\langle \operatorname{F}_{4}\right\rangle _{n}%
=\#\left\langle \operatorname*{T}\nolimits_{4}\right\rangle _{n}$. \ The
asymptotics of the remaining affine classes follow from the rules
\begin{align*}
\log_{2}\#\left\langle G,\operatorname{NOT}\right\rangle _{n}  &  =n+\log
_{2}\#\left\langle G\right\rangle _{n},\\
\log_{2}\#\left\langle G,\operatorname{NOTNOT}\right\rangle _{n}  &
=n-1+\log_{2}\#\left\langle G\right\rangle _{n},
\end{align*}
where $\left\langle G\right\rangle $ is a linear class.
\end{theorem}

\begin{proof}
Most of these results follow directly from Theorem~\ref{exactcount} with
liberal use of well-known logarithm properties, especially Stirling's
approximation:
\[
\log_{2}(m!)=m\log_{2}m-\frac{m}{\ln2}+\frac{1}{2}\log_{2}2\pi m+O\left(
\frac{1}{m}\right)  .
\]
For the affine classes, we use the fact that
\begin{align*}
\sum_{i=1}^{m}\log_{2}(2^{i}-1)  &  =\frac{m(m+1)}{2}+\sum_{i=1}^{m}\log
_{2}(1-2^{-i})\\
&  =\frac{m(m+1)}{2}-\alpha+O(2^{-m})
\end{align*}
where $\alpha=-\sum_{i=1}^{\infty}\log_{2}\left(  1-2^{-i}\right)  $. \ Note
that $\alpha=-\log_{2}\left(  1/2;1/2\right)  _{\infty}$ where
$(1/2;1/2)_{\infty}$ is the $q$\textit{-Pochhammer symbol}. \ Similarly,
$\beta:=-\sum_{i=1}^{\infty}\log_{2}\left(  1-2^{-2i}\right)  =-\log
_{2}\left(  1/4;1/4\right)  _{\infty}$ differs from the $m^{th}$ partial sum
by $O(2^{-2m})$.

It turns out that the even and odd cases of $\#\left\langle \operatorname*{T}%
\nolimits_{4}\right\rangle $ have the same asymptotic behavior, and similarly
for the four cases of $\#\left\langle \operatorname*{T}\nolimits_{6}%
\right\rangle $.

However, there are two special cases that require extra care: $\left\langle
\operatorname*{C}\nolimits_{k}\right\rangle $ (for $k\geq3$) and
$\langle\operatorname{Fredkin}\rangle$. \ Recall that
\[
\#\left\langle \operatorname*{C}\nolimits_{k}\right\rangle _{n}=\prod
_{i=0}^{k-1}a_{i}!.
\]
where we define $a_{i}=\sum_{j\equiv i\left(  \operatorname{mod}k\right)
}\binom{n}{j}$. \ Clearly $a_{i}=\frac{2^{n}}{k}\left(  1+o(1)\right)  $.
\ Then Stirling's approximation gives
\begin{align*}
\log_{2}\#\left\langle \operatorname*{C}\nolimits_{k}\right\rangle _{n}  &
=\sum_{i=0}^{k-1}\left(  a_{i}\log_{2}a_{i}-\frac{a_{i}}{\ln2}+o(a_{i})\right)
\\
&  =\sum_{i=0}^{k-1}\left(  a_{i}\log_{2}\frac{2^{n}}{k}+a_{i}\log_{2}\left(
1+o(1)\right)  -\frac{a_{i}}{\ln2}+o(a_{i})\right) \\
&  =n2^{n}-\frac{2^{n}}{\ln2}-2^{n}\log_{2}k+o(2^{n}).
\end{align*}

For $\langle\operatorname{Fredkin}\rangle$, we use the fact that if $x$\ is a
uniformly-random $n$-bit string, then the entropy of $\left\vert x\right\vert
$\ is
\[
\frac{1}{2}\log_{2}\frac{\pi en}{2}+O\left(  \frac{1}{n}\right)  =-\sum
_{i=0}^{n}2^{-n}\binom{n}{i}\log_{2}\left(  2^{-n}\binom{n}{i}\right)  .
\]
One can show this by approximating the binomial with a Gaussian distribution.
\ Rearranging gives us
\[
\sum_{i=0}^{n}\binom{n}{i}\log_{2}\binom{n}{i}=n2^{n}-2^{n}\log_{2}\frac{\pi
e\sqrt{n}}{2}-O\left(  \frac{2^{n}}{n}\right)  .
\]
Now we can apply Stirling's approximation to $\#\left\langle
\operatorname{Fredkin}\right\rangle _{n}$, as calculated in
Theorem~\ref{exactcount}:
\begin{align*}
\log_{2}\#\left\langle \operatorname{Fredkin}\right\rangle _{n}  &
=\sum_{i=0}^{n}\left[  \binom{n}{i}\log_{2}\binom{n}{i}-\binom{n}{i}+o\left(
\binom{n}{i}\right)  \right] \\
&  =n2^{n}-\frac{2^{n}}{\ln2}-2^{n}\log_{2}\frac{\pi e\sqrt{n}}{2}-o(2^{n}).
\end{align*}

\end{proof}

One can clearly see \textquotedblleft the pervasiveness of
universality\textquotedblright\ in Table~\ref{table:counting}: within almost
every class, the gates that are universal for that class quickly come to
dominate the gates that are not universal for that class in number. \ Theorem
\ref{approxcount}\ lets us make that observation rigorous.

\begin{corollary}
Let $\mathcal{C}$ be any reversible gate class, and let $G$\ be an $n$-bit
gate chosen uniformly at random from $\mathcal{C}$. \ Then%
\[
\Pr\left[  G\text{ generates }\mathcal{C}\right]  =1-O\left(  2^{-n}\right)
,
\]
unless $\mathcal{C}$ is one of the \textquotedblleft$\operatorname{NOT}%
$\ classes\textquotedblright\ ($\left\langle \operatorname{Fredkin}%
,\operatorname{NOT}\right\rangle $, $\left\langle \operatorname{F}%
_{4},\operatorname{NOT}\right\rangle $, $\left\langle \operatorname{T}%
_{6},\operatorname{NOT}\right\rangle $, or $\left\langle \operatorname{NOT}%
\right\rangle $), in which case%
\[
\Pr\left[  G\text{ generates }\mathcal{C}\right]  =\frac{1}{2}-O\left(
2^{-n}\right)  .
\]

\end{corollary}

\section{Appendix: Alternate Proofs of Theorems \ref{noshifter} and
\ref{k2or4}\label{ALTPROOF}}

\begin{proof}
[Alternate Proof of Theorem \ref{noshifter}]Suppose $j\not \equiv 0\left(
\operatorname{mod}k\right)  $, and let $q$\ be $j$'s order mod $k$ (that is,
the least positive $i$\ such that $ij\equiv0\left(  \operatorname{mod}%
k\right)  $). \ We first show that $q$ must be a power of $2$. \ For
$i\in\left\{  0,\ldots,q-1\right\}  $, let $S_{i}$ be the set of all
$x\in\left\{  0,1\right\}  ^{n}$ whose Hamming weight satisfies $\left\vert
x\right\vert \equiv ij\left(  \operatorname{mod}k\right)  $. \ Let $q$ be the
number of distinct $S_{i}$'s. \ Now, since the gate $G$ maps everything in
$S_{i}$ to $S_{\left(  i+1\right)  \operatorname{mod}q}$, we have%
\[
\left\vert S_{0}\right\vert =\cdots=\left\vert S_{q-1}\right\vert =\frac
{2^{n}}{q}.
\]
But the above must be an integer.

Observe that, if there existed a $G$ such that $\left\vert G\left(  x\right)
\right\vert \equiv\left\vert x\right\vert +j\left(  \operatorname{mod}%
k\right)  $, where $j$'s order mod $k$ was any positive power of $2$\ (say
$2^{p}$), then the iterated map $G^{2^{p-1}}$ would satisfy%
\[
\left\vert G^{2^{p-1}}\left(  x\right)  \right\vert \equiv\left\vert
x\right\vert +\frac{k}{2}\left(  \operatorname{mod}k\right)  ,
\]
and so would have order \textit{exactly} $2$ mod $k$. \ For that reason, it
suffices to rule out, for all $k\geq2$ and all $n$, the possibility of a
reversible transformation $G$ that satisfies%
\[
\left\vert G\left(  x\right)  \right\vert \equiv\left\vert x\right\vert
+k\left(  \operatorname{mod}2k\right)
\]
for all $x\in\left\{  0,1\right\}  ^{n}$.

To do the above, it is necessary and sufficient to show that there is a
\textquotedblleft cardinality obstruction\textquotedblright\ to any $G$ of the
required form. \ In other words, for all $j\in\left\{  0,\ldots,2k-1\right\}
$, let%
\[
A_{n,j}:=\left\{  x\in\left\{  0,1\right\}  ^{n}:\left\vert x\right\vert
\equiv j\left(  \operatorname{mod}2k\right)  \right\}
\]
be the set of $n$-bit strings of Hamming weight $j$ mod\ $2k$. \ Then the
problem boils down to showing that for all $k\geq2$ and $n$, there exists a
$j<k$ such that $\left\vert A_{n,j}\right\vert \neq\left\vert A_{n,j+k}%
\right\vert $---and therefore, that no mapping from $A_{n,j}$\ to $A_{n,j+k}%
$\ (or vice versa) can be reversible.

This, in turn, can be interpreted as a statement about binomial coefficients:
for all $k\geq2$ and all $n$, there exists a $j$ such that%
\[
\sum_{i=j,j+2k,j+4k,\ldots}\binom{n}{i}\neq\sum_{i=j,j+2k,j+4k,\ldots}%
\binom{n}{i+k}.
\]
A nice way to prove the above statement is by using what we call the
\textit{wraparound Pascal's triangle of width }$2k$: that is, Pascal's
triangle with a periodic boundary condition. \ This is simply an iterative map
on row vectors $\left(  a_{0},\ldots,a_{2k-1}\right)  \in\mathbb{Z}^{2k}$,
obtained by starting from the row $\left(  1,0,\ldots,0\right)  $, then
repeatedly applying the update rule $a_{i}^{\prime}:=a_{i}+a_{\left(
i-1\right)  \operatorname{mod}2k}$\ for all $i$. \ So for example, when
$2k=4$\ we obtain%
\[%
\begin{array}
[c]{cccc}%
1 & 0 & 0 & 0\\
1 & 1 & 0 & 0\\
1 & 2 & 1 & 0\\
1 & 3 & 3 & 1\\
2 & 4 & 6 & 4\\
6 & 6 & 10 & 10\\
16 & 12 & 16 & 20\\
\vdots & \vdots & \vdots & \vdots
\end{array}
\]
It is not hard to see that the $i^{th}$\ entry of the $n^{th}$\ row of the
above \textquotedblleft triangle,\textquotedblright\ encodes $\left\vert
A_{n,i}\right\vert $: that is, the number of $n$-bit strings\ whose Hamming
weights are congruent to $i$\ mod\ $2k$.

So the problem reduces to showing that, when $k\geq2$, no row of the
wraparound Pascal's triangle of width $2k$\ can have the form%
\[
\left(  a_{0},\ldots,a_{k-1},a_{0},\ldots,a_{k-1}\right)  .
\]
That is, no row can consist of the same list of $k$ numbers repeated twice.
\ (Note we \textit{can} get rows that satisfy $a_{i}=a_{i+k}$\ for
\textit{specific} values of $i$: to illustrate, in the width-$4$ case
above,\ we have $a_{1}=a_{3}=4$ in the fifth row, and $a_{0}=a_{2}=16$ in the
seventh row. \ But we need to show that no row can satisfy $a_{i}=a_{i+k}%
$\ for all $i\in\left\{  0,\ldots,k-1\right\}  $\ simultaneously.) \ We prove
this as follows.

Notice that the update rule that defines the wraparound Pascal's triangle,
namely $a_{i}^{\prime}:=a_{i}+a_{\left(  i-1\right)  \operatorname{mod}2k}$,
is just a linear transformation on $\mathbb{R}^{2k}$, corresponding to a
$2k\times2k$ band-diagonal matrix $M$. \ For example, when $k=2$ we have%
\[
M=\left(
\begin{array}
[c]{cccc}%
1 & 1 & 0 & 0\\
0 & 1 & 1 & 0\\
0 & 0 & 1 & 1\\
1 & 0 & 0 & 1
\end{array}
\right)  .
\]
Notice further that $\operatorname*{rank}\left(  M\right)  =2k-1$. \ The image
of $M$ is a $\left(  2k-1\right)  $-dimensional\ subspace $P\leq
\mathbb{R}^{2k}$ (the \textquotedblleft parity-respecting
subspace\textquotedblright), which is defined by the linear equation%
\[
a_{0}+a_{2}+\cdots+a_{2k-2}=a_{1}+a_{3}+\cdots+a_{2k-1}.
\]
Thus, $M$ acts invertibly, as long we restrict to vectors in $P$.

Next, let $D\leq\mathbb{R}^{2k}$ (the \textquotedblleft duplicate
subspace\textquotedblright) be the $k$-dimensional subspace defined by the
$k$\ linear equations%
\[
a_{0}=a_{k},\ldots,a_{k-1}=a_{2k-1}.
\]
Then let $S=P\cap D$\ be the $\left(  k-1\right)  $-dimensional intersection
of the parity-respecting and duplicate subspaces.

Observe that $S$ is an invariant subspace of $M$: that is, if $x\in S$, then
$Mx\in S$. \ But now, using the fact that $M$ acts invertibly within $P$, this
means that the converse also holds: namely, if $x\in P\setminus S$, then
$Mx\in P\setminus S$. \ In other words: as we generate more and more rows of
the wraparound Pascal's triangle, if we're not \textit{already} in $S$ by the
second row (i.e., after the first time we've applied $M$), then we're never
going to get into $S$.

Now, the first row of the wraparound Pascal's triangle is $\left(
1,0,\ldots,0\right)  $, and the second row is $\left(  1,1,0,\ldots,0\right)
$. \ This second row is not in $S$ unless $k=1$.
\end{proof}

\begin{proof}
[Alternate Proof of Theorem \ref{k2or4}]We will actually prove a stronger
result, that if $G$ is any nontrivial \textit{affine} gate that preserves
Hamming weight mod $k$, then either $k=2$\ or $k=4$. \ We have $G\left(
x\right)  =Ax\oplus b$, where $A$ is an $n\times n$ invertible matrix over
$\mathbb{F}_{2}$, and $b\in\mathbb{F}_{2}^{n}$. \ Since $G$ is nontrivial,
Lemma \ref{affcontriv}\ implies that\ at least one of $A$'s column vectors
$v_{1},\ldots,v_{n}$ must have Hamming weight at least $2$; assume without
loss of generality that $v_{1}$\ is such a column. \ Notice that $\left\vert
G\left(  0^{n}\right)  \right\vert \equiv\left\vert b\right\vert
\equiv0\left(  \operatorname{mod}k\right)  $, while%
\[
\left\vert G\left(  e_{1}\right)  \right\vert \equiv\left\vert v_{1}\oplus
b\right\vert \equiv1\left(  \operatorname{mod}k\right)
\]

Clearly $\left\vert G\left(  e_{1}\right)  \right\vert \equiv\left\vert
e_{1}\right\vert \equiv1\left(  \operatorname{mod}k\right)  $. \ Let $y$ be an
$n$-bit string whose first bit is $0$. \ Then by Lemma \ref{inclex}, we have%
\begin{align*}
1+\left\vert y\right\vert  &  \equiv\left\vert e_{1}\oplus y\right\vert \\
&  \equiv\left\vert G\left(  e_{1}\oplus y\right)  \right\vert \\
&  \equiv\left\vert G\left(  e_{1}\right)  \oplus G\left(  y\right)  \oplus
b\right\vert \\
&  \equiv\left\vert Ae_{1}\oplus b\oplus b\oplus G\left(  y\right)
\right\vert \\
&  \equiv\left\vert v_{1}\oplus G\left(  y\right)  \right\vert \\
&  \equiv\left\vert v_{1}\right\vert +\left\vert G\left(  y\right)
\right\vert -2\left(  v_{1}\cdot G\left(  y\right)  \right) \\
&  \equiv\left\vert v_{1}\right\vert +\left\vert y\right\vert -2\left(
v_{1}\cdot G\left(  y\right)  \right)  \left(  \operatorname{mod}k\right)  .
\end{align*}
Thus%
\[
2\left(  v_{1}\cdot G\left(  y\right)  \right)  \equiv\left\vert
v_{1}\right\vert -1\left(  \operatorname{mod}k\right)  .
\]
Note that the above equation must hold for all $2^{n-1}$\ possible $y$'s that
start with $0$. \ Such $y$'s, of course, account for half of all $n$-bit
strings. \ So we deduce that%
\[
\Pr_{x\in\left\{  0,1\right\}  ^{n}}\left[  2\left(  v_{1}\cdot x\right)
\equiv\left\vert v_{1}\right\vert -1\left(  \operatorname{mod}k\right)
\right]  \geq\frac{1}{2}.
\]
Equivalently, if we let $S$\ be the set of all $x\in\left\{  0,1\right\}
^{\left\vert v_{1}\right\vert }$ such that $2\left\vert x\right\vert
\equiv\left\vert v_{1}\right\vert -1\left(  \operatorname{mod}k\right)  $,
then we find that%
\begin{equation}
\Pr_{x\in\left\{  0,1\right\}  ^{\left\vert v_{1}\right\vert }}\left[  x\in
S\right]  \geq\frac{1}{2}, \label{prin}%
\end{equation}
or $\left\vert S\right\vert \geq2^{\left\vert v_{1}\right\vert -1}$. \ But we
will prove this impossible.

First suppose $k$ is even. \ Then for the inequality (\ref{prin})\ to have any
chance of being satisfied, $\left\vert v_{1}\right\vert $\ needs to be odd, so
assume it is. \ Then $S$ equals the set of all $x\in\left\{  0,1\right\}
^{\left\vert v_{1}\right\vert }$\ such that%
\begin{equation}
\left\vert x\right\vert \equiv0\left(  \operatorname{mod}\frac{k}{2}\right)  .
\label{prin2}%
\end{equation}
If $k=2$, then $\left\vert x\right\vert \equiv0\left(  \operatorname{mod}%
1\right)  $\ holds for all $x$, while if $k=4$, then $\left\vert x\right\vert
\equiv0\left(  \operatorname{mod}2\right)  $ holds whenever $\left\vert
x\right\vert $\ is even. \ In either case, (\ref{prin})\ is satisfied. \ On
the other hand, suppose $k\geq6$. \ Then we claim that (\ref{prin}) cannot
hold: in other words, that $\left\vert S\right\vert <2^{\left\vert
v_{1}\right\vert -1}$. \ To prove this, let%
\[
S^{\prime}:=\left\{  x\oplus e_{1}:x\in S\right\}
\]
contain, for each $x\in S$, the string $x^{\prime}$\ obtained by flipping the
first bit of $x$. \ Then clearly\ $\left\vert S\right\vert =\left\vert
S^{\prime}\right\vert $, and $S$ and $S^{\prime}$\ are disjoint (since no two
elements of $S$\ are neighbors in the Hamming cube). \ So it suffices to show
that $S\cup S^{\prime}$\ still does not cover all of $\left\{  0,1\right\}
^{\left\vert v_{1}\right\vert }$. \ Since $\frac{k}{2}\geq3$, observe that
$S^{\prime}$\ can contain at most one string of Hamming weight $1$, namely
$x^{\prime}=10\cdots0$\ (the neighbor of $x=0^{\left\vert v_{1}\right\vert }%
$). \ But since $\left\vert v_{1}\right\vert \geq2$, there are other strings
of Hamming weight $1$, not included in $S^{\prime}$. \ Hence $S\cup S^{\prime
}\neq\left\{  0,1\right\}  ^{\left\vert v_{1}\right\vert }$.

Next suppose $k\geq3$ is odd. Then first, we claim that we cannot have
$\left\vert v_{1}\right\vert =2$. \ For suppose we did. \ Then $\left\vert
b\oplus v_{1}\right\vert $ would be either $\left\vert b\right\vert $,\ or
$\left\vert b\right\vert -2$, or $\left\vert b\right\vert +2$. \ But this
contradicts the facts that $\left\vert b\right\vert \equiv0\left(
\operatorname{mod}k\right)  $, while $\left\vert b\oplus v_{1}\right\vert
\equiv1\left(  \operatorname{mod}k\right)  $. \ Since $\left\vert
v_{1}\right\vert \neq1$, this means that $\left\vert v_{1}\right\vert \geq3$.
\ But in that case, we can use a similar argument as before to show that
(\ref{prin}) cannot hold, and that $\left\vert S\right\vert <2^{\left\vert
v_{1}\right\vert -1}$. \ Letting $S^{\prime}$\ be as above, we again have that
$\left\vert S\right\vert =\left\vert S^{\prime}\right\vert $, and that $S$ and
$S^{\prime}$\ are disjoint. \ And we will again show that $S\cup S^{\prime}%
$\ fails to cover all of $\left\{  0,1\right\}  ^{\left\vert v_{1}\right\vert
}$. \ Notice that, since the Hamming weights of the $S$\ elements are
separated by $k\geq3$, every $S^{\prime}$ element that is \textquotedblleft
below\textquotedblright\ an $S$\ element must start with $0$, and every
$S^{\prime}$ element that is \textquotedblleft above\textquotedblright\ an
$S$\ element must start with $1$. \ Also, since $\left\vert v_{1}\right\vert
\geq3$, there must be some $x^{\prime}\in S^{\prime}$\ with a Hamming weight
that is neither maximal nor minimal (that is, neither $\left\vert
v_{1}\right\vert $\ nor $0$). \ But since the first bit of $x^{\prime}$\ has a
fixed value, not all strings of Hamming weight $\left\vert x^{\prime
}\right\vert $\ can belong to $S^{\prime}$. \ Hence $S\cup S^{\prime}%
\neq\left\{  0,1\right\}  ^{\left\vert v_{1}\right\vert }$, and $\left\vert
S\right\vert <2^{\left\vert v_{1}\right\vert -1}$.
\end{proof}

\end{document}